\documentclass[12pt]{article}

\input{draft_style.sty}

\begin{document}
\title{A Relaxation Approach to Synthetic Control\thanks{Liao: \protect\url{lchw@link.cuhk.edu.hk}. Shi (corresponding author): \protect\url{zhentao.shi@cuhk.edu.hk}. 9F Esther Lee Building, The Chinese University of Hong Kong, Shatin, New Territories, Hong Kong SAR, China. Zheng: \protect\url{yapengzheng@link.cuhk.edu.hk}. Shi acknowledges partial financial support from the Research Grants Council of Hong Kong (Project No. 14617423) and the National Natural Science Foundation of China (Project No. 72425007, 72133002).
The authors thank the Co-Editor and two anonymous referees for their constructive comments, which greatly improve this paper. We also thank Han Hong, Toru Kitagawa, Whitney Newey, and Kevin Song for their helpful comments.
}}
\author{Chengwang Liao\quad{}Zhentao Shi\quad{}Yapeng Zheng }
\date{}
\maketitle
\vspace{-1.5cm}
\begin{center}
\itshape\large
The Chinese University of Hong Kong
\end{center}
\bigskip
\begin{abstract}
\onehalfspacing

The synthetic control method (SCM) is widely used for constructing the counterfactual of a treated unit based on data from control units in a donor pool. Allowing the donor pool to contain more control units than time periods,
we propose \emph{SCM-relaxation}, an adaptation of the stable-balancing framework to the synthetic control setting. This approach minimizes an information-theoretic measure of the weights subject to a set of relaxed linear inequality constraints in addition to the simplex constraint. When the donor pool exhibits a group structure, SCM-relaxation approximates the equal weights within each group to diversify the prediction risk. Asymptotically, the proposed estimator achieves oracle performance in terms of out-of-sample prediction accuracy, and can be used to conduct inference about the average treatment effect on the treated. Our empirical application assesses the economic impact of Brexit on the United Kingdom's real GDP.

\bigskip


\end{abstract}

\newpage
\section{Introduction}

Policy evaluation is difficult because potential outcomes in the absence of treatment are unobservable in the treated group. When observational data are available, one prominent solution is the synthetic control method (SCM) \citep*{abadie2003economic,abadieSyntheticControlMethods2010}. A typical application involves a single treated unit of interest and a set of control units connected through latent common factors. SCM has been successfully applied to many economic studies, for example, \citet*{kleven2013taxation}, \citet*{bohn2014did}, \citet*{acemoglu2016value}, and \citet*{cunningham2018decriminalizing}. \citet*{athey2017state} hailed SCM as ``arguably the most important innovation in the policy evaluation literature in the last 15 years.''

A pivotal step in implementing SCM is assigning weights to the units in the donor pool to construct a “synthetic” version of the treated unit. These weights are determined by minimizing the discrepancy between the treated unit and the synthetic control in the pre-treatment period. The resulting synthetic unit then extrapolates the outcomes in the post-treatment period, with the estimated policy effect given by the difference between the observed outcome of the treated unit and the counterfactual prediction. SCM's weights are constrained to the simplex---that is, they must be non-negative and sum to one. This constraint often leads to sparse solutions, where a few individual units receive non-zero weights. While such sparsity enhances interpretability, it may not always be desirable. In many practical economic settings, there is no compelling reason to assume that only a small subset of variables are relevant to the target variable \citep*{giannone2021economic}. When researchers invest substantial effort in collecting data from a large pool of potential controls, they may prefer dense weighting schemes that make full use of the available information. A dense solution can also help diversify prediction risk, potentially improving the robustness of the counterfactual estimates \citep*{liao2024does}.

To encourage dense solutions, we employ a relaxation approach within the synthetic control framework, which we call \emph{SCM-relaxation}, for estimating the weights and predicting the counterfactual outcomes. It formulates weight estimation as a constrained optimization problem, where the objective function is an information-theoretic divergence measure of the weight vector. Staying true to the spirit of SCM, we retain the simplex constraint.

In high-dimensional settings, developing asymptotic theory requires structural assumptions that enable effective dimension reduction. The search for a convex combination in SCM parallels the classical econometric problem of forecast combination, making it natural to leverage the aggregated information from multiple control units. Recent advancements in the literature on latent groups in panel data---\citet*{bonhomme2015grouped}, \citet*{su2016identifying}, \citet*{vogt2017classification}, and \citet*{chetverikovSpectralPostspectralEstimators2022}---provide a useful framework to bridge the two extremes of full homogeneity (identical units) and full heterogeneity (distinct units). Under the latent group structures in the loadings of a factor model, this paper's \cref{thm:w_convergence} shows that our estimated weight vector converges sufficiently fast to a desirable target, even when the number of control units ($J$) is much larger than the number of pre-treatment time periods ($T_{0}$). Once the weight is estimated, the average treatment effect on the treated (ATT) can be obtained through the predicted counterfactual. \cref{thm:oracle_inequalities} establishes that the prediction risk of the counterfactuals is asymptotically equal to that under the oracle weights.

\phantomsection\label{para:stable-balancing-family}Our formulation belongs to a broad family of stable balancing weights and calibration estimators \citep*{deville1992calibration,zubizarreta2015stable,athey2018approximate,benmichael2021balancing}, which are obtained by minimizing a dispersion measure subject to constraints on covariate imbalance. In this paper, the imbalance is induced by a relaxed first-order condition (FOC) from SCM, and the dispersion may be chosen from a class of convex functions. \cref{thm:wg_convergence} develops asymptotic guarantees for this familiar class of balancing estimators in the high-dimensional panel setting with latent group structures.

We can learn the ATT based on SCM-relaxation. \cref{thm:att} shows that the ATT estimator is consistent. For inference, we apply the cross-fitting $t$-test of \citet*{chernozhukov_t-test_2024} to the SCM-relaxation approach. 
\cref{thm:asym-dist} establishes that the test statistic follows asymptotically a $t$-distribution, from which confidence intervals can be constructed.

We conduct Monte Carlo simulations to demonstrate the finite-sample performance of our proposal in comparison with other off-the-shelf machine learning methods. Our empirical example evaluates the impact of Brexit on the real GDP growth of the United Kingdom. We find that the estimated weights of the control units from our method are more interpretable than those from the vanilla SCM.

\medskip

We make two contributions in developing the asymptotic theory. First, we overcome a key technical challenge in handling many non-negativity constraints arising from the simplex constraint. The same issue appears in risk management for large portfolios \citep*{jagannathan2003risk}, but to the best of our knowledge, there is no theoretical investigation of the consequence. We provide a rigorous proof by accounting for the non-negativity constraints. In particular, \cref{rem:key-tech-remark} explains our strategy for pushing the estimated vector to its desirable target by establishing the compatibility condition \citep*{buhlmann2011statistics} that connects the $L_{2}$-norm of the high-dimensional weight vector with its low-dimensional counterpart. Second, while the group structure and the low-dimensional factor model are dimension-reduction assumptions, we allow both the number of groups ($K$) and the number of factors ($r$) to diverge to infinity, regardless of their relative order. This offers a flexible accommodation for approximate group features 
and factor proliferation, 
and it breaks the severe restriction $K\leq r$ in \citet*{shiL2RelaxationApplicationsForecast2022}.

\medskip

This paper builds on several strands of literature. We follow \citet*{abadieUsingSyntheticControls2021} by maintaining the simplex constraint, while taking advantage of the latent group structures. \citet*{doudchenkoBalancingRegressionDifferenceDifferences2016} develop a constrained regression framework for balancing, regression, difference-in-differences, and SCM. Ours serves as a complementary method when the donor pool exhibits latent groups. When the pre-treatment outcomes of the treated unit lie in the convex hull of the pre-treatment outcomes of the control units, the estimated weights induced by SCM may not be unique. This occurs particularly when $J$ is larger than $T_{0}$. \citet*{abadiePenalizedSyntheticControl2021} utilize a penalization method to guarantee a unique solution, where the penalty is imposed on the pairwise distance of covariates between units. \citet*{arkhangelsky2021synthetic} and \citet*{chernozhukov2021exact} propose different versions of penalized SCM, including $L_{1}$- (Lasso), $L_{2}$- (ridge), and the elastic net. Departing from penalty-based regularization, this paper takes the relaxation approach.

Our estimator is closely related to the literature of stable balancing. The general balancing framework traces back to calibration weighting in survey sampling \citep*{deville1992calibration} and has been developed for causal inference and high-dimensional estimation by
\citet*{zubizarreta2015stable},
\citet*{robbins2017framework},
\citet*{athey2018approximate}, \citet*{wang2020minimal},
\citet*{hirshberg2019minimax},
and \citet*{kallus2020generalized}, among others. \citet*{benmichael2021balancing} provide a review and discuss its connection with SCM.
\phantomsection\label{para:enrich-framework}While our $L_{\infty}$ imbalance constraint and a convex dispersion objective fit into this line of work, we enrich this framework by analyzing the estimator under the panel-data factor structure with a diverging donor pool and latent group loadings.


A variety of asymptotic frameworks have appeared to justify SCM. \citet*{abadieSyntheticControlMethods2010} derive bounds on the SCM estimator when perfect pre-treatment fit is met under fixed $J$ and diverging $T_{0}$. \citet*{ferman2021synthetic} and \citet*{fermanPropertiesSyntheticControl2021} analyze the asymptotic property of estimated weights under the imperfect pre-treatment fit assumption when $T_{0}$ goes to infinity, with $J$ held fixed and with $J$ diverging, respectively. A key message is that when $J$ is fixed, the SCM estimator will be biased even when $T_{0}$ goes to infinity. If $J$ diverges, the estimator can be unbiased as weights spread out across the units. In a panel of two-way fixed effects, \citet*{arkhangelsky2023large} show consistency and asymptotic normality of SCM where the treatment is selected on permanent and time-varying components. \citet*{zhangAsymptoticPropertiesSynthetic2022} establish an optimality property that the SCM weight estimator asymptotically minimizes the post-treatment mean squared error of the estimated treatment effects.
In a high-dimensional asymptotic framework allowing $J > T_0$, we provide a toolkit for applied researchers by establishing consistency of weight estimation and risk minimization, as well as statistical inference about ATT.

\textbf{Organization. }The rest of the paper is organized as follows. \cref{sec:Model-and-estimation} sets up the data generating process (DGP) and motivates the SCM-relaxation method. \cref{sec:Asymptotic-theory} then establishes the asymptotic theory of the SCM-relaxation estimators under group structures of factor loadings. \cref{sec:Inference-on-ATT} studies estimation and inference of ATT based on SCM-relaxation. We conduct Monte Carlo experiments in \cref{sec:Monte-Carlo-experiments}, and in \cref{sec:Empirical-applications} we apply our method to evaluate the Brexit shock. Proofs of all theoretical results are relegated to the online appendix.

\section{Setup }\label{sec:Model-and-estimation}

We first define the notation.
For a positive integer $N$, we use $[N]$ as shorthand for $\{1,2,\dots,N\}$. Let $a\wedge b:=\min(a,b)$.
An \emph{absolute constant} is a positive finite number independent of the sample size. Vectors and matrices are in boldface. For a generic $n\times1$ vector $\bm{a}$, let $\left\Vert \bm{a}\right\Vert _{p}=(\sum_{i\in[n]}|a_{i}|^{p})^{1/p}$ and $\left\Vert \bm{a}\right\Vert _{\infty}=\max_{i\in[n]}|a_{i}|$ be its $L_{p}$- and $L_{\infty}$-norm, respectively. We denote $\bm{1}_{n}$ and $\bm{0}_{n}$ as the $n\times1$ vector of ones and zeros, respectively. The $n$-dimensional simplex is $\Delta_{n}=\{\bm{a}\in\mathbb{R}^{n}\colon\min_{i\in[n]}a_{i}\geq0,\bm{a}'\bm{1}_{n}=1\}$. For an $m\times n$ matrix $\bm{A}$, let $\varsigma_{\max}(\bm{A})$ be its largest singular value and $\varsigma_{\min}(\bm{A})$ its minimum \emph{nonzero} singular value. Let $\|\bm{A}\|_{2}=\varsigma_{\max}(\bm{A})$ be the spectral norm. For a symmetric matrix $\bm{B}$, let $\phi_{\min}(\bm{B})$ and $\phi_{\max}(\bm{B})$ be its minimum and maximum eigenvalues, respectively. The symbols ``$\to_{p}$'' and ``$\to_{d}$'' signify \emph{convergence in probability} and \emph{convergence in distribution}, respectively.

\subsection{Model}

A researcher has data on $J+1$ cross-sectional units and $T_{0}+T_{1}$ time periods, where $T_{0}$ is the number of pre-treatment periods, and $T_{1}$ the number of post-treatment periods. Let $\mathcal{T}_{0}:=[T_{0}]$ denote the index set of the pre-treatment periods, $\mathcal{T}_{1}:=\{T_{0}+1,\dots,T_{0}+T_{1}\}$ denote that of the post-treatment periods, and $\mathcal{T}:=\mathcal{T}_{0}\cup\mathcal{T}_{1}$ with $T:=T_{0}+T_{1}$. Unit $j=0$ is the sole individual that receives treatment during $t\in\mathcal{T}_{1}$. The remaining $J$ units, indexed by $j\in[J]$, form the donor pool of potential controls. In the potential outcome framework, let $y_{jt}^{I}$ be the outcome of unit $j$ being treated at time $t$, and $y_{jt}^{N}$ if untreated.\footnote{Following the literature, the superscripts ``$N$'' and ``$I$'' stand for ``non-intervened'' (interchangeable with ``untreated'') and ``intervened'' (interchangeable with ``treated''), respectively.} In reality, we observe $y_{jt}=d_{jt}y_{jt}^{I}+(1-d_{jt})y_{jt}^{N}$, where $d_{jt}=1$ for treatment and $d_{jt}=0$ otherwise. Since only $j=0$ is treated when $t\in\mathcal{T}_{1}$, the treatment indicator is $d_{jt}=1\left\{ j=0,t\in\mathcal{T}_{1}\right\} $, where $1\left\{ \cdot\right\} $ is the indicator function. The treatment effect at time $t\in\mathcal{T}_{1}$ is
$\tau_{0t}=y_{0t}^{I}-y_{0t}^{N}.$
Since $y_{0t}^{I}$ is observable while $y_{0t}^{N}$ is unobserved in the post-treatment periods, the counterfactual $y_{0t}^{N}$ must be estimated from the data. The ATT over the post-treatment periods is
$
\tau_0:=T_1^{-1}\sum_{t\in\mathcal T_1}\tau_{0t}.
$
The ATT $\tau_{0}$ is the main target of estimation in this paper.

This paper presents a simple form of SCM with the outcome variable only, as in \citet*{ferman2021synthetic}, \citet*{fermanPropertiesSyntheticControl2021}, and \citet*{chenSyntheticControlOnline2023}. SCM takes a linear combination $\sum_{j\in[J]}w_{j}y_{jt}$ of the control units' outcomes, where the weights $w_{j}$ lie in the $J$-dimensional simplex $\Delta_{J}$---they are non-negative and sum to one. The synthetic control estimator for $\bm{w}=(w_{1},\dots,w_{J})'$ is obtained by minimizing the $L_{2}$-distance between the treated unit and synthetic control in the pre-treatment periods, i.e.,
\begin{equation}
\hat{\bm{w}}^{\mathrm{SC}}:=\arg\min_{\bm{w}\in\Delta_{J}}\left\Vert \bm{y}_{0}-\bm{Y}\bm{w}\right\Vert _{2}^{2},\label{eq:SCM}
\end{equation}
where $\bm{y}_{j}=(y_{j1},\dots,y_{jT_{0}})'$ and $\bm{Y}=(\bm{y}_{1},\dots,\bm{y}_{J})$.\footnote{When additional covariates are available, the synthetic control can incorporate them into the quadratic form. \citet*{abadie2003economic} minimize $(\bm{x}_{0}-\bm{X}\bm{w})'\bm{V}(\bm{x}_{0}-\bm{X}\bm{w})$, where $\bm{X}=(\bm{x}_{1},\cdots,\bm{x}_{J})$ consists of individual characteristics $\bm{x}_{j}$ which can include $\boldsymbol{y}_{j}$, and $\bm{V}$ is a symmetric and positive semi-definite matrix of conformable dimension.} When no confusion arises, we suppress the superscript $N$ for the pre-treatment potential outcomes $y_{jt}^{N}$ for $t\in\mathcal{T}_{0}$ since they are observable.

\citet*{abadieSyntheticControlMethods2010} motivate \eqref{eq:SCM} by a factor model. Suppose that the untreated potential outcomes $y_{jt}^{N}$ of each individual are generated from the linear factor model
\begin{equation}\label{eq:model}
y_{jt}^{N} =\bm{\lambda}_{j}'\bm{f}_{t}+u_{jt}\quad\text{for \ensuremath{j\in\{0\}\cup[J]} and \ensuremath{t\in\mathcal{T}}},
\end{equation}
where $u_{jt}$ is the idiosyncratic error, $\bm{f}_{t}$ is an $r\times1$ vector of latent common factors at time $t$, and it is multiplied by an $r\times1$ vector $\bm{\lambda}_{j}$ of factor loadings for individual $j$.\footnote{The factor model \eqref{eq:model} can incorporate individual- and time-specific fixed effects by setting $\bm{\lambda}_{j}=(1,\alpha_{j},\tilde{\bm{\lambda}}_{j}')'$ and $\bm{f}_{t}=(\gamma_{t},1,\tilde{\bm{f}}_{t}')'$.} Throughout this paper, we view the factors $\bm{f}_{t}$ as random whereas factor loadings $\bm{\lambda}_{j}$ as deterministic. The common factors shared by the treated unit and the untreated ones generate correlations to be utilized for counterfactual prediction. \citet*{abadieUsingSyntheticControls2021} stresses the importance of the simplex constraint. It ensures that all members in the donor pool make non-negative contributions as a weighted average, and in practice the solution $\hat{\bm{w}}^{\mathrm{SC}}$ is often sparse with many zeros.

In practical observational studies with panel data, the researcher may have many potential control units in the donor pool. In the meantime, the time dimension, mostly collected in low frequency, is usually not too long. \citet*{abadieSyntheticControlMethods2010}'s example of the California tobacco control program has $(J,T_{0})=(38,19)$, and \citet*{abadie2015comparative}'s example of German reunification has $(J,T_{0})=(16,30)$. The number of control units $J$ is often comparable to, or even exceeds, the number of pre-treatment periods $T_{0}$.

\subsection{SCM-relaxation}\label{subsec:SCM-relaxation}

To prevent in-sample over-fitting when $J$ is close to or exceeding $T_{0}$, we estimate the weights by solving the following convex optimization problem
\begin{equation}
\min_{\bm{w}\in\Delta_{J}}\left\Vert \bm{w}\right\Vert _{2}^{2}\quad\text{s.t.}\ \Vert\hat{\bm{\Sigma}}\bm{w}-\hat{\bm{\Upsilon}}\Vert_{\infty}\leq\eta,\label{eq:SCM-relaxation}
\end{equation}
where $\hat{\bm{\Sigma}}=T_0^{-1}\bm{Y}'\bm{Y}$, $\hat{\bm\Upsilon}=T_0^{-1}\bm{Y}'\bm{y}_0$, and $\eta=\eta_{T_{0},J}$ is a user-specified tuning parameter. The objective $\left\Vert \bm{w}\right\Vert _{2}^{2}$ encourages diversified weights, and the $L_{\infty}$ constraint relaxes the first-order condition $\hat{\bm{\Sigma}}\bm{w}-\hat{\bm{\Upsilon}}=\bm{0}_{J}$ implied by \eqref{eq:SCM} without the simplex constraint. We call \eqref{eq:SCM-relaxation} the \emph{$L_{2}$-SCM-relaxation} problem.

The objective function $\left\Vert \bm{w}\right\Vert _{2}^{2}$ in \eqref{eq:SCM-relaxation} is strictly convex in $\bm{w}$, and the feasible set \begin{equation} S_{\eta}:=\left\{ \bm{w}\in\Delta_{J}:\Vert\hat{\bm{\Sigma}}\bm{w}-\hat{\bm{\Upsilon}}\Vert_{\infty}\leq\eta\right\} \label{eq:feasible-set} \end{equation} is the intersection of the simplex $\Delta_{J}$ with the $L_{\infty}$-ball $\{\bm{w}:\Vert\hat{\bm{\Sigma}}\bm{w}-\hat{\bm{\Upsilon}}\Vert_{\infty}\leq\eta\}$, hence convex. If $S_{\eta}\neq\emptyset$, the solution to \eqref{eq:SCM-relaxation} is unique.

Problem \eqref{eq:SCM-relaxation} belongs to the class of stable balancing weight estimators \citep*{zubizarreta2015stable,athey2018approximate}, which minimize a dispersion measure subject to constraints on the imbalance between the treated and control units.
\[
\min_{\bm{w}}\;\text{Dispersion}(\bm{w})\quad\text{s.t.}\quad\text{Imbalance}(\bm{w},\bm{Y},\bm{y}_{0})\leq\eta.
\]
This line of research originated from the calibration weighting in survey sampling \citep*{deville1992calibration}. The imbalance can be expressed as the worst-case difference in means over a function class $\mathcal{F}$ \citep*{kallus2020generalized,hirshberg2019minimax}:
\[
\text{Imbalance}_{\mathcal{F}}(\bm{w},\bm{Y},\bm{y}_{0})
:=\sup_{f\in\mathcal{F}}{}\Biggl|\sum_{j=1}^{J}w_{j}f(\bm{y}_{j})-f(\bm{y}_{0})\Biggr|.
\]

\begingroup
\makeatletter
\let\remarkstylewithoutfootnote\th@remark
\def\th@remark{\remarkstylewithoutfootnote\thm@headpunct{.\footnotemark}}
\makeatother
\begin{rem}\label{rem:balancing-framework}
\leavevmode\footnotetext{We thank one of the referees for providing this perspective and interpretation.}\ignorespaces
For a linear function class $\mathcal{F}=\{f(\bm{x})=\bm{x}'\bm{\beta}\}$ with bounded norm, the imbalance reduces to familiar dual norms \citep[see][for a generic $\mathcal{F}$ in econometrics]{kong2025asymptotics}.
According to  \citet*{ben2021augmented},
with $\mathcal{F}=\{f(\bm{x})=\bm{x}'\bm{\beta}:\|\bm{\beta}\|_{2}\leq1\}$ we obtain the $L_{2}$ imbalance $\|\bm{y}_{0}-\bm{Y}\bm{w}\|_{2}$ used in standard SCM, while $\mathcal{F}=\{f(\bm{x})=\bm{x}'\bm{\beta}:\|\bm{\beta}\|_{1}\leq1\}$ gives the $L_{\infty}$ imbalance $\|\bm{y}_{0}-\bm{Y}\bm{w}\|_{\infty}$ used in stable balancing. Our constraint $\|\hat{\bm{\Sigma}}\bm{w}-\hat{\bm{\Upsilon}}\|_{\infty}\leq\eta$ is equivalent to $\|\bm{Y}^{\prime}(\bm{y}_{0}-\bm{Y}\bm{w})\|_{\infty}\leq\eta T_{0}$, which is the semi-norm $\|\bm{y}_{0}-\bm{Y}\bm{w}\|_{\bm{Y}^{\prime},\infty}$ corresponding to $\mathcal{F}=\{f(\bm{x})=\bm{x}'\bm{\beta}:\bm{\beta}\in\{\bm{Y}\bm{\gamma}:\|\bm{\gamma}\|_{1}\leq1\}\}$. When $\bm{Y}$ has full row rank, this semi-norm is non-degenerate on the column space generated by the pre-treatment outcomes, which is different from the usual coordinate-wise $L_{\infty}$ imbalance.

\phantomsection\label{para:dual-characterization}
A distinctive feature of this framework is its dual characterization: under regularity conditions, the optimal weights admit the representation $\bm{w}^{*}=g(\bm{Y}^{\prime}\bm{\beta}^{*})$ through a link function $g$ determined by the dispersion \citep*{benmichael2021balancing}. For the $L_{2}$ dispersion without non-negativity, $g$ is the identity link, so the weights are linear in $\bm{Y}$. With non-negativity imposed in the dual characterization, $g$ becomes the positive-part truncation $(x)_{+}$. For the entropy dispersion, $g$ is the exponential link \citep*{hainmuellerEntropyBalancingCausal2012}. Our theoretical contribution is to provide new asymptotic guarantees for this family of estimators in the panel data setting with latent group structures, where $J$ may diverge with $T_{0}$.
\end{rem}
\endgroup

Let $\hat{\bm{w}}$ be the solution to \eqref{eq:SCM-relaxation} if $S_{\eta}\neq\emptyset$, which holds if $\eta$ is sufficiently large. On the one hand, if $\eta\geq\Vert\hat{\bm{\Sigma}}\boldsymbol{1}_{J}/J-\hat{\bm{\Upsilon}}\Vert_{\infty}$, then $\hat{\bm{w}}=J^{-1}\boldsymbol{1}_{J}$ is feasible and the objective function reaches the lower bound.
On the other hand, if $S_{\eta=0}$ is nonempty, then \eqref{eq:SCM-relaxation} entails the exact equality $\hat{\bm{\Sigma}}\bm{w}=\hat{\bm{\Upsilon}}$. With a proper choice of $\eta\in[0,\Vert\hat{\bm{\Sigma}}\boldsymbol{1}_{J}/J-\hat{\bm{\Upsilon}}\Vert_{\infty}]$, the relaxation scheme balances the bias and variance and mitigates the sensitivity of SCM weights to the sampling noise in $(\hat{\bm{\Sigma}},\hat{\bm{\Upsilon}})$, effectively preventing in-sample over-fitting.

\begin{rem}
While $L_{2}$-SCM-relaxation adheres to the simplex constraint in SCM, the literature has witnessed many alternative regularization methods. For example, \citet[Section 2.3.3]{chernozhukov2021exact} estimate the weights by
\[
\min_{\bm{w}\in\mathbb{R}^{J}}\frac{1}{2}\|\bm{y}_{0}-\bm{Y}\bm{w}\|_{2}^{2}+\mathcal{P}(\bm{w}),
\]
where $\mathcal{P}(\bm{w})$ is a penalty on $\bm{w}$, and the ridge penalty $\mathcal{P}(\bm{w})=\eta\|\bm{w}\|_{2}^{2}$ is a popular choice. 
We will present in \cref{sec:Monte-Carlo-experiments} numerical results of penalized synthetic control methods, including Lasso and ridge, among other competitive methods.
\end{rem}

Next, we explore the asymptotic properties of the relaxation scheme.

\section{Asymptotic Theory }\label{sec:Asymptotic-theory}

Statistical analysis of high-dimensional problems typically postulates certain structures on the DGP for dimension reduction. For example, variable selection methods such as Lasso \citep*{tibshirani1996regression} and SCAD \citep*{fan2001variable} are suitable for linear regressions with sparse coefficients, meaning most of the parameters are either exactly zero or approximately zero. Similarly, in large covariance matrix estimation, various structures have been proposed: \citet*{bickel2008covariance} impose sparse covariance entries, and \citet*{tong2025cluster} assume a block correlation matrix.

The group pattern in panel data has been widely adopted as a middle ground between homogeneous coefficients and fully heterogeneous coefficients. Were the factors $\boldsymbol{f}_{t}$ observable, the factor loadings would play the role of the slope coefficients. Based on this resemblance, we assume latent groups across the factor loadings of the control units.
The donor pool is partitioned into $K$ disjoint groups $\mathcal{G}_{1},\dots,\mathcal{G}_{K}$, where $\mathcal{G}_{k}\cap\mathcal{G}_{\ell}=\emptyset$ for any $k\neq\ell$ and $\cup_{k=1}^{K}\mathcal{G}_{k}=[J]$. The factor loadings are homogeneous within each group, that is, for each $k\in[K]$, for any $i,j\in\mathcal{G}_{k}$, $\bm{\lambda}_{i}=\bm{\lambda}_{j}=\bm{\lambda}_{\mathcal{G}_{k}}$. To encode the group identities, we introduce a $J\times K$ membership matrix $\bm{Z}$, where $\bm{Z}_{jk}=1\{j\in\mathcal{G}_{k}\}$. Then the group structure can be written as
\begin{equation}
\underset{(J\times r)}{\bm{\Lambda}}=\underset{(J\times K)}{\bm{Z}}\underset{(K\times r)}{\bm{\Lambda}^{\mathrm{co}}},\label{eq:grouped_structure}
\end{equation}
where $\bm{\Lambda}=(\bm{\lambda}_{1},\dots,\bm{\lambda}_{J})'$ is the matrix of factor loadings and $\bm{\Lambda}^{\mathrm{co}}=(\bm{\lambda}_{1}^{\mathrm{co}},\dots,\bm{\lambda}_{K}^{\mathrm{co}})^{\prime}$ collects all group-specific loadings. Here the superscript ``co'' stands for ``core.'' Let $J_{k}:=|\mathcal{G}_{k}|$ be the number of controls in the $k$-th group, and then $J=\sum_{k\in[K]}J_{k}$.

\begin{rem}\label{rem:group-structure-remark}
The group structure is common in real data. For instance, group structures are ubiquitous in microeconometrics \citep*{abadie2023should}, time series forecast combination \citep*{chan2018some}, and finance \citep*{engle2012dynamic}. While each of these papers explicitly specifies a prior criterion for group allocation, no knowledge about the membership is required to implement $L_{2}$-SCM-relaxation. Block equicorrelation is taken as a latent structure in this paper, and we do not intend to recover the group membership.
\end{rem}

\begin{rem}
The group structure on the factor loadings is not necessary for consistent estimation of the ATT.
In \cref{sec:without-group} of the Online Appendices, we show that the ATT estimator, under some regularity conditions, remains consistent when there exist weights that reproduce the treated unit's factor loadings and are sufficiently spread across control units, even when the donor loadings do not admit a group structure.
The group structure, nonetheless, is a primitive and interpretable condition for the spread of the weights. It yields a provable convergence rate to the weight vector, and it further delivers a faster rate of convergence to ATT.

\end{rem}

\subsection{Conditions}\label{subsec:conditions}

While the literature mostly assumes finite groups and finite factors, here we allow $K$ and $r$ to be either fixed or diverging, which accommodates ``many groups'' and ``many factors.'' In asymptotic statements, we will explicitly send $T_{0}\to\infty$, whereas $K$, $r$, $J$, and $T_{1}$ are viewed as deterministic functions of $T_{0}$. We impose the following regularity conditions.

\begin{assumption}[Factors and loadings]\label{assu:DGP}
When $T_{0}$ is sufficiently large, there exist two absolute constants $\underline{c}$ and\/ $\bar{c}$ such that:
\begin{enumerate}[label=(\alph*),parsep=0pt,topsep=0pt,font=\upshape]
\item \label{enu:factors}Factors: $\hat{\bm{\Omega}}_{\bm{F}}:=T_{0}^{-1}\bm{F}^{\prime}\bm{F}$, where $\bm{F}=(\bm{f}_{1},\dots,\bm{f}_{T_{0}})^{\prime}$, is invertible, and there exists an $r\times r$ deterministic positive definite matrix $\bm{\Omega}_{\bm{F}}$ with $\underline{c}\leq\phi_{\min}(\bm{\Omega}_{\bm{F}})\leq\phi_{\max}(\bm{\Omega}_{\bm{F}})\leq\bar{c}$ such that $\|\hat{\bm{\Omega}}_{\bm{F}}-\bm{\Omega}_{\bm{F}}\|_{2}\to_{p}0$ as $T_{0}\to\infty$.
\item \label{enu:Loadings}Loadings: $\|\bm{\lambda}_{0}\|_{2}+\sup_{k\in[K]}\|\bm{\lambda}_{k}^{\mathrm{co}}\|_{2}\leq\bar{c}$, $\mathrm{rank}(\bm{\Lambda}^{\mathrm{co}})=K\wedge r$, and $\underline{c}K/(K\wedge r)\leq\varsigma_{\min}^{2}(\bm{\Lambda}^{\mathrm{co}})\leq\varsigma_{\max}^{2}(\bm{\Lambda}^{\mathrm{co}})\leq\bar{c}K/(K\wedge r)$. Moreover, if $K>r$ and $\bm{1}_{K}\notin\operatorname{col}(\bm{\Lambda}^{\mathrm{co}})$, then
\[
\phi_{\min}\left(\sum_{k=1}^{K}\frac{J_{k}}{J}(\bm{\lambda}_{k}^{\mathrm{co}}-\bar{\bm{\lambda}})(\bm{\lambda}_{k}^{\mathrm{co}}-\bar{\bm{\lambda}})'\right)\geq\frac{\underline{c}}{r},
\]
where $\bar{\bm{\lambda}}:=\sum_{k=1}^{K}(J_{k}/J)\bm{\lambda}_{k}^{\mathrm{co}}$.
\end{enumerate}
\end{assumption}

Due to their multiplicative form, the factors and loadings cannot be uniquely identified. \cref{assu:DGP} requires the existence of desirable factors and loadings. Part~\ref{enu:factors} is standard in factor models \citep{bai2002determining,bai2003inferential}. The invertibility condition and the boundedness of eigenvalues ensure that the factors are non-degenerate. We require $\hat{\bm{\Omega}}_{\bm{F}}$ to be invertible itself, not just to have an invertible limit, because the expression of the oracle weight (introduced later) implicitly relies on a nonsingular $\hat{\bm{\Omega}}_{\bm{F}}$. In an asymptotic sense, an invertible limit would suffice. Part~\ref{enu:Loadings} requires bounded loadings $\boldsymbol{\lambda}_{0}$ for the treated unit. Regarding the control units, a sufficient condition for bounded $\|\bm{\lambda}_{k}^{\mathrm{co}}\|_{2}$ is that the maximum loading is of order $O(1/\sqrt{r})$, as in \citet{liDeterminingNumberFactors2017}, to accommodate a diverging number of factors. The rank condition implies that $\bm{\Lambda}^{\mathrm{co}}$ has $K\wedge r$ nonzero singular values, and thus $\sum_{k\in[K\wedge r]}\varsigma_{k}^{2}(\bm{\Lambda}^{\mathrm{co}})=\mathrm{trace}(\bm{\Lambda}^{\mathrm{co}\prime}\bm{\Lambda}^{\mathrm{co}})=\sum_{k\in[K]}\|\bm{\lambda}_{k}^{\mathrm{co}}\|_{2}^{2}=O(K)$. The assumption implies all singular values are of the same order.\footnote{Our theory can accommodate various degrees of factor strength; for example, some singular values can be of weaker order $O(K^{\alpha}/(K\wedge r))$ for some $\alpha\in(0,1)$. Weak factors would slow down the convergence rates and complicate notation, but they do not materially affect our theoretical conclusions. To facilitate exposition, we assume all factors are strong.} The last condition in Part~\ref{enu:Loadings} requires sufficient variation in the factor loadings after removing their mean across control units. The full-rank condition in Part~\ref{enu:Loadings} alone does not imply this bound, because centering can substantially reduce a singular value.

\begin{assumption}[Convex hull]\label{assu:convex_hull}
There exists $\bm{\pi}\in\Delta_{K}$ such that $\bm{\Lambda}^{\mathrm{co}\prime}\bm{\pi}=\bm{\lambda}_{0}$.
\end{assumption}

\cref{assu:convex_hull} imposes that the treated unit's factor loading $\bm{\lambda}_{0}$ lies in the convex hull of the group-level loadings $\{\bm{\lambda}_{k}^{\mathrm{co}}\}_{k\in[K]}$. Under the group structure \eqref{eq:grouped_structure}, this implies $\bm{\Lambda}'\bm{w} = \bm{\lambda}_0$ with $\bm{w}=\bm{Z}(\bm{Z}'\bm{Z})^{-1}\bm{\pi}$, meaning that the treated unit's loading can be recovered as a convex combination of the control units' loadings. The perfect recovery of the treated unit's loading plays an important role when perfect fit of the treated unit's outcome is unattainable. That said, this condition can be relaxed to allow for asymptotic recovery of the treated unit's loading (i.e., $\|\bm{\Lambda}^{\mathrm{co}\prime}\bm{\pi}-\bm{\lambda}_{0}\|_{2}\to0$) as in \citet{fermanPropertiesSyntheticControl2021} and \citet{fryMethodMomentsApproach2024}; a recent work by \cite{liu2025synthetic} further works with parallel trends and studies partial identification.

With \Cref{assu:convex_hull}, we specify the oracle weight $\bm{w}^{*}$ as an asymptotic target for the $L_{2}$-SCM-relaxation estimator $\hat{\bm{w}}$. To describe it, let us introduce some notation. By the latent group structure \eqref{eq:grouped_structure}, the sample covariance matrix for $\bm{Y}$ can be decomposed as $\hat{\bm{\Sigma}}=\hat{\bm{\Sigma}}^{*}+\hat{\bm{\Sigma}}^{\mathrm{e}}$, where
\[
\hat{\bm{\Sigma}}^{*}:=\frac{1}{T_{0}}\bm{Z}\bm{\Lambda}^{\mathrm{co}}\bm{F}'\bm{F}\bm{\Lambda}^{\mathrm{co}}{}'\bm{Z}',
\]
and thus $\hat{\bm{\Sigma}}^{\mathrm{e}}:=(\bm{Z}\bm{\Lambda}^{\mathrm{co}}\bm{F}'\bm{U}+\bm{U}'\bm{F}\bm{\Lambda}^{\mathrm{co}}{}'\bm{Z}'+\bm{U}^{\prime}\bm{U})/T_{0}$, with $\bm{U}=(\bm{u}_{1},\dots,\bm{u}_{J})$ and $\bm{u}_{j}=(u_{j1},\dots,u_{jT_{0}})'$ for $j\in[J]$. Here, the superscript ``$*$'' denotes the covariance component generated solely by the signal part $\bm{\Lambda}\bm{f}_{t}$ of the outcomes, while the superscript ``$\mathrm{e}$'' collects the remaining terms involving idiosyncratic errors. Similarly, the covariance between $\bm{Y}$ and $\bm{y}_{0}$ can be written as $\hat{\bm{\Upsilon}}=\hat{\bm{\Upsilon}}^{*}+\hat{\bm{\Upsilon}}^{\mathrm{e}}$, where
\[
\hat{\bm{\Upsilon}}^{*}:=\frac{1}{T_{0}}\bm{Z}\bm{\Lambda}^{\mathrm{co}}\bm{F}'\bm{F}\bm{\lambda}_{0},
\]
and $\hat{\bm{\Upsilon}}^{\mathrm{e}}:=T_0^{-1}(\bm{Z}\bm{\Lambda}^{\mathrm{co}}\bm{F}'\bm{u}_{0}+\bm{U}^{\prime}\bm{F}\bm{\lambda}_{0}+\bm{U}^{\prime}\bm{u}_{0})$ with $\bm{u}_{0}=(u_{01},\dots,u_{0T_{0}})^{\prime}$.

Consider the oracle problem of $L_{2}$-SCM-relaxation with sample covariance matrices of the infeasible signal that are free of the idiosyncratic errors:
\begin{equation}
\min_{\bm{w}\in\Delta_{J}}\left\Vert \bm{w}\right\Vert _{2}^{2}\quad \text{s.t. }\hat{\bm{\Sigma}}^{*}\bm{w}-\hat{\bm{\Upsilon}}^{*}=\bm{0}_{J}.\label{eq:oracle_relaxation_problem}
\end{equation}
Denote the solution to the above problem as $\bm{w}^{*}$, and let $S^{*}:=\{\bm{w}\in\Delta_{J}\colon\hat{\bm{\Sigma}}^{*}\bm{w}-\hat{\bm{\Upsilon}}^{*}=\bm{0}_{J}\}$ denote its feasible set.
As the constraint can be equivalently written as $\bm{\Lambda}\hat{\bm{\Omega}}_{\bm{F}}(\bm{\Lambda}'\bm{w}-\bm{\lambda}_{0})=\bm{0}_J$, it is feasible under \cref{assu:convex_hull}.


\begin{rem}\label{rem:oracle-feasibility}
The constraint $\hat{\bm{\Sigma}}^{*}\bm{w}-\hat{\bm{\Upsilon}}^{*}=\bm{0}_{J}$ is a moment matching condition that matches the signal component of the treated unit's outcome with that of the control units. The exact match is convenient for defining the oracle weight but may be restrictive. It can be relaxed in two ways. In an early version of this paper \citep*{liao2026relaxation}, we introduce an unrestricted scalar shift $\gamma$ to relax this exact match as $\hat{\bm{\Sigma}}^{*}\bm{w}-\hat{\bm{\Upsilon}}^{*}+\gamma\bm{1}_{J}=\bm{0}_{J}$. Another approach is to exploit the asymptotic framework. For example, the equality constraint may be replaced with $\|\hat{\bm{\Sigma}}^{*}\bm{w}-\hat{\bm{\Upsilon}}^{*}\|_{2}\leq\delta_J$ where $\delta_J\to0$ as $J\to\infty$. Either way,   feasibility is satisfied and the oracle weight is well-defined. It is important to specify the oracle target so that convergence of the estimated weights can be established as in \Cref{thm:w_convergence}.

\end{rem}

The following lemma characterizes the oracle solution $\bm{w}^{*}$.

\begin{lem}[Oracle weights]
\label{lem:oracle_target_relaxation} If \cref{assu:DGP,assu:convex_hull} hold, then \eqref{eq:oracle_relaxation_problem} has a unique solution, which must be within-group equal, i.e., $w_{i}^{*}=w_{j}^{*}$ if $i$ and $j$ belong to the same group. Furthermore, if the solution is in the interior of the simplex, i.e., $w_{j}^{*}>0$ for all $j\in[J]$, then it has the following closed form
\[
\bm{w}^{*}=\bm{Z}(\bm{Z}'\bm{Z})^{-1}\bm{w}_{\mathcal{G}}^{*},
\]
where the expression of $\bm{w}_{\mathcal{G}}^{*}=(w_{\mathcal{G}_{1}}^{*},\dots,w_{\mathcal{G}_{K}}^{*})'$ is given by \eqref{eq:w_G_oracle} for $K\leq r$, and \eqref{eq:wg_oracle_r_less_K_in_col} or \eqref{eq:wg_oracle_r_less_K_not_in_col} for $K>r$ in \cref{subsec:Proof-of-Lemma}.
\end{lem}

\cref{lem:oracle_target_relaxation} demonstrates that for $j\in\mathcal{G}_{k}$, the oracle weights $w_{j}^{*}=w_{\mathcal{G}_{k}}^{*} / J_{k}$ are evenly distributed within group $\mathcal{G}_{k}$. We allow the group weight $w_{\mathcal{G}_{k}}^{*}=0$ for some $k$. Such a group does not contribute to the prediction of $\bm{y}_{0}$. 
Next, we will show that $\hat{\bm{w}}$ and $\bm{w}^{*}$ are sufficiently close in both $L_{1}$- and $L_{2}$-norms under regularity conditions. Define $\phi_{J,T_{0}}:=\sqrt{(\log J)/T_{0}}$, which will serve as an upper bound of sampling errors and average cross-sectional dependence in the following assumption. Note that $\phi_{J,T_{0}}$ can shrink to 0 even if $J$ grows much faster than $T_{0}$.

\begin{assumption}\label{assu:errors}
The factors and idiosyncratic errors satisfy:
\begin{enumerate}[label=(\alph*),parsep=0pt,topsep=0pt,font=\upshape]
\item \label{enu:uncorr}$\mathbb{E}(u_{jt})=0$, and $(u_{jt})_{j\in\{0\}\cup[J]}$ and $\bm{f}_{t}$ are uncorrelated for all $t\in\mathcal{T}_{0}$.
\item \label{enu:Fu}$\sup_{j\in\{0\}\cup[J]}\|T_{0}^{-1}\bm{F}'\bm{u}_{j}\|_{\infty}=O_{p}(\phi_{J,T_{0}})$.
\item \label{enu:uu}$\sup_{j\in[J]}J^{-1}\sum_{i=1}^{J}|\mathbb{E}(T_{0}^{-1}\bm{u}_{i}^{\prime}\bm{u}_{j})|+\sup_{j\in[J]}|\mathbb{E}(T_{0}^{-1}\bm{u}_{0}^{\prime}\bm{u}_{j})|=O(\phi_{J,T_{0}})$.
\item \label{enu:uu_dev}$\sup_{i,j\in\{0\}\cup[J]}\bigl|T_{0}^{-1}[\bm{u}_{i}^{\prime}\bm{u}_{j}-\mathbb{E}(\bm{u}_{i}^{\prime}\bm{u}_{j})]\bigr|=O_{p}(\phi_{J,T_{0}})$.
\end{enumerate}
\end{assumption}

These are high-level moment conditions for idiosyncratic errors. Condition~\ref{enu:uncorr} means that $\mathbb{E}(\bm{f}_{t}u_{jt})=\bm{0}_{r\times1}$ and Condition~\ref{enu:Fu} ensures that the corresponding sampling errors are controlled by $\phi_{J,T_{0}}$ uniformly for all $j$. Condition~\ref{enu:uu} allows for mild cross-sectional dependence in idiosyncratic errors of the control units. This is a weak assumption since it only requires that on average $\mathbb{E}(T_{0}^{-1}\bm{u}_{i}^{\prime}\bm{u}_{j})$ should be controlled by $\phi_{J,T_{0}}$. On the other hand, the weak correlation of the idiosyncratic errors between the treated unit and the controls is for simplicity, so that the predictability is dominantly due to the factors. Sampling errors are restricted in Condition~\ref{enu:uu_dev}. The order $\phi_{J,T_{0}}=\sqrt{(\log J)/T_{0}}$ for the sup-norm of the sampling errors is commonplace in high-dimensional statistics, and can be established from low-level assumptions. See, for instance, \citet{fan2013large} and \citet[Section~6.5]{wainwright2019high}.

\begin{assumption}\label{assu:eta}
As $T_{0}\to\infty$,
\begin{enumerate}[label=(\alph*),parsep=0pt,topsep=0pt,font=\upshape]
\item \label{enu:group-size}Group size: $K\left(\min_{k\in[K]}J_{k}/J\right)\geq\underline{c}$ for some absolute constant $\underline{c}$.
\item \label{enu:tuning}Tuning parameter: $(K\wedge r)^{2}\eta+K\sqrt{r}\phi_{J,T_{0}}/\eta\to0$.
\end{enumerate}
\end{assumption}

Condition~\ref{enu:group-size} ensures that none of the groups is negligible in size. Otherwise, a tiny group contributes little to prediction but its associated weight can be disproportionately large. Part~\ref{enu:tuning} requires the tuning parameter $\eta$ to shrink to 0 at a rate faster than $1/(K\wedge r)^{2}$ to guarantee the convergence of $\hat{\bm{w}}$, but slower than $K\sqrt{r}\phi_{J,T_{0}}$ so that the oracle weight $\bm{w}^{*}$ satisfies the relaxation $\Vert\hat{\bm{\Sigma}}\bm{w}^{*}-\hat{\bm{\Upsilon}}\Vert_{\infty}\leq\eta$ with high probability. This requires that $K$ and $r$, if allowed to diverge, grow at a much slower rate than $T_{0}$. If they have the same order as $T_{0}$, then this condition clearly fails. Intuitively, if $K$ is of the same order as $T_{0}$, then each group contains a finite number of units on average, from which it is not possible to draw a consistent estimator for the oracle weight. They can, for example, tend to infinity at a rate of $\log(T_{0})$, and then $\eta=[\log(T_{0})]^{-3}$ would satisfy this condition. If $K$ and $r$ are fixed, then $\eta$ can be any vanishing sequence tending to zero slower than $\phi_{J,T_{0}}$.

\begin{rem}
\citet*{shiL2RelaxationApplicationsForecast2022}'s theory hinges on the assumption $K\leq r$. This paper breaks this restriction. A large $K$ makes it possible to use the group structures to approximate continuously distributed factor loadings \citep*{bonhomme2022discretizing}, and a large $r$ can deal with proliferation of factors, as is extensively studied in the financial asset pricing literature \citep*{feng2020taming}.
\end{rem}

\phantomsection\label{para:tuning-selection}
The tuning parameter $\eta$ plays an important role in the finite-sample behavior of the estimator. In asymptotic analysis we specify the rate for a proper $\eta$ as the sample size increases. In practice, the tuning parameter is typically chosen by cross validation, and it is what we employ in the numerical work throughout this paper.

\subsection{\texorpdfstring{$L_{2}$-norm Objective Function}{L2-norm Objective Function}}\label{subsec:L2-objective}

The assumptions in the previous section are prepared for the consistency of $L_{2}$-SCM-relaxation.

\begin{thm}[Convergence of $\hat{\bm{w}}$]
\label{thm:w_convergence}
If \crefrange{assu:DGP}{assu:eta} hold, then
\begin{align*}
\|\hat{\bm{w}}-\bm{w}^{*}\|_{1} & =O_{p}\bigl([(K\wedge r)^{2}\eta]^{1/3}\bigr)=o_{p}(1)\text{ and }\\
\|\hat{\bm{w}}-\bm{w}^{*}\|_{2} & =O_{p}\left(\frac{[(K\wedge r)^{2}\eta]^{1/3}}{\sqrt{J}}\right)=o_{p}\left(J^{-1/2}\right).
\end{align*}
\end{thm}

This theorem shows that the estimator $\hat{\bm{w}}$ converges to $\bm{w}^{*}$ under both $L_{1}$- and $L_{2}$-norm with desirable rates of convergence. Since the oracle weight is within-group equal, these convergence results suggest that the estimated weights become approximately equal within each true group.

The rate regime differs from the standard high-dimensional regression rate. Under \Cref{assu:eta}, the tuning parameter $\eta$ vanishes at a rate slower than $K\sqrt{r}\phi_{J,T_0}$. That means the $L_1$-error bound $[(K\wedge r)^2\eta]^{1/3}$ is bounded below by $\phi_{J,T_0}^{1/3}=[(\log J)/T_0]^{1/6}$, slower than the typical rate $\sqrt{\log(J)/T_0}$ for high-dimensional regression.
In our proof, this slower rate arises because the oracle target $\bm{w}^*$ is a high-dimensional dense vector living on the simplex, which makes it more difficult to estimate than a sparse vector. Nevertheless, this rate on the weight does not undermine the resulting ATT estimator for its desirable asymptotic properties; see \Cref{sec:Inference-on-ATT}.

\begin{rem}\label{rem:group-recovery}
SCM-relaxation is not intended to recover the group identities. Nevertheless, under an additional minimum separation condition, it is possible to recover the group membership from the estimated weight vector. Specifically, define $a_k^* = w^*_{\mathcal{G}_k} J/J_k$ and assume $\min_{k\neq\ell}|a_k^*-a_\ell^*| > \underline{\Delta} > 0$ for an absolute constant $\underline{\Delta}$.
By \cref{thm:w_convergence}, all but a vanishing fraction of units have $J\hat{w}_j$ within $\underline{\Delta}/2$ of their group value $a_k^*$. Then, we can conduct a classification algorithm to determine the group membership. As group estimation is not of our interest, we do not pursue a formal theorem for group selection here.
\end{rem}

\begin{rem}\label{rem:key-tech-remark}
The proof of \cref{thm:w_convergence} deviates substantially from that in \citet*{shiL2RelaxationApplicationsForecast2022}, which resorts to the dual problem (unconstrained optimization) of their primal $L_{2}$-relaxation (constrained optimization). If we mimicked their proof, as many as $J$ non-negativity constraints from the simplex would make the dual formulation intractable.
As a result, here, we must develop a new strategy that directly works with the primal problem.
We sketch the proof strategy.\footnote{An unexpected windfall of this strategy of proof is that the asymptotic analysis can be carried over into general convex objective functions, to be elaborated in \cref{subsec:Other-objective-functions}. Moreover, it overcomes a longstanding technical challenge of establishing asymptotic results with a high-dimensional simplex constraint, which was also encountered as the no-short-sell constraint in large portfolios \citep*{jagannathan2003risk}.
}
\phantomsection\label{para:key-step}We first show in \cref{lem:w_star_feasibility} that the oracle weights are feasible with high probability, which implies the basic inequality $\|\hat{\bm{w}}\|_{2}^{2}\leq \|\bm{w}^{*}\|_{2}^{2}$. Since $\bm{w}^{*}=\bm{Z}(\bm{Z}'\bm{Z})^{-1}\bm{w}_{\mathcal{G}}^{*}$ is within-group equal by \cref{lem:oracle_target_relaxation}, the basic inequality then yields a compatibility-type inequality in \cref{lem:link_two_norms}, which links the high-dimensional deviation $\|\hat{\bm{w}}-\bm{w}^{*}\|_{2}$ to the group-level deviation $\|\hat{\bm{w}}_{\mathcal{G}}-\bm{w}_{\mathcal{G}}^{*}\|_{2}$. This step plays a role analogous to \emph{compatibility conditions} in the theory of Lasso for high-dimensional regression \citep[see, e.g.,][Section 6.13]{buhlmann2011statistics}.\footnote{While the Lasso literature usually \emph{assumes} compatibility conditions or related restricted eigenvalue conditions \citep*{bickel2009simultaneous,buhlmann2011statistics}, we instead \emph{derive} a compatibility-type inequality from the latent group structure.} We then use the factor-loading structure to control $\|\hat{\bm{w}}_{\mathcal{G}}-\bm{w}_{\mathcal{G}}^{*}\|_{2}$ or $\|\bm{\Lambda}^{\mathrm{co}\prime} (\hat{\bm{w}}_{\mathcal{G}}-\bm{w}_{\mathcal{G}}^{*})\|_{2}$. Combining these ingredients yields the $L_{1}$- and $L_{2}$-convergence rates of $\hat{\bm{w}}$. See the detailed proof of \cref{thm:w_convergence} in \cref{subsec:proof-thm1} of the Online Appendix.
\end{rem}

\begin{rem}\label{rem:w_convergence_comparison}
We compare our \cref{thm:w_convergence} with that of \citet*{fermanPropertiesSyntheticControl2021}, which analyzes the convergence of $\hat{\bm{\lambda}}^{\mathrm{SC}}:=\bm{\Lambda}'\hat{\bm{w}}^{\mathrm{SC}}$ and shows $|\bm{f}_{t}'(\hat{\bm{\lambda}}^{\mathrm{SC}}-\bm{\lambda}_{0})|\to_{p}0$ under a key condition that there exists a sequence of ``oracle'' weights $\bm{w}^{*}$ such that $\bm{\Lambda}'\bm{w}^{*}\to\bm{\lambda}_{0}$, the asymptotic unbiasedness condition. We instead \emph{prove} that the SCM-relaxation estimator $\hat{\bm{w}}$ converges to $\bm{w}^{*}$ at a sufficiently fast rate. Another interesting result in \citet*{fermanPropertiesSyntheticControl2021} is that the SCM weight $\hat{\bm{w}}^{\mathrm{SC}}$ will spread out over many control units as $J\to\infty$, giving rise to $\Vert\hat{\bm{w}}^{\mathrm{SC}}\Vert_{2}\to_{p}0$, which leads the author to suggest reporting $\Vert\hat{\bm{w}}^{\mathrm{SC}}\Vert_{2}$ in empirical applications to check whether many control units help reduce the bias. Our \cref{thm:w_convergence} proves the oracle-approximation rate $\|\hat{\bm{w}}-\bm{w}^{*}\|_{2}=o_{p}(J^{-1/2})$. Since the oracle is within-group equal, $\|\bm{w}^{*}\|_{2}=O_{p}(\sqrt{K/J})$ implies $\|\hat{\bm{w}}\|_{2}=O_{p}(\sqrt{K/J})$.
\end{rem}

Next, we continue with the empirical risk. For a generic weight $\bm{w}\in\Delta_{J}$, define the in-sample empirical risk as $R_{\mathcal{T}_{0}}(\bm{w}):=T_{0}^{-1}\sum_{t\in\mathcal{T}_{0}}(\bm{w}'\bm{y}_{t}-y_{0t})^{2}$, where $\bm{y}_{t}:=(y_{1t},\dots,y_{Jt})'$, and, similarly, the out-of-sample empirical risk as $R_{\mathcal{T}_{1}}(\bm{w})=T_{1}^{-1}\sum_{t\in\mathcal{T}_{1}}(\bm{w}'\bm{y}_{t}-y_{0t}^{N})^{2}$.

\begin{thm}[Oracle empirical risks]\label{thm:oracle_inequalities}
Under \crefrange{assu:DGP}{assu:eta}, we have
\begin{enumerate}[label=(\roman*),parsep=0pt,topsep=0pt,font=\upshape]
\item \label{enu:in-sample}$R_{\mathcal{T}_{0}}(\hat{\bm{w}})=R_{\mathcal{T}_{0}}(\bm{w}^{*})+o_{p}(1)$ as $T_{0}\to\infty$.
\item \label{enu:out-sample}If $(y_{jt}^{N})_{j\in\{0\}\cup[J],t\in\mathcal{T}_{1}}$ follow the same DGP as that in the pre-treatment periods, and $r\log(J)/T_{1}=O(1)$, then $R_{\mathcal{T}_{1}}(\hat{\bm{w}})=R_{\mathcal{T}_{1}}(\bm{w}^{*})+o_{p}(1)$ as $T_{0},T_{1}\to\infty$.
\end{enumerate}
\end{thm}

Part~\ref{enu:in-sample} is an in-sample oracle equality, and Part~\ref{enu:out-sample} is an out-of-sample oracle equality, where the extra condition $r\log(J)/T_{1}=O(1)$ is mild in that it allows $T_{1}$ to diverge at a much slower rate than $T_{0}$. This theorem shows that empirical risk under the weights estimated by $L_{2}$-SCM-relaxation from the training data is asymptotically as low as the oracle $\bm{w}^{*}$, up to an asymptotically negligible $o_{p}(1)$ gap. It highlights the effectiveness of the relaxation scheme: even though the relaxation method does not seek to identify the group membership, the risk of our estimator would be as good as if we were informed of the infeasible oracle group identities of the control units.

\subsection{Information-theoretic Objective Functions }\label{subsec:Other-objective-functions}

The $L_{2}$-norm objective in \eqref{eq:SCM-relaxation} is one of the information-theoretic divergence measures that have been widely used in the balancing weight literature \citep*{benmichael2021balancing}.
In this section, we will show that the asymptotic guarantees developed above extend to a broader class of convex dispersion functions beyond the $L_{2}$-norm. Let $g(\cdot)$ denote a generic divergence measure that is strictly convex and sufficiently smooth. The associated $g$-SCM-relaxation problem solves
\begin{equation}
\min_{\bm{w}\in S_{\eta}}{}\sum_{j\in[J]}g(w_{j}),\label{eq:general_SCM_relaxation}
\end{equation}
where $S_{\eta}$ is the feasible set defined by \eqref{eq:feasible-set}. Denote the solution as $\hat{\bm{w}}_{(g)}$. Conformably, the oracle weights, denoted $\bm{w}_{(g)}^{*}$, are the solution to
\begin{equation}
\min_{\bm{w}\in S^{*}}{}\sum_{j\in[J]}g(w_{j}).\label{eq:g-oracle}
\end{equation}
Similar to \cref{lem:oracle_target_relaxation}, it is easy to show that $\bm{w}_{(g)}^{*}$ is within-group equal, though a closed-form expression is unavailable for a generic $g$ function.

We study the asymptotic properties of $g$-SCM-relaxation, represented by two popular choices of $g(\cdot)$. The function $g(x)=-\log x$ implicitly keeps weights positive. This choice is closely related to the empirical likelihood \citep*{owen1988empirical}. We call the associated problem \emph{EL-SCM-relaxation}, and denote the solution as $\hat{\bm{w}}^{\mathrm{EL}}$. Alternatively, the entropy function $g(x)=x\log x$ also ensures $x>0$, and it leads to \emph{entropy-SCM-relaxation} with the solution $\hat{\bm{w}}^{\mathrm{entr}}$. \citet*{zheng2024dynamic} use empirical likelihood for synthetic control, while \citet*{hainmuellerEntropyBalancingCausal2012} develops entropy balancing for covariate adjustment in observational studies, both in association with moment equalities. Their theories do not generalize to our contexts, where the moment conditions must be relaxed by inequalities. The general formulation \eqref{eq:general_SCM_relaxation} with a convex dispersion function and an $L_{\infty}$ imbalance constraint also belongs to the class of balancing weight estimators surveyed in \citet*{benmichael2021balancing}.

\begin{rem}
The Cressie-Read (CR) discrepancies \citep*{Cressie1984multi}, $g(x)=(x^{\gamma+1}-1)/[\gamma(\gamma+1)]$ where $\gamma$ is the parameter indexing the family, include the $L_{2}$-norm, EL, and entropy functions as special cases: (i) $g(x)=-\log x$ if $\gamma=-1$, (ii) $g(x)=x\log x$ if $\gamma=0$, and (iii) $g(x)=x^{2}$ if $\gamma=1$. For all $\gamma\in[-1,1]$ the CR discrepancy functions are strictly convex, and can serve as the objective functions for the SCM-relaxation formulation.
\end{rem}

\begin{rem}
The choice of $g$ is related to \emph{Bregman divergence}. For any differentiable convex function $\psi\colon\mathcal{C}\to\mathbb{R}$, the Bregman divergence between points $\bm{x}$ and $\bm{y}$ in $\mathcal{C}$ is defined as
\[
D_{\psi}(\bm{x},\bm{y})=\psi(\bm{x})-\psi(\bm{y})-\langle\nabla\psi(\bm{y}),\bm{x}-\bm{y}\rangle.
\]
The objective functions we study can all be recast as a Bregman divergence between $\bm{w}\in\Delta_{J}$ and $\bm{1}_{J}/J$ induced by respective functions. Specifically, the $L_{2}$-norm is associated with $D_{\psi_{1}}(\bm{w},J^{-1}\bm{1}_{J})=\frac{1}{2}\|\bm{w}-J^{-1}\bm{1}_{J}\|_{2}^{2}$, the EL function with $D_{\psi_{2}}(\bm{w},J^{-1}\bm{1}_{J})=-\sum_{j\in[J]}(\log w_{j}-\log J^{-1})$, and the entropy function with $D_{\psi_{3}}(\bm{w},J^{-1}\bm{1}_{J})=\sum_{j\in[J]}w_{j}(\log w_{j}-\log J^{-1})$. Notice that the $L_{2}$-norm is symmetric, whereas EL and entropy impose asymmetric penalty on positive and negative deviations from the simple average $\bm{1}_{J}/J$ since the divergence is asymmetric.
\end{rem}

To carry out analysis in a unified framework, we introduce a few notions that characterize the shape of function $g$. A differentiable function $g\colon\mathcal{D}\to\mathbb{R}$ is called $\alpha_{g}$\emph{-strongly convex} on $\mathcal{D}\subseteq\mathbb{R}$ if for any $x,y\in\mathcal{D}$, we have $g(x)-g(y)\geq\frac{\mathrm{d}g(y)}{\mathrm{d}y}(x-y)+\frac{\alpha_{g}}{2}(x-y)^{2}$. Since a function $f$ is convex if and only if $f(x)-f(y)\geq\frac{\mathrm{d}f(y)}{\mathrm{d}y}(x-y)$ for any $x$ and $y$, an equivalent condition for $\alpha_{g}$-strong convexity is that $f(x)=g(x)-\frac{\alpha_{g}}{2}x^{2}$ is convex. Clearly, an $\alpha_{g}$-strongly convex function has its Bregman divergence bounded from below by $\frac{\alpha_{g}}{2}(x-y)^{2}$. A function $g\colon\mathcal{D}\to\mathbb{R}$ is called $\beta_{g}$-\emph{Lipschitz} on $\mathcal{D}\subseteq\mathbb{R}$ if for any $x,y\in\mathcal{D}$, we have $|g(x)-g(y)|\leq\beta_{g}|x-y|$. A differentiable $\beta_{g}$-Lipschitz function $g$ must possess bounded derivative $|\mathrm{d}g(x)/\mathrm{d}x|\leq\beta_{g}$ for all $x\in\mathcal{D}$.

\begin{assumption}[Tuning parameter, as a surrogate for \cref{assu:eta}\ref{enu:tuning}]
\label{assu:add_assu}Let $\mathcal{W}_{K,J}\subseteq[0,1]$ be a deterministic interval that can depend on $K$ and $J$ such that $\hat{w}_{(g),j},w_{(g),j}^{*}\in\mathcal{W}_{K,J}$ with probability approaching one uniformly over all $j\in[J]$. Suppose the function $g$ is $\alpha_{g}$-strongly convex and $\beta_{g}$-Lipschitz on $\mathcal{W}_{K,J}$. The tuning parameter $\eta$ satisfies $(K\wedge r)J^{2}(\beta_{g}/\alpha_{g})^{2}\eta+K\sqrt{r}\phi_{J,T_{0}}/\eta\to0$.
\end{assumption}

Typically, the weights $\hat{w}_{(g),j}$ and $w_{(g),j}^{*}$ lie in some interval $[\underline{c}/J,\bar{c}K/J]$. \cref{assu:add_assu} requires the ratio $\beta_{g}/\alpha_{g}$ to be small enough in that interval. We discuss this for $g(x)=-\log x$ and $g(x)=x\log x$.

\begin{example}[EL]
If $g(x)=-\log x$, then for $0<\varepsilon_{1}<\varepsilon_{2}\leq1$, we can show that $g(x)$ is $\varepsilon_{2}^{-2}$-strongly convex on $[\varepsilon_{1},\varepsilon_{2}]$ and $\sup_{x\in[\varepsilon_{1},\varepsilon_{2}]}|\mathrm{d}g(x)/\mathrm{d}x|\leq\varepsilon_{1}^{-1}$. If we take $\varepsilon_{1}\asymp1/J$ and $\varepsilon_{2}\asymp K/J$, then $\beta_{g}/\alpha_{g}=\varepsilon_{1}^{-1}/\varepsilon_{2}^{-2}=O(K^{2}/J)$. Therefore \cref{assu:add_assu} holds if $K^{5}\eta+K\sqrt{r}\phi_{J,T_{0}}/\eta\to0$.
\end{example}

\begin{example}[Entropy]
If $g(x)=x\log x$, for $0<\varepsilon_{1}<\varepsilon_{2}\leq1$, the function $g(x)$ is $\varepsilon_{2}^{-1}$-strongly convex on $[\varepsilon_{1},\varepsilon_{2}]$ and $\sup_{x\in[\varepsilon_{1},\varepsilon_{2}]}|\mathrm{d}g(x)/\mathrm{d}x|\leq1+|\log\varepsilon_{1}|$. Taking $\varepsilon_{1}\asymp1/J$ and $\varepsilon_{2}\asymp K/J$, we get $\beta_{g}/\alpha_{g}=(1+|\log\varepsilon_{1}|)\varepsilon_{2}=O(K[\log J]/J)$. \cref{assu:add_assu} is satisfied if $K^{3}(\log J)^{2}\eta+K\sqrt{r}\phi_{J,T_{0}}/\eta\to0$.
\end{example}

Given a proper choice of $\eta$, the consistency of $\hat{\bm{w}}_{(g)}$
follows.

\begin{thm}\label{thm:wg_convergence}
If \cref{assu:DGP,assu:convex_hull,assu:errors,assu:add_assu} and \cref{assu:eta}\ref{enu:group-size} hold, we have
\begin{align*}
\|\hat{\bm{w}}_{(g)}-\bm{w}_{(g)}^{*}\|_{1} & =O_{p}\left(\left[\frac{(K\wedge r)J^{2}\beta_{g}^{2}\eta}{\alpha_{g}^{2}}\right]^{1/3}\right)=o_{p}(1)\text{ and}\\
\|\hat{\bm{w}}_{(g)}-\bm{w}_{(g)}^{*}\|_{2} & =O_{p}\left(\left[\frac{(K\wedge r)J^{2}\beta_{g}^{2}\eta}{\alpha_{g}^{2}}\right]^{1/3}\frac{1}{\sqrt{J}}\right)=o_{p}(J^{-1/2}).
\end{align*}
\end{thm}

\cref{thm:wg_convergence} guarantees that $\hat{\bm{w}}_{(g)}$ behaves as well as the oracle weight $\bm{w}^{*}$ in large samples, as in the following corollary.

\begin{cor}[Oracle equalities]\label{cor:mse_wg}
Under the conditions in \cref{thm:wg_convergence}, we have $R_{\mathcal{T}_{0}}(\hat{\bm{w}}_{(g)})=R_{\mathcal{T}_{0}}(\bm{w}^{*})+o_{p}(1)$. In addition, if the conditions in \cref{thm:oracle_inequalities}\ref{enu:out-sample} also hold, then $R_{\mathcal{T}_{1}}(\hat{\bm{w}}_{(g)})=R_{\mathcal{T}_{1}}(\bm{w}^{*})+o_{p}(1)$.
\end{cor}

\phantomsection\label{para:l2-advantages}Although \cref{cor:mse_wg} shows that SCM-relaxation with the EL or entropy objective function can also achieve out-of-sample oracle performance analogous to the $L_{2}$-norm objective in \cref{thm:oracle_inequalities}, the $L_{2}$ objective has two practical advantages in the setting studied here. First, to achieve the same convergence rate, both EL and entropy rely on extra technical conditions to ensure $\left\Vert \hat{\bm{w}}_{(g)}\right\Vert _{\infty}\leq\bar{c}K/J$.\footnote{This upper bound, determined by $K$ and $J$, is proven to hold for the $L_{2}$-norm objective.} Second, EL and entropy implicitly require the weights to be strictly positive, which incurs bias when some true group weights $w_{\mathcal{G}_{k}}^{*}$ are zero. This is particularly undesirable because the true weight can lie on the boundary of the simplex, as stressed by \citet*{abadieUsingSyntheticControls2021}.

\begin{rem}
In the large literature of the GEL \citep*{neweyHigherOrderProperties2004}, researchers often view the choices of the divergence measures as interchangeable---after all, the empirical likelihood, exponential tilting \citep*{schennachPointEstimationExponentially2007}, and continuously updated GMM \citep*{Hansen1996finite} are asymptotically equivalent. This perspective is valid when the observations are independent and identically distributed, in that each of the implied probability weights tends to the equal weight asymptotically. In sharp contrast, here we allocate weights on heterogeneous cross-sectional units. If there are at least two groups, then the group oracle weight is in general unequal, and it is possible that one of the groups has zero weight. This highlights the key difference between the knowledge we learn from the GEL literature and the current context of SCM where we are faced with heterogeneous individual units.
\end{rem}

Before we conclude this section, we would like to mention that feature engineering matters for all machine learning methods, including ours. If the scales of the empirical data vary substantially across individuals, then a scale-standardization step can be desirable before feeding the data into the SCM-relaxation algorithm. All theoretical results discussed above can be extended to the standardized version in a straightforward fashion. See \cref{sec:Standardized-SCM-relaxation} for details.

\section{Inference on ATT}\label{sec:Inference-on-ATT}

In many applications where the number of post-treatment periods $T_{1}$ is comparable with $T_{0}$, researchers are interested in conducting inference on the \emph{average treatment effect on the treated} (ATT), where the average is taken over post-treatment periods:
\[
\tau_{0}=\frac{1}{T_{1}}\sum_{t\in\mathcal{T}_{1}}(y_{0t}^{I}-y_{0t}^{N}).
\]
Since $y_{0t}^{I}$ is observed in post-treatment periods, only $y_{0t}^{N}$ needs to be estimated via synthetic control. After obtaining the weight estimate, one can construct the post-treatment counterfactual of the treated as
\[
\hat{y}_{0t}^{N}=\sum_{j\in[J]}\hat{w}_{j}y_{jt}.
\]
The ATT estimator is thus
\[
\hat{\tau}=\frac{1}{T_{1}}\sum_{t\in\mathcal{T}_{1}}(y_{0t}^{I}-\hat{y}_{0t}^{N})=\frac{1}{T_{1}}\sum_{t\in\mathcal{T}_{1}}\Biggl[y_{0t}^{I}-\sum_{j\in[J]}\hat{w}_{j}y_{jt}\Biggr].
\]

To establish the asymptotic properties of the ATT estimator, we impose additional regularity conditions on the factors and idiosyncratic errors $u_{jt}$.

\begin{assumption}\label{assu:additional}\mbox{}
\begin{enumerate}[label=(\alph*),parsep=0pt,topsep=0pt,font=\upshape]
\item \label{enu:mixing}$\{(\bm{f}_{t},u_{0t},u_{1t},\dots,u_{Jt})\}_{t\in[T]}$ are strictly stationary and $\beta$-mixing with mixing coefficient $\beta_{\mathrm{mix}}(\tau)\leq C\tau^{-\kappa}$ for some constant $C>0$ and $\kappa\geq2$.
\item \label{enu:tail}$\sup_{\ell\in[r],t\in[T]}\mathbb{E}(|f_{\ell t}-\mathbb{E}(f_{\ell t})|^{4+\varepsilon})+\sup_{j\in\{0\}\cup[J],t\in[T]}\mathbb{E}(|u_{jt}|^{4+\varepsilon})<\infty$ for some $\varepsilon>0$.
\end{enumerate}
\end{assumption}

Condition~\ref{enu:mixing} is a standard assumption used in the time series literature to materialize weak dependence as $\beta$-mixing. The intensity of weak dependence is controlled by the decaying mixing coefficient. Condition~\ref{enu:tail} imposes uniform moment boundedness for the (centered) factors and idiosyncratic errors. These two conditions allow us to invoke a version of the central limit theorem for weakly dependent time series.

\begin{thm}[Consistency of ATT]\label{thm:att}
Suppose \crefrange{assu:DGP}{assu:eta} and \ref{assu:additional} hold. As $T_{0},T_{1}\to\infty$,
\[
\hat{\tau}-\tau_{0}=O_{p}\bigl(\sqrt{r(K\wedge r)}\,\eta\bigr)+O_{p}\bigl(\sqrt{K/T_{1}}\bigr).
\]
Consequently, $\hat{\tau}$ is consistent for $\tau_{0}$ whenever $\sqrt{r(K\wedge r)}\,\eta\to0$ and $K/T_{1}\to0$.
\end{thm}

The condition $\bm{\Lambda}^{\mathrm{co}\prime}\bm{w}_{\mathcal{G}}^{*}=\bm{\lambda}_{0}$ in \cref{assu:convex_hull} ensures that the oracle weights exactly replicate the treated unit's factor loading, so that the bias term vanishes. Without \cref{assu:convex_hull}, the ATT estimator can still be computed, but our theory does not guarantee its consistency. The within-group equal oracle weights yield $\|\bm{w}^{*}\|_{\infty}=O(K/J)$ and $\|\bm{w}^{*}\|_{2}^{2}=O(K/J)$. Consequently, the weighted noise term is diluted through diversification across the $J/K$ units per group, without requiring the idiosyncratic error variance itself to shrink.

There are multiple ways to develop an inferential procedure for our SCM-relaxation estimator. 
\phantomsection\label{para:inference-method}
For inference on ATT, \citet*{chernozhukov_t-test_2024}'s approach is well suited and sufficiently general. It requires $L_{2}$-consistency of the weight estimator. For completeness, we briefly introduce this method.

The regularization due to the simplex constraint and the relaxation constraint is expected to introduce asymptotic bias to the raw ATT estimator $\hat{\tau}$. The cross-fitting correction targets this ATT bias rather than the bias of the estimated weights themselves. We form $B$ consecutive disjoint blocks from the pre-treatment periods
$
H_{1}\cup H_{2}\cup\dots\cup H_{B}\subseteq\mathcal{T}_{0}.
$
Let $T=T_{0}+T_{1}$. The block size is given by $d_{T}:=\lfloor T_{0}/B\rfloor\wedge T_{1}$
so that
\[
H_{k}=\{(k-1)d_{T}+1,\dots,kd_{T}\}\text{ for }k\in[B].
\]
For each block $k=1,\dots,B$, compute
\[
\widehat{\tau}_{k}=\underbrace{\frac{1}{T_{1}}\sum_{t\in\mathcal{T}_{1}}(y_{0t}-\widehat{\bm{w}}'_{(k)}\bm{y}_{t})}_{\text{ATT estimator}}-\underbrace{\frac{1}{|H_{k}|}\sum_{t\in H_{k}}(y_{0t}-\widehat{\bm{w}}_{(k)}'\bm{y}_{t})}_{\text{bias estimate using SCM-relaxation}},
\]
where $\widehat{\bm{w}}_{(k)}$ is obtained by applying SCM-relaxation to the pre-treatment periods excluding block $k$, $H_{(-k)}:=\mathcal{T}_{0}\backslash H_{k}$. This construction induces approximate independence between $\widehat{\bm{w}}_{(k)}$ and the data in $H_{k}\cup\mathcal{T}_{1}$. Across all blocks, we take the simple average of $\widehat{\tau}_{1},\dots,\widehat{\tau}_{B}$ to construct the final bias-corrected (bc) estimator for ATT as
\[
\widehat{\tau}_{\text{bc}}=\frac{1}{B}\sum_{k\in[B]}\widehat{\tau}_{k}.
\]
Although this estimator can be shown to be asymptotically normal, to avoid estimating long-run variance, \citet*{chernozhukov_t-test_2024} recommend the scale-free $t$-test statistic
\[
\mathbb{T}_{B}=\frac{\sqrt{B}(\widehat{\tau}_{\text{bc}}-\tau_{0})}{\widehat{\sigma}}, \quad \mbox{ where }
\widehat{\sigma}=\sqrt{1+\frac{Bd_{T}}{T_{1}}}\sqrt{\frac{1}{B-1}\sum\nolimits_{k\in[B]}(\widehat{\tau}_{k}-\widehat{\tau}_{\text{bc}})^{2}}.
\]
The following theorem establishes the asymptotic distribution of $\mathbb{T}_B$.

\begin{thm}\label{thm:asym-dist}
Suppose \crefrange{assu:DGP}{assu:additional} hold, $\sqrt{r(K\wedge r)}\,\eta\to0$, and $r=o((\log T)^{2})$. If $T_{0}\asymp T_{1} \to\infty$, $J=O(T_{0})$, and $B\geq2$ is fixed, then
\[
\mathbb{T}_{B}\to_{d}t_{B-1},
\]
where $t_{B-1}$ is a random variable following the $t$-distribution with degrees of freedom $B-1$.
\end{thm}

To obtain concentration, this theorem requires $T_{1}$ to be of the same order as $T_{0}$, which is applicable to our empirical application. We will also illustrate its numerical performance in simulations.

\section{Monte Carlo Experiments}\label{sec:Monte-Carlo-experiments}

In this section, we conduct Monte Carlo simulations to study the finite sample performance. We consider three sets of simulations: (i) the factor loadings follow an exact group structure as \eqref{eq:grouped_structure} specifies, (ii) the factor loadings fluctuate around the group means so that \eqref{eq:grouped_structure} is violated, and (iii) the factor loadings follow a general factor structure without groups. For each set, we consider a correctly specified case, in which the loading of the treated unit lies in the convex hull of the donor loadings, and a misspecified case, in which the treated unit has an intercept and is more exposed than any donor to a local shock that only some donors share.

The simulated data consist of $J+1$ units and $T_{0}+T_{1}$ time periods. The treatment to the unit $j=0$ occurs immediately after time $T_{0}$ and sustains from $T_{0}+1$ on. The potential outcome, if untreated, follows the factor model \eqref{eq:model} with $r=4$ factors. The idiosyncratic errors $u_{jt}\stackrel{\text{i.i.d.}}{\sim}\mathcal{N}(0,1)$ and are independent of the common factors. The $r$ factors are mutually independent, and each follows a stationary AR(1) process
\begin{equation}\label{eq:sim-factors}
f_{l t} = 0.5 f_{l, t-1} + \varepsilon_{l t}^f,
\text{ where } \varepsilon_{l t}^f \stackrel{\text{i.i.d.}}{\sim} \mathcal{N}(0, 1)\text{ for } l \in [r].
\end{equation}
In the misspecified case, we add a local shock $g_t$, which follows the AR(1) process $g_t = 0.5\, g_{t-1} + \varepsilon_t^g$ with $\varepsilon_t^g \stackrel{\text{i.i.d.}}{\sim} \mathcal{N}(0, 0.5)$ and is independent of $\bm{f}_t$. We randomly draw 20\% of the donors as the neighbors of the treated unit and collect them in $\mathcal{R}$. The potential outcomes are
\begin{equation}\label{eq:sim-misspec}
\begin{aligned}
y_{jt}^N & = \bm{\lambda}_j'\bm{f}_t + h_j\, g_t + u_{jt}\text{ for } j \in [J]\text{ and }\\
y_{0t}^N & = 0.25 + \bm{\lambda}_0'\bm{f}_t + 0.5\, g_t + u_{0t},
\end{aligned}
\end{equation}
where $h_j \stackrel{\text{i.i.d.}}{\sim} U(0.15, 0.35)$ for $j \in \mathcal{R}$ and $h_j = 0$ otherwise. The augmented loading $(\bm{\lambda}_0', 0.5)'$ thus lies outside the convex hull of the donor loadings $(\bm{\lambda}_j', h_j)'$, and \cref{assu:convex_hull} fails. The additional intercept introduces a constant bias for synthetic controls.

We estimate the weights by the relaxation methods, including $L_{2}$-, EL-, and entropy-SCM-relaxation, and compare them with SCM \citep*{abadieSyntheticControlMethods2010}. We also compute the weights by the off-the-shelf Lasso, Ridge, and Group Lasso \citep*{yuan2006model} methods:
\begin{align*}
\hat{\bm{w}}^{\mathrm{Lasso}} & =\arg\min_{\bm{w}\in\Delta_{J}}\Bigl\{ T_{0}^{-1}\bigl\Vert \bm{y}_{0}-\bm{Y}\bm{w}\bigr\Vert _{2}^{2}+\lambda^{\mathrm{Lasso}}\|\bm{w}-J^{-1}\bm{1}_{J}\|_{1}\Bigr\} ,\\
\hat{\bm{w}}^{\mathrm{Ridge}} & =\arg\min_{\bm{w}\in\Delta_{J}}\Bigl\{ T_{0}^{-1}\bigl\Vert \bm{y}_{0}-\bm{Y}\bm{w}\bigr\Vert _{2}^{2}+\lambda^{\mathrm{Ridge}}\|\bm{w}-J^{-1}\bm{1}_{J}\|_{2}^{2}\Bigr\} ,\text{ and}\\
\hat{\bm{w}}^{\mathrm{GL}} & =\arg\min_{\bm{w}\in\Delta_{J}}\Bigl\{ T_{0}^{-1}\bigl\Vert \bm{y}_{0}-\bm{Y}\bm{w}\bigr\Vert _{2}^{2}+\lambda^{\mathrm{GL}}\sum_{k\in[K]}\sqrt{J_{k}}\|\bm{w}_{\mathcal{G}_{k}}-J^{-1}\bm{1}_{J_{k}}\|_{2}\Bigr\} ,
\end{align*}
where $\lambda^{\mathrm{Lasso}}$, $\lambda^{\mathrm{Ridge}}$, and $\lambda^{\mathrm{GL}}$ are tuning parameters, and the weights are shrunken toward the simple average for a fair comparison. Group Lasso utilizes the oracle group information, so it is an infeasible benchmark, and we use it only in the designs with group structures. \phantomsection\label{para:added-methods}We further report the performance of five methods from the literature: the forward-selected PDA (fsPDA) of \citet{shi2023forward}, the Synthetic Difference-in-Differences (SDID) of \citet*{arkhangelsky2021synthetic}, the Augmented SCM (AugSC) of \citet*{ben2021augmented}, the Penalized SCM (PenSC) of \citet{abadiePenalizedSyntheticControl2021}, and the Generalized SCM (GSC) of \citet{xu2017generalized}. AugSC and PenSC use all pre-treatment outcomes as covariates, so that all methods use the same pre-treatment information.

The tuning parameters of the relaxation methods, Lasso, Ridge, and Group Lasso are selected via cross-validation (CV). For $\eta$, we use 10-fold CV over contiguous time blocks, and the grid has ten points spaced geometrically on $[\underline{\eta},\bar{\eta}]$, where $\underline{\eta}$ is the smallest value for which the relaxed problem is feasible in the full sample and in every CV training sample, and the upper bound $\bar{\eta}=\Vert\hat{\bm{\Sigma}}\bm{1}_{J}/J-\hat{\bm{\Upsilon}}\Vert_{\infty}$ so that $J^{-1}\bm{1}_{J}$ is a feasible solution. The $\eta$ that CV selects most often minimizes the prediction error in most settings (\cref{tab:eta-sensitivity}). In all sets of experiments, $T_{1}$ is fixed at $50$, $J\in\{40,80,120\}$, and $T_{0}\in\{30,60,90\}$. With group structures, $K\in\{3,4,5\}$, which corresponds to $K<r$, $K=r$, and $K>r$, respectively. The experiment is repeated 1000 times in each DGP.

We calculate the oracle weights $\bm{w}^{*}$ following \eqref{eq:oracle_relaxation_problem}. In the approximate group design, we compute them from the unperturbed loadings $\bm{Z}\bm{\Lambda}^{\mathrm{co}}$, and in the general factor design, $\bm{w}^*$ is the weight on the simplex with the smallest $L_2$-norm that reproduces the loading of the treated unit. Define the oracle synthetic control as $y_{0t}^{N,*}:=\sum_{j=1}^{J}w_{j}^{*}y_{jt}^{N}$. In the correctly specified case, we assess the out-of-sample performance of each estimator by the root mean squared error $\text{RMSE}^{\text{o}} = \bigl[T_1^{-1}\sum_{t \in \mathcal{T}_1} (\hat{y}_{0t}^N - y_{0t}^{N,*})^2\bigr]^{1/2}$. Moreover, we present the $L_{2}$-distance $\|\hat{\bm{w}}-\bm{w}^{*}\|_{2}$ to check the quality of the weight estimation. In the misspecified case, no oracle weight exists, and we use the root mean squared error relative to the true outcome, $\text{RMSE}^{\star} = \bigl[T_1^{-1}\sum_{t \in \mathcal{T}_1} (\hat{y}_{0t}^N - y_{0t}^N)^2\bigr]^{1/2}$, which is at least 1 due to the unpredictable variance. Finally, we employ the cross-fitting method with $B=5$ to compute the debiased treatment effects. We report the absolute bias of the ATT estimator, the empirical coverage probability (CP) of the confidence interval (CI) with 90\% nominal rate, and the ratio of its average length to that of SCM.

\subsection{Exact Group Structures}\label{subsec:MC-exact}

Each entry in the core factor loadings $\bm{\Lambda}^{\mathrm{co}}$ is independently drawn from $\mathcal{N}(0,9/r)$. The loadings for the control units are given by $\bm{\Lambda}=\bm{Z}\bm{\Lambda}^{\mathrm{co}}$, where the first $K-1$ groups have $\lfloor J/K\rfloor$ members each and the last group has the remaining donors. The loading for the treated unit is generated as a convex combination of the core loadings with zero weight on the first group, and the weights on the other groups are jointly generated from a Dirichlet distribution with all parameters equal to 1.5.

The RMSE\textsuperscript{o} in \cref{tab:rmse_oracle_exact} reveals that Lasso, Ridge, Group Lasso, and the relaxation methods all outperform the canonical SCM. Lasso, Ridge, and Group Lasso surpass SCM by penalizing solutions toward the simple average. Moreover, Ridge also encourages dense solutions, which gives it an edge over Lasso. Group Lasso utilizes the oracle group information, and its performance is close to that of Ridge. Leveraging latent group structures and dense weights, $L_2$- and entropy-SCM-relaxation further improve on Ridge and Group Lasso in all settings. Among the relaxation approaches, the $L_{2}$-norm objective excels, especially when $K<r$, due to its advantage of allowing for exact zero weights. The RMSE\textsuperscript{o} of the relaxation methods decreases as $T_0$ increases. Among the other methods from the literature, AugSC is close to SCM, and PenSC, SDID, and fsPDA are worse than SCM. The unfavorable performance of fsPDA is anticipated, as the oracle weights are dense and constrained on the simplex while the estimated weights of fsPDA do not lie on the simplex and are typically sparse. GSC outperforms SCM, and its advantage grows with $T_0$ because the factors and the loadings are estimated more accurately as $T_0$ grows.

\begin{table}[!htbp]
\centering
\small
\caption{RMSE\textsuperscript{o} (exact group)}
\label{tab:rmse_oracle_exact}
\smallskip
\sbox{\tabbox}{\resizebox{\textwidth}{!}{\hdreset\hdbump{$L_2$}\hdbump{EL}\hdbump{Entropy}\hdbump{SCM}\hdbump{Lasso}\hdbump{Ridge}\hdbump{GLasso}\hdbump{fsPDA}\hdbump{AugSC}\hdbump{PenSC}\hdbump{SDID}\hdbump{GSC}\hdbump{\textbf{8.888}}\setlength{\tabcolsep}{4pt}%
\begin{tabular}{rrS[table-format=1.3]S[table-format=1.3]S[table-format=1.3]S[table-format=1.3]S[table-format=1.3]S[table-format=1.3]S[table-format=1.3]S[table-format=1.3]S[table-format=1.3]S[table-format=1.3]S[table-format=1.3]S[table-format=1.3]}
\toprule
\multicolumn{2}{c}{} & \multicolumn{3}{c}{Relaxation} & \multicolumn{9}{c}{} \\
\cmidrule(lr){3-5}
$J$ & $T_0$ & \hd{$L_2$} & \hd{EL} & \hd{Entropy} & \hd{SCM} & \hd{Lasso} & \hd{Ridge} & \hd{GLasso} & \hd{fsPDA} & \hd{AugSC} & \hd{PenSC} & \hd{SDID} & \hd{GSC} \\
\midrule
\multicolumn{14}{c}{\footnotesize\textbf{Panel A. $\bm{K<r}$}} \\
40 & 30 & \textbf{.313} & .394 & .326 & .464 & .428 & .383 & .381 & .906 & .487 & .700 & .784 & .396 \\
 & 60 & \textbf{.220} & .300 & .230 & .371 & .327 & .282 & .285 & .641 & .371 & .554 & .715 & .273 \\
 & 90 & \textbf{.185} & .265 & .194 & .328 & .284 & .236 & .238 & .558 & .327 & .437 & .694 & .220 \\
\addlinespace[.45em]
80 & 30 & \textbf{.312} & .389 & .318 & .461 & .419 & .365 & .367 & 1.066 & .466 & .692 & .581 & .388 \\
 & 60 & \textbf{.218} & .303 & .229 & .374 & .326 & .276 & .274 & .669 & .381 & .539 & .514 & .263 \\
 & 90 & \textbf{.177} & .257 & .186 & .332 & .284 & .234 & .236 & .589 & .334 & .454 & .495 & .213 \\
\addlinespace[.45em]
120 & 30 & \textbf{.302} & .398 & .314 & .461 & .426 & .366 & .365 & 1.155 & .462 & .672 & .522 & .375 \\
 & 60 & \textbf{.207} & .301 & .217 & .370 & .327 & .276 & .278 & .694 & .371 & .559 & .440 & .253 \\
 & 90 & \textbf{.173} & .260 & .183 & .329 & .283 & .231 & .229 & .597 & .328 & .447 & .408 & .205 \\
\midrule
\multicolumn{14}{c}{\footnotesize\textbf{Panel B. $\bm{K=r}$}} \\
40 & 30 & \textbf{.373} & .414 & .378 & .494 & .453 & .415 & .417 & .937 & .522 & .720 & .759 & .442 \\
 & 60 & \textbf{.265} & .310 & .270 & .390 & .348 & .312 & .310 & .641 & .394 & .549 & .697 & .308 \\
 & 90 & \textbf{.215} & .260 & .219 & .340 & .299 & .259 & .262 & .556 & .340 & .462 & .676 & .248 \\
\addlinespace[.45em]
80 & 30 & \textbf{.358} & .407 & .365 & .488 & .447 & .405 & .404 & 1.057 & .493 & .717 & .598 & .410 \\
 & 60 & \textbf{.251} & .303 & .258 & .385 & .338 & .296 & .297 & .673 & .390 & .560 & .521 & .282 \\
 & 90 & \textbf{.208} & .256 & .216 & .341 & .289 & .248 & .247 & .585 & .342 & .457 & .481 & .233 \\
\addlinespace[.45em]
120 & 30 & \textbf{.354} & .408 & .364 & .485 & .444 & .401 & .398 & 1.148 & .487 & .716 & .525 & .415 \\
 & 60 & \textbf{.247} & .300 & .254 & .384 & .334 & .293 & .294 & .700 & .385 & .527 & .443 & .287 \\
 & 90 & \textbf{.206} & .259 & .216 & .342 & .294 & .249 & .247 & .599 & .343 & .446 & .409 & .233 \\
\midrule
\multicolumn{14}{c}{\footnotesize\textbf{Panel C. $\bm{K>r}$}} \\
40 & 30 & .396 & .417 & \textbf{.395} & .514 & .466 & .437 & .438 & .916 & .557 & .778 & .715 & .452 \\
 & 60 & \textbf{.282} & .311 & .287 & .400 & .354 & .324 & .327 & .631 & .401 & .549 & .676 & .320 \\
 & 90 & \textbf{.238} & .267 & .242 & .346 & .303 & .272 & .275 & .546 & .346 & .487 & .631 & .267 \\
\addlinespace[.45em]
80 & 30 & \textbf{.385} & .408 & .385 & .501 & .451 & .417 & .421 & 1.028 & .506 & .777 & .571 & .430 \\
 & 60 & \textbf{.272} & .299 & .277 & .397 & .345 & .312 & .312 & .656 & .402 & .540 & .495 & .296 \\
 & 90 & \textbf{.231} & .260 & .235 & .349 & .299 & .264 & .268 & .574 & .350 & .476 & .455 & .250 \\
\addlinespace[.45em]
120 & 30 & \textbf{.392} & .423 & .396 & .508 & .459 & .425 & .429 & 1.145 & .511 & .742 & .538 & .440 \\
 & 60 & \textbf{.275} & .301 & .278 & .392 & .340 & .304 & .303 & .689 & .394 & .551 & .424 & .293 \\
 & 90 & \textbf{.221} & .250 & .225 & .343 & .292 & .257 & .258 & .578 & .344 & .458 & .390 & .235 \\
\bottomrule
\end{tabular}
}}\usebox{\tabbox}
\par\vspace{.35\baselineskip}
\begin{minipage}{\wd\tabbox}\footnotesize
\textit{Notes}: RMSE\textsuperscript{o} $= \bigl[T_1^{-1}\sum_{t \in \mathcal{T}_1} (\hat{y}_{0t}^N - y_{0t}^{N,*})^2\bigr]^{1/2}$ compares the estimated counterfactual with the oracle counterfactual $y_{0t}^{N,*} = \sum_{j=1}^J w_j^* y_{jt}^N$.
\end{minipage}
\end{table}

The performance of the estimated weights is displayed in \cref{tab:wdist_exact}. Lasso and Ridge are better at recovering the oracle weights than SCM, and Ridge outperforms Lasso in terms of $L_{2}$-distance as it tends to produce weights with a smaller $L_{2}$-norm. Among the SCM-relaxation methods, $L_{2}$-SCM-relaxation recovers the within-group equal weights most accurately, since only the $L_2$ objective allows for exact zero weights, while the EL and entropy objectives force positive weights onto the zero-weight group. These findings align with our theory.

\begin{table}[!htbp]
\centering
\small
\caption{$\|\hat{\bm{w}}-\bm{w}^*\|_{2}$ (exact group)}
\label{tab:wdist_exact}
\smallskip
\sbox{\tabbox}{\resizebox{0.78\textwidth}{!}{\hdreset\hdbump{$L_2$}\hdbump{EL}\hdbump{Entropy}\hdbump{SCM}\hdbump{Lasso}\hdbump{Ridge}\hdbump{GLasso}\hdbump{AugSC}\hdbump{PenSC}\hdbump{\textbf{.888}}\setlength{\tabcolsep}{4pt}%
\begin{tabular}{rrS[table-format=0.3]S[table-format=0.3]S[table-format=0.3]S[table-format=0.3]S[table-format=0.3]S[table-format=0.3]S[table-format=0.3]S[table-format=0.3]S[table-format=0.3]}
\toprule
\multicolumn{2}{c}{} & \multicolumn{3}{c}{Relaxation} & \multicolumn{6}{c}{} \\
\cmidrule(lr){3-5}
$J$ & $T_0$ & \hd{$L_2$} & \hd{EL} & \hd{Entropy} & \hd{SCM} & \hd{Lasso} & \hd{Ridge} & \hd{GLasso} & \hd{AugSC} & \hd{PenSC} \\
\midrule
\multicolumn{11}{c}{\footnotesize\textbf{Panel A. $\bm{K<r}$}} \\
40 & 30 & \textbf{.097} & .195 & .113 & .373 & .267 & .189 & .188 & .389 & .465 \\
 & 60 & \textbf{.074} & .159 & .085 & .318 & .230 & .152 & .151 & .319 & .396 \\
 & 90 & \textbf{.060} & .144 & .071 & .287 & .206 & .124 & .123 & .288 & .348 \\
\addlinespace[.45em]
80 & 30 & \textbf{.094} & .189 & .107 & .370 & .259 & .174 & .173 & .374 & .463 \\
 & 60 & \textbf{.068} & .155 & .079 & .320 & .227 & .143 & .139 & .325 & .389 \\
 & 90 & \textbf{.053} & .128 & .062 & .291 & .205 & .123 & .123 & .293 & .347 \\
\addlinespace[.45em]
120 & 30 & \textbf{.090} & .195 & .107 & .369 & .267 & .174 & .176 & .371 & .459 \\
 & 60 & \textbf{.067} & .162 & .076 & .319 & .232 & .145 & .145 & .320 & .393 \\
 & 90 & \textbf{.054} & .137 & .063 & .288 & .199 & .116 & .112 & .288 & .346 \\
\midrule
\multicolumn{11}{c}{\footnotesize\textbf{Panel B. $\bm{K=r}$}} \\
40 & 30 & \textbf{.099} & .167 & .111 & .362 & .250 & .187 & .184 & .383 & .440 \\
 & 60 & \textbf{.078} & .138 & .086 & .312 & .217 & .153 & .149 & .314 & .373 \\
 & 90 & \textbf{.067} & .121 & .074 & .282 & .199 & .130 & .129 & .282 & .337 \\
\addlinespace[.45em]
80 & 30 & \textbf{.095} & .165 & .107 & .364 & .250 & .180 & .181 & .367 & .441 \\
 & 60 & \textbf{.070} & .132 & .080 & .311 & .212 & .140 & .139 & .315 & .371 \\
 & 90 & \textbf{.059} & .115 & .068 & .286 & .192 & .119 & .119 & .287 & .335 \\
\addlinespace[.45em]
120 & 30 & \textbf{.087} & .161 & .102 & .360 & .245 & .176 & .171 & .361 & .433 \\
 & 60 & \textbf{.067} & .133 & .077 & .310 & .213 & .137 & .137 & .311 & .363 \\
 & 90 & \textbf{.056} & .114 & .064 & .285 & .194 & .119 & .118 & .286 & .329 \\
\midrule
\multicolumn{11}{c}{\footnotesize\textbf{Panel C. $\bm{K>r}$}} \\
40 & 30 & \textbf{.108} & .153 & .115 & .358 & .237 & .185 & .185 & .387 & .444 \\
 & 60 & \textbf{.087} & .131 & .094 & .308 & .211 & .154 & .154 & .309 & .366 \\
 & 90 & \textbf{.077} & .113 & .082 & .278 & .191 & .132 & .134 & .279 & .331 \\
\addlinespace[.45em]
80 & 30 & \textbf{.097} & .148 & .106 & .356 & .229 & .176 & .174 & .359 & .438 \\
 & 60 & \textbf{.076} & .115 & .082 & .308 & .200 & .140 & .139 & .311 & .364 \\
 & 90 & \textbf{.067} & .103 & .071 & .283 & .186 & .126 & .129 & .283 & .330 \\
\addlinespace[.45em]
120 & 30 & \textbf{.090} & .143 & .100 & .354 & .225 & .165 & .167 & .356 & .425 \\
 & 60 & \textbf{.071} & .115 & .078 & .304 & .195 & .131 & .131 & .305 & .360 \\
 & 90 & \textbf{.062} & .096 & .068 & .281 & .183 & .122 & .120 & .282 & .323 \\
\bottomrule
\end{tabular}
}}\usebox{\tabbox}
\par\vspace{.35\baselineskip}
\begin{minipage}{\wd\tabbox}\footnotesize
\textit{Notes}: The weight distance is reported for methods with weights on the simplex and for AugSC, whose weights are the SCM weights plus a ridge correction. fsPDA, SDID, and GSC are excluded because the weights of fsPDA and GSC are not restricted to the simplex, while SDID constructs the counterfactual not only using the unit weights.
\end{minipage}
\end{table}

To see how well these methods recover the group structure of oracle weights, \cref{fig:weight-comparison-simu} illustrates the estimated weights in a single replication with 3 groups, 4 factors, 120 control units, and 90 pre-treatment periods. As is clear in the figure, SCM produces sparse weights, concentrating on groups with nonzero oracle weights but deviating far from the within-group equal oracle weight. Lasso sets nearly all weights of the zero-weight group to zero but tends to collapse its other weights to the simple average, the target of shrinkage. Ridge, with the denser solution, better approximates the group structure, and Group Lasso approximates it even better because its penalty function uses the oracle group information. As predicted by our theory, $L_{2}$-SCM-relaxation nearly perfectly traces the group pattern, including the group with zero oracle weight. While EL and entropy objectives partially recover the group patterns, EL fluctuates most for groups with higher oracle weights, and entropy shrinks the middle group and inflates the top one. This variability stems from the nonuniform curvature of their objective functions, as measured by second derivatives \citep[Section 9.3]{wainwright2019high}. The second derivative is $1/x$ for the entropy function $x\log x$ and $1/x^{2}$ for the EL function $-\log x$, both of which decrease sharply as the weight increases. Consequently, these functions impose disproportionately heavy penalties on groups with near-zero weights compared to those with higher weights. In contrast, the constant curvature of the $L_{2}$ penalty ensures uniform penalization across all groups.

\begin{figure}[!htbp]
\centering
\includegraphics[width=\textwidth]{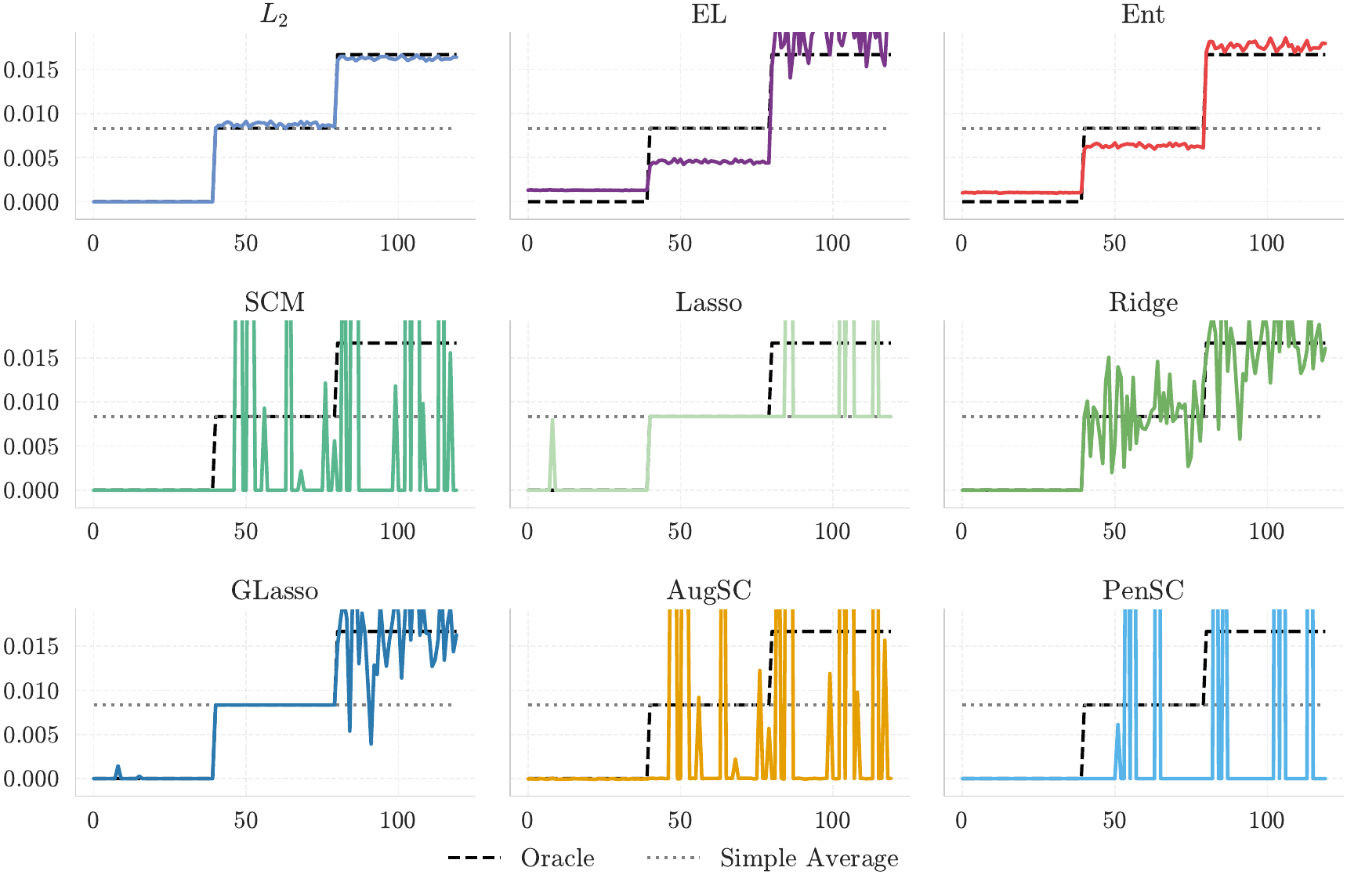}
\caption{Estimated weight vectors in a single replication (exact group structure, $J = 120$, $T_0 = 90$, $K = 3$, $r = 4$, correctly specified, oracle group weights $(0,1/3,2/3)$)}
\label{fig:weight-comparison-simu}
\end{figure}

Under misspecification, $L_{2}$-SCM-relaxation is the most accurate method in \cref{tab:rmse_truth_misspec}, followed by entropy-SCM-relaxation. It reduces RMSE\textsuperscript{$\star$} relative to SCM by 1.7\% to 4.1\%. The gains are smaller than under correct specification because RMSE\textsuperscript{$\star$} contains the unpredictable variance.

\begin{table}[!htbp]
\centering
\small
\caption{RMSE\textsuperscript{$\star$} (exact group, misspecified)}
\label{tab:rmse_truth_misspec}
\smallskip
\sbox{\tabbox}{\resizebox{\textwidth}{!}{\hdreset\hdbump{$L_2$}\hdbump{EL}\hdbump{Entropy}\hdbump{SCM}\hdbump{Lasso}\hdbump{Ridge}\hdbump{GLasso}\hdbump{fsPDA}\hdbump{AugSC}\hdbump{PenSC}\hdbump{SDID}\hdbump{GSC}\hdbump{\textbf{8.888}}\setlength{\tabcolsep}{4pt}%
\begin{tabular}{rrS[table-format=1.3]S[table-format=1.3]S[table-format=1.3]S[table-format=1.3]S[table-format=1.3]S[table-format=1.3]S[table-format=1.3]S[table-format=1.3]S[table-format=1.3]S[table-format=1.3]S[table-format=1.3]S[table-format=1.3]}
\toprule
\multicolumn{2}{c}{} & \multicolumn{3}{c}{Relaxation} & \multicolumn{9}{c}{} \\
\cmidrule(lr){3-5}
$J$ & $T_0$ & \hd{$L_2$} & \hd{EL} & \hd{Entropy} & \hd{SCM} & \hd{Lasso} & \hd{Ridge} & \hd{GLasso} & \hd{fsPDA} & \hd{AugSC} & \hd{PenSC} & \hd{SDID} & \hd{GSC} \\
\midrule
\multicolumn{14}{c}{\footnotesize\textbf{Panel A. $\bm{K<r}$}} \\
40 & 30 & \textbf{1.162} & 1.185 & 1.166 & 1.212 & 1.197 & 1.182 & 1.181 & 1.432 & 1.227 & 1.331 & 1.338 & 1.190 \\
 & 60 & \textbf{1.139} & 1.160 & 1.143 & 1.178 & 1.165 & 1.153 & 1.153 & 1.268 & 1.178 & 1.253 & 1.299 & 1.154 \\
 & 90 & \textbf{1.134} & 1.148 & 1.135 & 1.160 & 1.151 & 1.141 & 1.141 & 1.223 & 1.160 & 1.210 & 1.288 & 1.143 \\
\addlinespace[.45em]
80 & 30 & \textbf{1.150} & 1.176 & 1.154 & 1.197 & 1.183 & 1.166 & 1.166 & 1.530 & 1.200 & 1.325 & 1.226 & 1.184 \\
 & 60 & \textbf{1.124} & 1.145 & 1.126 & 1.162 & 1.150 & 1.138 & 1.138 & 1.275 & 1.164 & 1.246 & 1.194 & 1.136 \\
 & 90 & \textbf{1.118} & 1.135 & 1.120 & 1.147 & 1.136 & 1.126 & 1.126 & 1.218 & 1.147 & 1.192 & 1.176 & 1.127 \\
\addlinespace[.45em]
120 & 30 & \textbf{1.139} & 1.170 & 1.143 & 1.188 & 1.176 & 1.156 & 1.155 & 1.593 & 1.189 & 1.297 & 1.186 & 1.168 \\
 & 60 & \textbf{1.122} & 1.142 & 1.124 & 1.156 & 1.144 & 1.131 & 1.132 & 1.284 & 1.157 & 1.222 & 1.156 & 1.135 \\
 & 90 & \textbf{1.119} & 1.138 & 1.122 & 1.148 & 1.139 & 1.128 & 1.127 & 1.229 & 1.148 & 1.199 & 1.150 & 1.126 \\
\midrule
\multicolumn{14}{c}{\footnotesize\textbf{Panel B. $\bm{K=r}$}} \\
40 & 30 & \textbf{1.187} & 1.200 & 1.188 & 1.226 & 1.213 & 1.198 & 1.200 & 1.450 & 1.245 & 1.366 & 1.332 & 1.212 \\
 & 60 & \textbf{1.148} & 1.161 & 1.150 & 1.178 & 1.167 & 1.158 & 1.158 & 1.268 & 1.179 & 1.263 & 1.281 & 1.156 \\
 & 90 & \textbf{1.136} & 1.147 & 1.138 & 1.163 & 1.153 & 1.143 & 1.145 & 1.220 & 1.163 & 1.214 & 1.267 & 1.144 \\
\addlinespace[.45em]
80 & 30 & \textbf{1.176} & 1.190 & 1.177 & 1.222 & 1.206 & 1.189 & 1.191 & 1.572 & 1.226 & 1.348 & 1.230 & 1.207 \\
 & 60 & \textbf{1.136} & 1.146 & 1.137 & 1.167 & 1.154 & 1.143 & 1.144 & 1.276 & 1.170 & 1.240 & 1.189 & 1.146 \\
 & 90 & \textbf{1.129} & 1.138 & 1.130 & 1.154 & 1.142 & 1.134 & 1.134 & 1.220 & 1.155 & 1.198 & 1.183 & 1.136 \\
\addlinespace[.45em]
120 & 30 & \textbf{1.178} & 1.196 & 1.180 & 1.217 & 1.206 & 1.187 & 1.187 & 1.623 & 1.219 & 1.337 & 1.205 & 1.198 \\
 & 60 & \textbf{1.132} & 1.145 & 1.134 & 1.160 & 1.150 & 1.137 & 1.137 & 1.285 & 1.161 & 1.234 & 1.156 & 1.140 \\
 & 90 & \textbf{1.123} & 1.133 & 1.124 & 1.149 & 1.138 & 1.129 & 1.129 & 1.235 & 1.150 & 1.197 & 1.144 & 1.130 \\
\midrule
\multicolumn{14}{c}{\footnotesize\textbf{Panel C. $\bm{K>r}$}} \\
40 & 30 & \textbf{1.196} & 1.205 & 1.197 & 1.240 & 1.218 & 1.207 & 1.208 & 1.441 & 1.255 & 1.376 & 1.296 & 1.224 \\
 & 60 & \textbf{1.153} & 1.159 & 1.154 & 1.183 & 1.168 & 1.160 & 1.161 & 1.262 & 1.184 & 1.265 & 1.259 & 1.161 \\
 & 90 & \textbf{1.144} & 1.149 & 1.145 & 1.167 & 1.156 & 1.150 & 1.151 & 1.215 & 1.168 & 1.224 & 1.252 & 1.150 \\
\addlinespace[.45em]
80 & 30 & \textbf{1.178} & 1.189 & 1.179 & 1.221 & 1.202 & 1.189 & 1.189 & 1.541 & 1.225 & 1.364 & 1.233 & 1.202 \\
 & 60 & 1.137 & 1.143 & \textbf{1.137} & 1.166 & 1.152 & 1.143 & 1.144 & 1.262 & 1.170 & 1.234 & 1.173 & 1.147 \\
 & 90 & \textbf{1.135} & 1.140 & 1.136 & 1.155 & 1.144 & 1.139 & 1.140 & 1.217 & 1.156 & 1.211 & 1.181 & 1.141 \\
\addlinespace[.45em]
120 & 30 & \textbf{1.184} & 1.192 & 1.185 & 1.220 & 1.201 & 1.188 & 1.191 & 1.586 & 1.222 & 1.357 & 1.202 & 1.207 \\
 & 60 & \textbf{1.142} & 1.148 & 1.143 & 1.170 & 1.156 & 1.147 & 1.148 & 1.287 & 1.171 & 1.245 & 1.156 & 1.150 \\
 & 90 & \textbf{1.125} & 1.133 & 1.126 & 1.149 & 1.137 & 1.129 & 1.129 & 1.220 & 1.149 & 1.198 & 1.140 & 1.130 \\
\bottomrule
\end{tabular}
}}\usebox{\tabbox}
\par\vspace{.35\baselineskip}
\begin{minipage}{\wd\tabbox}\footnotesize
\textit{Notes}: RMSE\textsuperscript{$\star$} measures the out-of-sample prediction risk. Every entry is at least $1$ because of the unpredictable variance.
\end{minipage}
\end{table}

\phantomsection\label{para:simulation-inference}
To verify that the relaxation methods can yield valid $t$-test CIs, we present the CPs in the correctly specified exact group design in \cref{tab:cp}, and the absolute bias and the ratio of the average CI length to that of SCM in \cref{tab:att_bias,tab:ci_ratio}. The ATT estimators of all methods are nearly unbiased, and their CPs are close to the nominal level of 90\%. The bias and CPs in the approximate group and general factor designs are similar, and so are the CPs under misspecification. We omit them for simplicity. In terms of CI length, $L_{2}$-SCM-relaxation yields the shortest CI in most settings, followed by entropy-SCM-relaxation, Ridge, and Group Lasso. The CIs of EL, Lasso, AugSC, and GSC are close to those of SCM, while PenSC, SDID, and fsPDA give longer CIs than SCM in all settings, as the standard error of $\widehat{\tau}_{\text{bc}}$ is computed from the spread of the $\widehat{\tau}_{k}$, which reflects the prediction errors in the held-out blocks.

\begin{table}[!htbp]
\centering
\small
\caption{Coverage probability (\%) of the nominal 90\% CI (exact group)}
\label{tab:cp}
\smallskip
\sbox{\tabbox}{\resizebox{\textwidth}{!}{\hdreset\hdbump{$L_2$}\hdbump{EL}\hdbump{Entropy}\hdbump{SCM}\hdbump{Lasso}\hdbump{Ridge}\hdbump{GLasso}\hdbump{fsPDA}\hdbump{AugSC}\hdbump{PenSC}\hdbump{SDID}\hdbump{GSC}\hdbump{\textbf{88.8}}\setlength{\tabcolsep}{4pt}%
\begin{tabular}{rrS[table-format=2.1]S[table-format=2.1]S[table-format=2.1]S[table-format=2.1]S[table-format=2.1]S[table-format=2.1]S[table-format=2.1]S[table-format=2.1]S[table-format=2.1]S[table-format=2.1]S[table-format=2.1]S[table-format=2.1]}
\toprule
\multicolumn{2}{c}{} & \multicolumn{3}{c}{Relaxation} & \multicolumn{9}{c}{} \\
\cmidrule(lr){3-5}
$J$ & $T_0$ & \hd{$L_2$} & \hd{EL} & \hd{Entropy} & \hd{SCM} & \hd{Lasso} & \hd{Ridge} & \hd{GLasso} & \hd{fsPDA} & \hd{AugSC} & \hd{PenSC} & \hd{SDID} & \hd{GSC} \\
\midrule
\multicolumn{14}{c}{\footnotesize\textbf{Panel A. $\bm{K<r}$}} \\
40 & 30 & 88.1 & 88.3 & 88.2 & 89.1 & 88.9 & 89.2 & 89.8 & 89.7 & 89.3 & 89.2 & 86.7 & 90.1 \\
 & 60 & 91.6 & 91.4 & 91.6 & 90.5 & 91.6 & 90.5 & 90.3 & 91.4 & 89.7 & 90.2 & 89.9 & 90.2 \\
 & 90 & 90.2 & 91.3 & 90.5 & 91.2 & 91.7 & 91.5 & 91.0 & 92.7 & 90.9 & 92.4 & 90.4 & 90.3 \\
\addlinespace[.45em]
80 & 30 & 90.4 & 89.6 & 89.9 & 89.7 & 89.6 & 89.3 & 89.1 & 90.8 & 89.7 & 90.3 & 88.5 & 89.1 \\
 & 60 & 90.5 & 89.7 & 90.6 & 90.5 & 90.3 & 90.9 & 91.2 & 92.7 & 90.1 & 90.7 & 89.4 & 91.1 \\
 & 90 & 88.5 & 88.9 & 88.7 & 90.2 & 88.9 & 89.3 & 88.8 & 90.2 & 90.1 & 90.2 & 87.4 & 88.4 \\
\addlinespace[.45em]
120 & 30 & 91.3 & 91.0 & 91.5 & 90.9 & 91.6 & 91.8 & 91.9 & 92.1 & 91.0 & 90.5 & 90.2 & 90.2 \\
 & 60 & 90.9 & 91.9 & 91.6 & 92.5 & 91.6 & 92.2 & 92.3 & 93.2 & 92.3 & 93.2 & 90.8 & 91.0 \\
 & 90 & 90.7 & 91.1 & 90.8 & 90.5 & 90.6 & 90.2 & 90.0 & 92.6 & 90.9 & 91.0 & 90.3 & 91.2 \\
\midrule
\multicolumn{14}{c}{\footnotesize\textbf{Panel B. $\bm{K=r}$}} \\
40 & 30 & 91.3 & 90.6 & 90.6 & 91.6 & 91.0 & 90.3 & 90.7 & 91.0 & 91.3 & 92.2 & 88.2 & 91.9 \\
 & 60 & 89.0 & 88.8 & 88.8 & 88.6 & 89.8 & 89.1 & 89.2 & 91.5 & 89.0 & 90.8 & 89.3 & 91.2 \\
 & 90 & 88.4 & 89.5 & 88.2 & 89.9 & 89.1 & 89.5 & 89.4 & 91.6 & 89.7 & 92.2 & 89.7 & 90.3 \\
\addlinespace[.45em]
80 & 30 & 88.9 & 89.3 & 89.4 & 88.8 & 89.8 & 89.9 & 89.9 & 91.3 & 88.6 & 90.0 & 87.6 & 88.9 \\
 & 60 & 89.7 & 89.0 & 88.7 & 89.8 & 88.2 & 89.4 & 89.5 & 91.5 & 90.0 & 91.2 & 89.3 & 89.8 \\
 & 90 & 92.1 & 91.6 & 91.9 & 91.2 & 91.8 & 91.9 & 91.9 & 91.7 & 91.4 & 92.9 & 91.3 & 91.4 \\
\addlinespace[.45em]
120 & 30 & 92.7 & 93.3 & 93.2 & 94.1 & 93.5 & 93.8 & 93.7 & 92.4 & 93.7 & 91.7 & 91.6 & 92.6 \\
 & 60 & 89.8 & 91.4 & 90.1 & 91.0 & 90.6 & 90.0 & 90.2 & 91.8 & 90.9 & 91.4 & 90.9 & 90.0 \\
 & 90 & 88.7 & 88.8 & 89.1 & 88.0 & 89.5 & 88.8 & 89.1 & 91.6 & 87.9 & 90.4 & 88.7 & 88.8 \\
\midrule
\multicolumn{14}{c}{\footnotesize\textbf{Panel C. $\bm{K>r}$}} \\
40 & 30 & 91.1 & 90.9 & 91.2 & 90.6 & 91.1 & 91.0 & 91.3 & 92.3 & 90.7 & 91.8 & 90.1 & 92.1 \\
 & 60 & 91.9 & 90.4 & 91.4 & 90.5 & 90.7 & 90.2 & 90.3 & 90.5 & 90.6 & 92.2 & 89.0 & 90.6 \\
 & 90 & 90.2 & 90.8 & 90.0 & 91.1 & 91.2 & 90.1 & 91.0 & 90.7 & 91.1 & 90.7 & 89.1 & 90.6 \\
\addlinespace[.45em]
80 & 30 & 88.6 & 89.2 & 89.1 & 88.8 & 88.0 & 89.0 & 88.6 & 91.6 & 89.1 & 89.9 & 88.8 & 89.5 \\
 & 60 & 89.2 & 88.4 & 88.5 & 88.7 & 88.9 & 88.2 & 88.1 & 91.4 & 88.9 & 92.0 & 89.3 & 89.8 \\
 & 90 & 89.3 & 90.1 & 89.8 & 89.9 & 90.4 & 90.2 & 89.7 & 92.4 & 89.7 & 91.4 & 88.4 & 90.1 \\
\addlinespace[.45em]
120 & 30 & 89.9 & 90.0 & 90.5 & 88.5 & 88.8 & 88.6 & 89.6 & 90.6 & 88.2 & 90.7 & 88.5 & 87.8 \\
 & 60 & 89.8 & 89.9 & 90.4 & 90.7 & 90.6 & 89.7 & 89.8 & 92.0 & 90.8 & 90.5 & 89.9 & 90.4 \\
 & 90 & 90.1 & 89.6 & 89.4 & 90.1 & 90.2 & 90.0 & 89.8 & 89.9 & 90.5 & 90.7 & 89.4 & 89.4 \\
\bottomrule
\end{tabular}
}}\usebox{\tabbox}
\par\vspace{.35\baselineskip}
\begin{minipage}{\wd\tabbox}\footnotesize
\textit{Notes}: The CIs are constructed by the cross-fitting $t$-test with $B=5$ blocks. The intercept of fsPDA and the level shift of SDID cancel in each $\widehat{\tau}_{k}$, so only their weights on the donors enter the $t$-test.
\end{minipage}
\end{table}

\subsection{Approximate Group Structures}\label{subsec:MC-approx}

\begin{table}[!htbp]
\centering
\small
\caption{RMSE\textsuperscript{o} (approximate group)}
\label{tab:rmse_oracle_approx}
\smallskip
\sbox{\tabbox}{\resizebox{\textwidth}{!}{\hdreset\hdbump{$L_2$}\hdbump{EL}\hdbump{Entropy}\hdbump{SCM}\hdbump{Lasso}\hdbump{Ridge}\hdbump{GLasso}\hdbump{fsPDA}\hdbump{AugSC}\hdbump{PenSC}\hdbump{SDID}\hdbump{GSC}\hdbump{\textbf{8.888}}\setlength{\tabcolsep}{4pt}%
\begin{tabular}{rrS[table-format=1.3]S[table-format=1.3]S[table-format=1.3]S[table-format=1.3]S[table-format=1.3]S[table-format=1.3]S[table-format=1.3]S[table-format=1.3]S[table-format=1.3]S[table-format=1.3]S[table-format=1.3]S[table-format=1.3]}
\toprule
\multicolumn{2}{c}{} & \multicolumn{3}{c}{Relaxation} & \multicolumn{9}{c}{} \\
\cmidrule(lr){3-5}
$J$ & $T_0$ & \hd{$L_2$} & \hd{EL} & \hd{Entropy} & \hd{SCM} & \hd{Lasso} & \hd{Ridge} & \hd{GLasso} & \hd{fsPDA} & \hd{AugSC} & \hd{PenSC} & \hd{SDID} & \hd{GSC} \\
\midrule
\multicolumn{14}{c}{\footnotesize\textbf{Panel A. $\bm{K<r}$}} \\
40 & 30 & \textbf{.356} & .427 & .370 & .500 & .458 & .408 & .406 & .909 & .537 & .732 & .788 & .425 \\
 & 60 & \textbf{.269} & .342 & .278 & .405 & .354 & .313 & .313 & .646 & .406 & .568 & .725 & .316 \\
 & 90 & \textbf{.231} & .303 & .242 & .358 & .311 & .265 & .267 & .566 & .358 & .469 & .693 & .271 \\
\addlinespace[.45em]
80 & 30 & \textbf{.327} & .399 & .338 & .485 & .435 & .379 & .376 & 1.016 & .490 & .701 & .595 & .389 \\
 & 60 & \textbf{.249} & .312 & .257 & .392 & .339 & .292 & .289 & .668 & .400 & .548 & .515 & .285 \\
 & 90 & \textbf{.215} & .282 & .223 & .352 & .300 & .250 & .250 & .580 & .353 & .473 & .492 & .248 \\
\addlinespace[.45em]
120 & 30 & \textbf{.340} & .410 & .352 & .495 & .447 & .391 & .390 & 1.071 & .497 & .682 & .528 & .400 \\
 & 60 & \textbf{.253} & .313 & .259 & .396 & .348 & .295 & .293 & .682 & .397 & .534 & .443 & .286 \\
 & 90 & \textbf{.209} & .275 & .216 & .351 & .301 & .248 & .248 & .576 & .353 & .468 & .419 & .237 \\
\midrule
\multicolumn{14}{c}{\footnotesize\textbf{Panel B. $\bm{K=r}$}} \\
40 & 30 & \textbf{.394} & .435 & .402 & .515 & .470 & .436 & .435 & .910 & .538 & .742 & .763 & .461 \\
 & 60 & \textbf{.297} & .338 & .302 & .415 & .368 & .334 & .335 & .637 & .417 & .574 & .697 & .338 \\
 & 90 & \textbf{.260} & .304 & .267 & .367 & .327 & .290 & .290 & .558 & .368 & .490 & .671 & .294 \\
\addlinespace[.45em]
80 & 30 & \textbf{.388} & .420 & .391 & .510 & .465 & .423 & .425 & 1.020 & .514 & .733 & .595 & .440 \\
 & 60 & \textbf{.282} & .319 & .286 & .408 & .353 & .313 & .313 & .658 & .412 & .562 & .529 & .312 \\
 & 90 & \textbf{.235} & .274 & .240 & .358 & .306 & .268 & .268 & .558 & .359 & .484 & .496 & .262 \\
\addlinespace[.45em]
120 & 30 & \textbf{.374} & .417 & .381 & .510 & .460 & .410 & .414 & 1.076 & .512 & .704 & .532 & .434 \\
 & 60 & \textbf{.283} & .322 & .288 & .403 & .351 & .306 & .307 & .676 & .405 & .551 & .447 & .304 \\
 & 90 & \textbf{.230} & .266 & .235 & .355 & .303 & .262 & .261 & .572 & .356 & .467 & .403 & .249 \\
\midrule
\multicolumn{14}{c}{\footnotesize\textbf{Panel C. $\bm{K>r}$}} \\
40 & 30 & \textbf{.418} & .445 & .422 & .530 & .474 & .444 & .445 & .892 & .572 & .762 & .718 & .477 \\
 & 60 & \textbf{.311} & .336 & .314 & .420 & .369 & .339 & .340 & .620 & .422 & .574 & .665 & .346 \\
 & 90 & \textbf{.270} & .299 & .273 & .371 & .326 & .297 & .297 & .540 & .372 & .496 & .628 & .301 \\
\addlinespace[.45em]
80 & 30 & \textbf{.401} & .427 & .404 & .521 & .461 & .427 & .425 & .973 & .527 & .735 & .577 & .447 \\
 & 60 & \textbf{.303} & .325 & .304 & .416 & .360 & .330 & .331 & .633 & .425 & .567 & .499 & .327 \\
 & 90 & \textbf{.254} & .276 & .256 & .366 & .310 & .280 & .281 & .543 & .366 & .477 & .467 & .271 \\
\addlinespace[.45em]
120 & 30 & \textbf{.397} & .419 & .398 & .511 & .455 & .421 & .420 & 1.037 & .513 & .756 & .517 & .436 \\
 & 60 & \textbf{.291} & .315 & .292 & .408 & .353 & .317 & .318 & .653 & .410 & .556 & .429 & .306 \\
 & 90 & \textbf{.243} & .267 & .246 & .357 & .302 & .269 & .271 & .546 & .360 & .474 & .405 & .254 \\
\bottomrule
\end{tabular}
}}\usebox{\tabbox}
\par\vspace{.35\baselineskip}
\begin{minipage}{\wd\tabbox}\footnotesize
\textit{Notes}: See the notes in \cref{tab:rmse_oracle_exact}.
\end{minipage}
\end{table}

In this subsection, we consider an approximate group design to check the robustness of the estimation methods when the exact group structures do not hold. All the settings are the same as those in the previous subsection except that the loadings are generated as $\bm{\Lambda}=\bm{Z}\bm{\Lambda}^{\mathrm{co}}+\bm{\Xi}$, where $\bm{\Xi}$ consists of entries drawn independently from $U(-1/\sqrt{r},1/\sqrt{r})$. The patterns in \cref{tab:rmse_oracle_approx} are similar to those in \cref{tab:rmse_oracle_exact} under exact group structures, and $L_2$- and entropy-SCM-relaxation outperform all other methods in all settings. Under misspecification, $L_{2}$-SCM-relaxation again has the smallest RMSE\textsuperscript{$\star$} in all but one setting (\cref{tab:rmse_truth_approx_misspec}).

\subsection{General Factor Structure}\label{subsec:MC-general}

\begin{table}[!htbp]
\centering
\small
\caption{RMSE\textsuperscript{o} (general factor)}
\label{tab:rmse_oracle_general}
\smallskip
\sbox{\tabbox}{\resizebox{\textwidth}{!}{\hdreset\hdbump{$L_2$}\hdbump{EL}\hdbump{Entropy}\hdbump{SCM}\hdbump{Lasso}\hdbump{Ridge}\hdbump{fsPDA}\hdbump{AugSC}\hdbump{PenSC}\hdbump{SDID}\hdbump{GSC}\hdbump{\textbf{.888}}\setlength{\tabcolsep}{4pt}%
\begin{tabular}{rrS[table-format=0.3]S[table-format=0.3]S[table-format=0.3]S[table-format=0.3]S[table-format=0.3]S[table-format=0.3]S[table-format=0.3]S[table-format=0.3]S[table-format=0.3]S[table-format=0.3]S[table-format=0.3]}
\toprule
\multicolumn{2}{c}{} & \multicolumn{3}{c}{Relaxation} & \multicolumn{8}{c}{} \\
\cmidrule(lr){3-5}
$J$ & $T_0$ & \hd{$L_2$} & \hd{EL} & \hd{Entropy} & \hd{SCM} & \hd{Lasso} & \hd{Ridge} & \hd{fsPDA} & \hd{AugSC} & \hd{PenSC} & \hd{SDID} & \hd{GSC} \\
\midrule
40 & 30 & \textbf{.452} & .478 & .458 & .548 & .509 & .491 & .750 & .581 & .775 & .792 & .504 \\
 & 60 & \textbf{.307} & .330 & .311 & .412 & .363 & .343 & .532 & .414 & .581 & .711 & .349 \\
 & 90 & \textbf{.245} & .272 & .250 & .353 & .310 & .286 & .454 & .353 & .487 & .690 & .299 \\
\addlinespace[.45em]
80 & 30 & \textbf{.455} & .474 & .457 & .563 & .505 & .480 & .774 & .571 & .752 & .610 & .490 \\
 & 60 & \textbf{.293} & .317 & .300 & .414 & .358 & .335 & .519 & .418 & .574 & .509 & .314 \\
 & 90 & \textbf{.240} & .263 & .245 & .359 & .303 & .279 & .443 & .360 & .478 & .477 & .264 \\
\addlinespace[.45em]
120 & 30 & \textbf{.434} & .457 & .438 & .542 & .485 & .460 & .782 & .545 & .747 & .545 & .461 \\
 & 60 & \textbf{.301} & .323 & .306 & .415 & .359 & .335 & .525 & .418 & .567 & .433 & .316 \\
 & 90 & \textbf{.239} & .260 & .243 & .357 & .295 & .269 & .433 & .358 & .476 & .393 & .253 \\
\bottomrule
\end{tabular}
}}\usebox{\tabbox}
\par\vspace{.35\baselineskip}
\begin{minipage}{\wd\tabbox}\footnotesize
\textit{Notes}: See the notes in \cref{tab:rmse_oracle_exact}.
\end{minipage}
\end{table}

In this subsection, we consider a general factor design without group structures. All the settings are the same as those in the previous subsections except that the loading of each donor on each factor is drawn independently from $\mathcal{N}(0,9/r)$ and the loading of the treated unit is fixed at $\bm{\lambda}_0 = (0.5, -0.5, 0.5, -0.5)'$. Although our asymptotic theory relies on the group structure, $L_2$- and entropy-SCM-relaxation still outperform all other methods in all settings of \cref{tab:rmse_oracle_general}. $L_{2}$-SCM-relaxation reduces RMSE\textsuperscript{o} relative to SCM by a median of 27\%, and the median gain rises from 19\% at $T_0=30$ to 33\% at $T_0=90$. Under misspecification, it reduces RMSE\textsuperscript{$\star$} relative to SCM by 1.9\% to 2.9\% and has the smallest RMSE\textsuperscript{$\star$} in all but one setting (\cref{tab:rmse_truth_general_misspec}).

\medskip{}

Overall, $L_{2}$-SCM-relaxation stands out as the best performer and is recommended in practice. It has the smallest RMSE\textsuperscript{o} in almost all settings of the correctly specified designs and the smallest RMSE\textsuperscript{$\star$} in almost all settings under misspecification, and its CIs have coverage close to the nominal level and the shortest length in most settings. The Monte Carlo simulation evidence also shows that the relaxation methods are robust against mild deviation from group structure, manifesting their practical usefulness.

\section{Empirical Application }\label{sec:Empirical-applications}

In this section, we examine the impact of the 2016 Brexit referendum on the real GDP of the United Kingdom (UK). On June 23, 2016, a referendum on the UK's membership in the European Union (EU) resulted in 51.89\% of voters favoring departure, triggering a withdrawal process that would conclude on January 31, 2020. The referendum's outcome was largely unexpected by markets. We designate the third quarter of 2016 as the starting point of the treatment.

\begin{figure}
\centering
\includegraphics[width=1\textwidth]{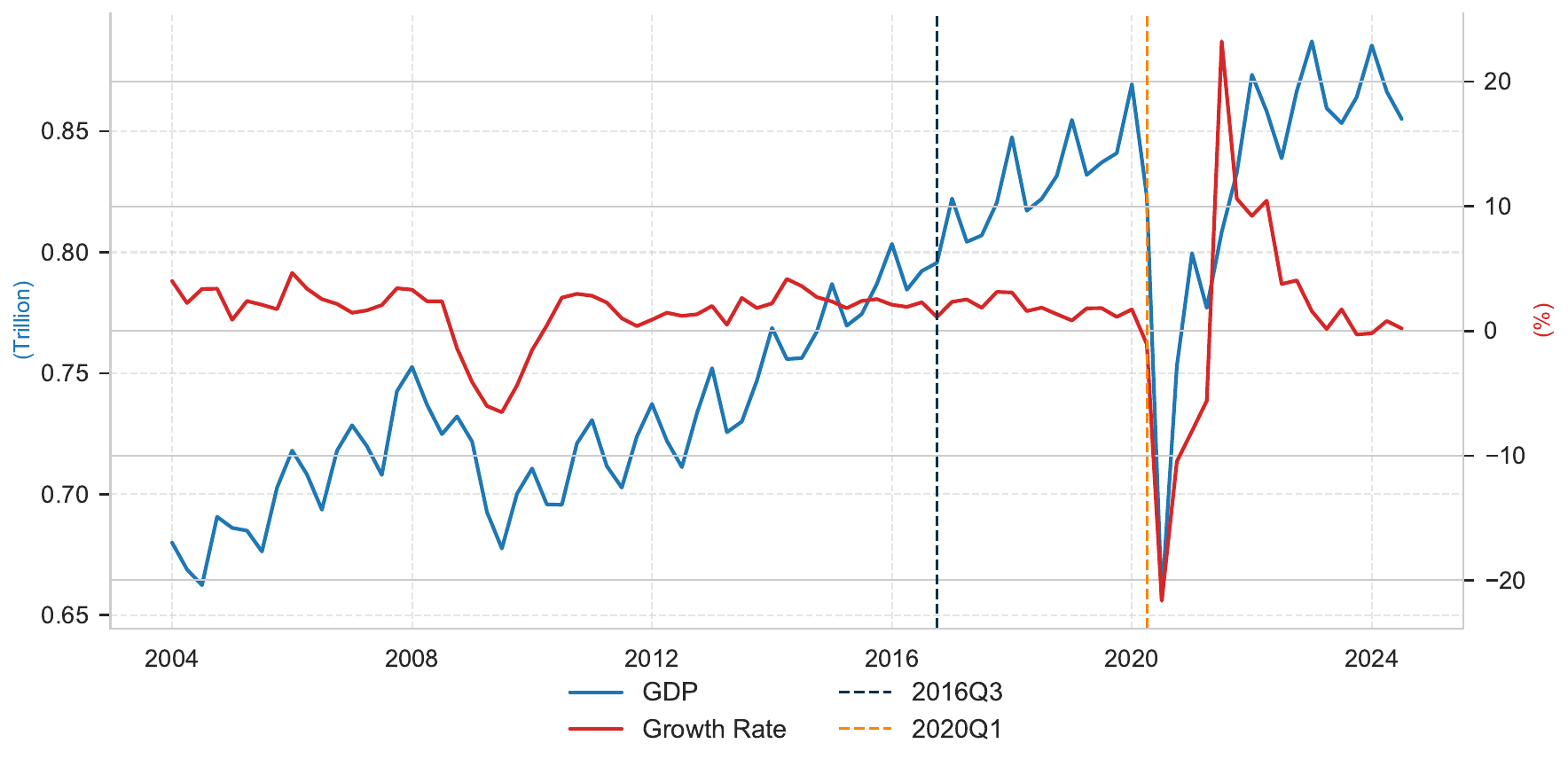}

\caption{Real GDP and growth rates of the UK}\label{fig:UK_GDP_GR}

\end{figure}

We collect quarterly real GDP data for all available economies from the database of CEIC, a subsidiary of Caixin providing business information. The completeness of the quarterly GDP time series varies across the countries and regions, with a fraction of them having records only in recent years. We shape the available data into a balanced panel, with a comprehensive donor pool of 57 economies from 2002Q4 to 2024Q2. To remove time trends and seasonality, we construct the year-over-year (YoY) GDP growth rate $\tilde{y}_{jt}=y_{jt}/y_{j(t-4)}-1$ for all economies. \Cref{fig:UK_GDP_GR} plots the UK's quarterly real GDP (blue line) and its YoY growth rate (red line).

\begin{figure}
\centering
\begin{raggedright}
\includegraphics[width=1\textwidth]{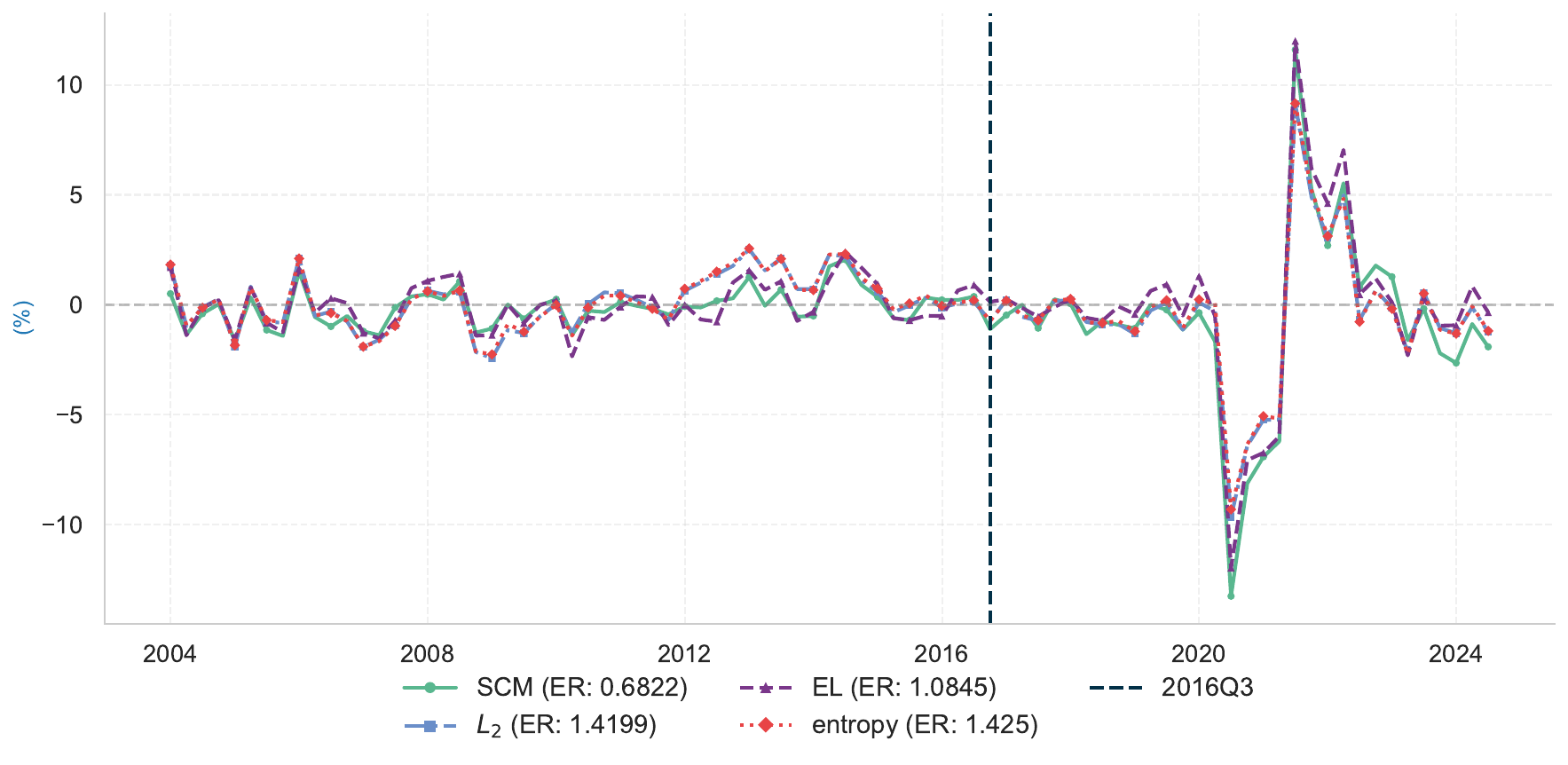}
\par\end{raggedright}
\begin{raggedright}
\footnotesize{Note: ER is the in-sample empirical risk.}
\par\end{raggedright}
\caption{Gap between realized growth rate and the fitted value (Treatment time
2016Q3) }\label{fig:Fitted-GDP-growth}
\end{figure}

\begin{figure}[ph]
\centering
\includegraphics[width=0.9\textwidth]{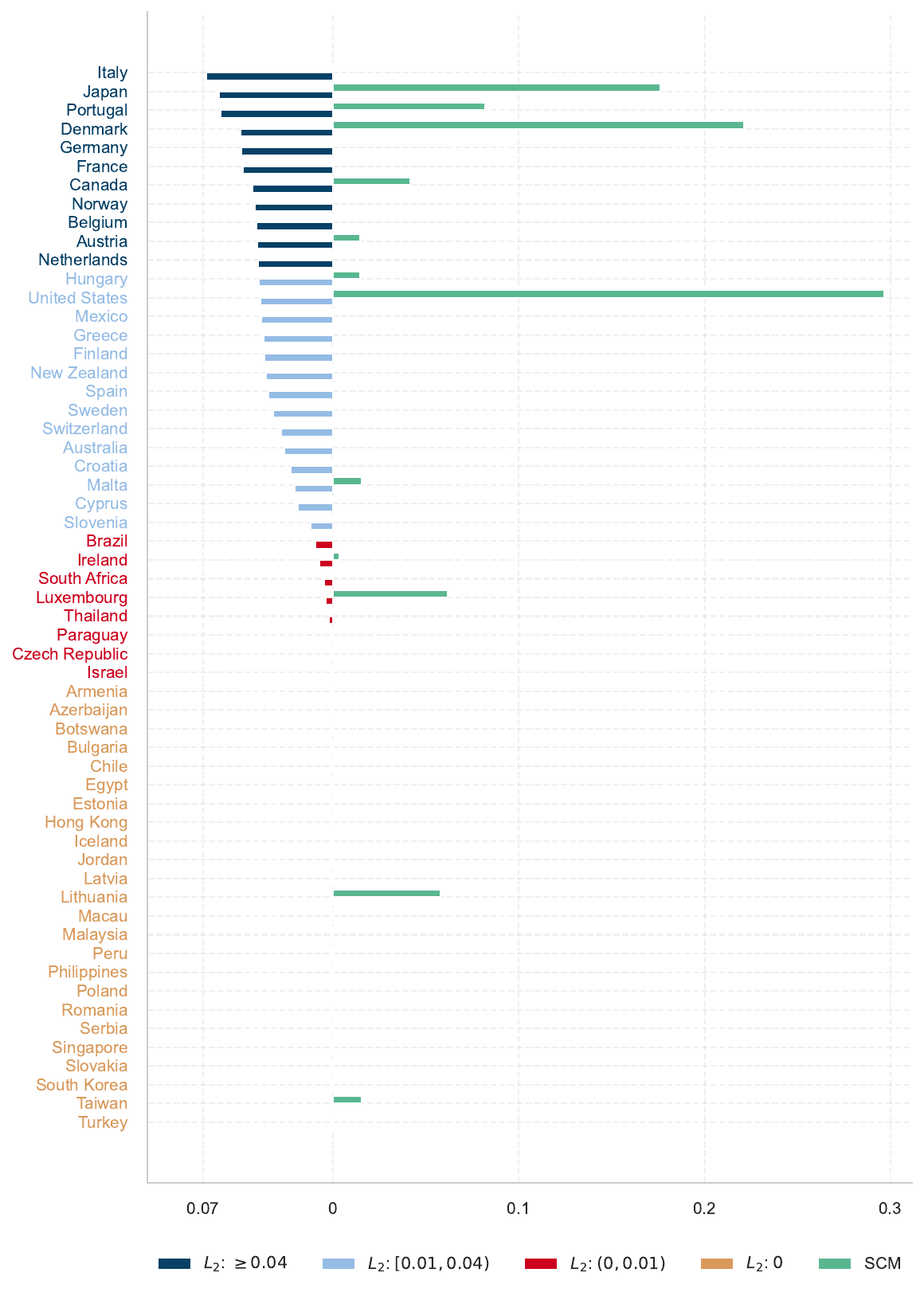}
\caption{$L_{2}$-SCM-relaxation weights versus SCM weights (Treatment time
2016Q3) }\label{fig:weight_L2_SCM_2016}
\end{figure}

\leavevmode\phantomsection\label{para:brexit-fit}We fit and predict the GDP growth rate by SCM and the three relaxation methods. As shown in \cref{fig:Fitted-GDP-growth}, all methods track the pre-treatment growth rate closely, with similar post-treatment paths. SCM achieves the smallest in-sample empirical risk by construction, while the relaxation approaches trade some in-sample precision for a more stable weight distribution.

Given that the counterfactuals are constructed by weighting the control units, we report in \cref{fig:weight_L2_SCM_2016} the weights from SCM and those of $L_{2}$-SCM-relaxation.\footnote{The weights from the entropy objective are similar to those of $L_{2}$ and are therefore omitted. The EL weights are dominated by Japan (42.5\%) and are also omitted for brevity.} We sort the country/region names by the size of the $L_{2}$ weights and \emph{manually} classify them into Group 1 ($\geq0.04$, dark blue), Group 2 ($[0.01,0.04)$, light blue), Group 3 ($(0,0.01)$, red), and Group 4 (exactly zero, brown). Group 1 includes major EU countries Italy, Germany, Denmark, and France, as well as Canada, Japan, Portugal, Norway, Belgium, Austria, and the Netherlands. They are important economic partners of the UK. Group 2 further covers several middle-sized EU economies.

Those in Groups 3 and 4 are mostly geographically distant from the UK. In contrast, the sparse weights produced by SCM are less interpretable. While the USA alone takes 30\% of the weight, the estimate excludes Germany, France, and Italy, the top three EU countries by GDP, but instead places non-trivial weights on Lithuania and Luxembourg, countries that play a minor economic role.

\begin{figure}
\centering
\includegraphics[width=1\textwidth]{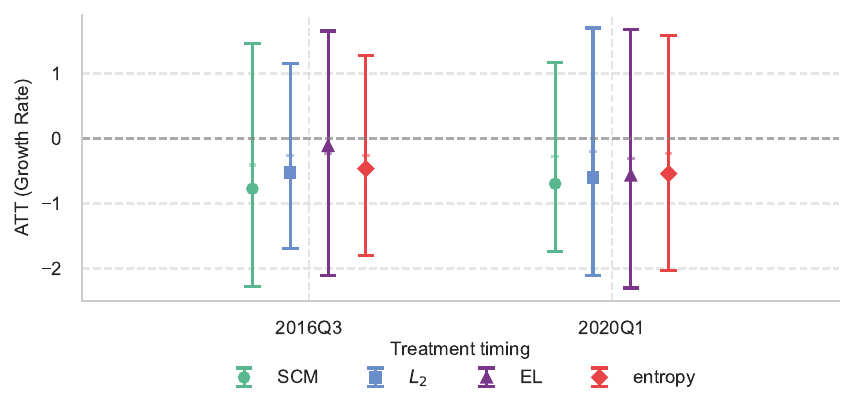}
\caption{Average treatment effects and confidence intervals}\label{fig:ATT}
\end{figure}

We report in \cref{fig:ATT} the 90\% CIs for the ATT via the cross-fitting method with $B=5$. The CI based on $L_{2}$-SCM-relaxation is the shortest for treatment time 2016Q3. The ATTs are
statistically insignificant for all methods, perhaps due to the limited sample size. It is more direct to assess the GDP loss in levels. We use the estimated GDP growth rate counterfactuals to reconstruct the counterfactual real GDP. Let $\widehat{\tilde{y}}_{0t}^{N}=\sum_{j\in[J]}\hat{w}_{j}\tilde{y}_{jt}$ in the post-treatment period, and we compute the level GDP as
\[
\hat{y}_{0t}^{N}=(1+\widehat{\tilde{y}}_{0t}^{N})z_{0(t-4)},\quad t\in\mathcal{T}_{1},
\]
where the base $z_{0t}=y_{0t}$ if $t\leq T_{0}$ and $z_{0t}=\hat{y}_{0t}^{N}$ for $t>T_{0}$.

\begin{figure}
\centering
\begin{centering}
\includegraphics[width=1\textwidth]{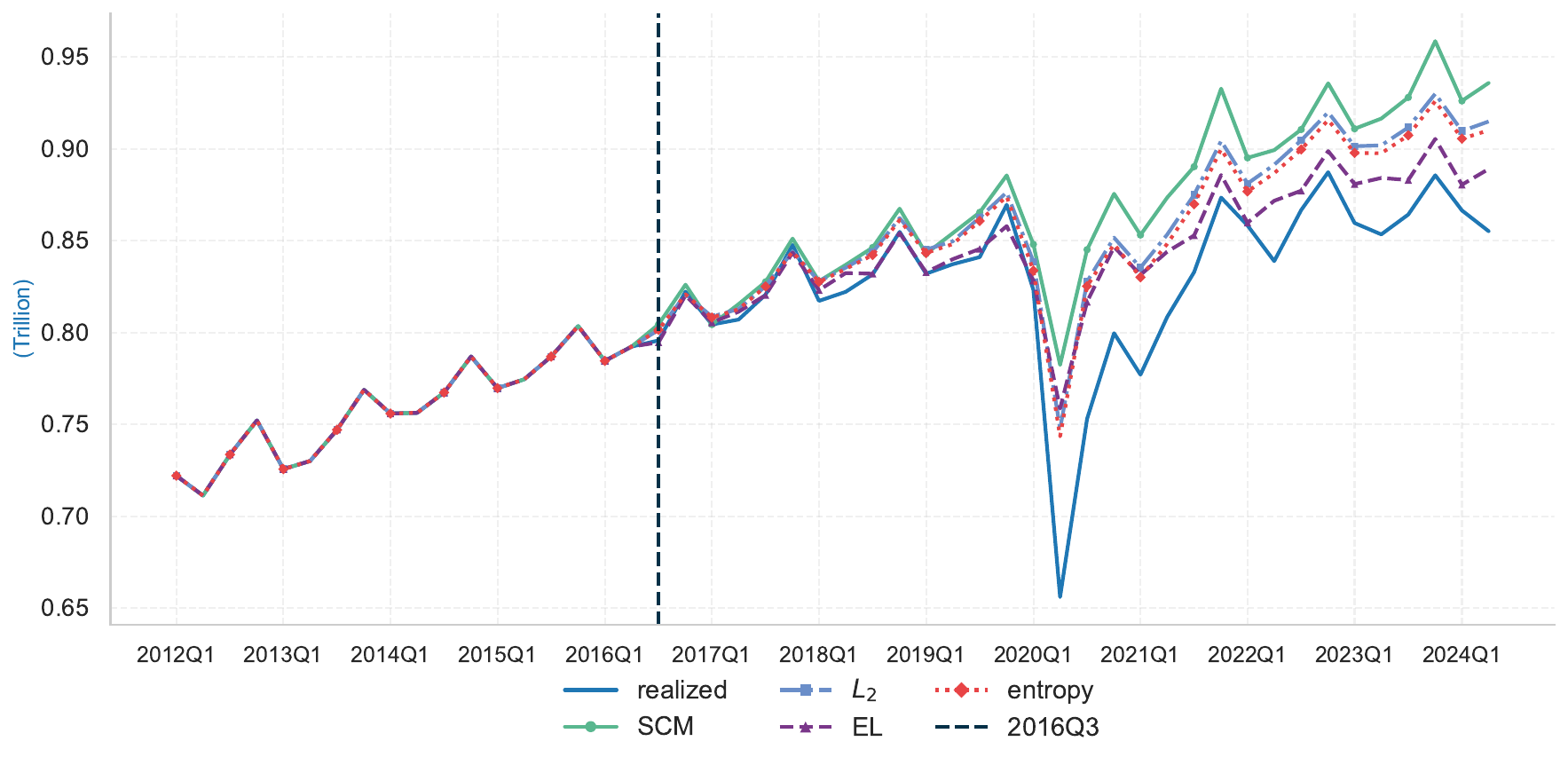}
\par\end{centering}
\caption{Real GDP versus counterfactual GDP (Treatment time 2016Q3) }\label{fig:(Log)-Difference}
\end{figure}

We calculate $\sum_{t\in\mathcal{T}_{1}}(y_{0t}^{I}-\hat{y}_{0t}^{N})/\sum_{t\in\mathcal{T}_{1}}y_{0t}^{I}$, the treatment effects relative to the UK's total GDP after the treatment. SCM predicts a substantial loss of $4.93\%$ of the total realized GDP. $L_{2}$-, EL- and entropy-SCM-relaxation estimate $3.59\%$, $2.05\%$ and $3.25\%$, respectively. All four estimates point to lower realized GDP relative to the constructed counterfactual. In \cref{fig:(Log)-Difference}, the gap is visible from 2016Q3 to 2020Q1 and widens after 2020Q1. This pattern is consistent with delayed economic costs becoming more pronounced after the formal exit.\footnote{The first quarter of the official departure coincided with the outbreak of Covid-19, which triggered a worldwide economic recession. The pandemic was a global event that affected all countries. Given that the UK's public health system was relatively well developed compared to other countries, the GDP gap could be even larger in the absence of the pandemic.} For the period after 2020Q1, SCM yields an estimated loss of $7.74\%$,\footnote{In an independent study released later than our paper, \citet{bloom2025economic} use SCM with a donor pool of 33 countries and find that Brexit reduced the UK's GDP by 6\% to 8\% by 2025.}
 whereas the relaxation method outcomes range from $3.58\%$ to $5.62$\%.

\begin{figure}
\centering
\begin{raggedright}
\includegraphics[width=1\textwidth]{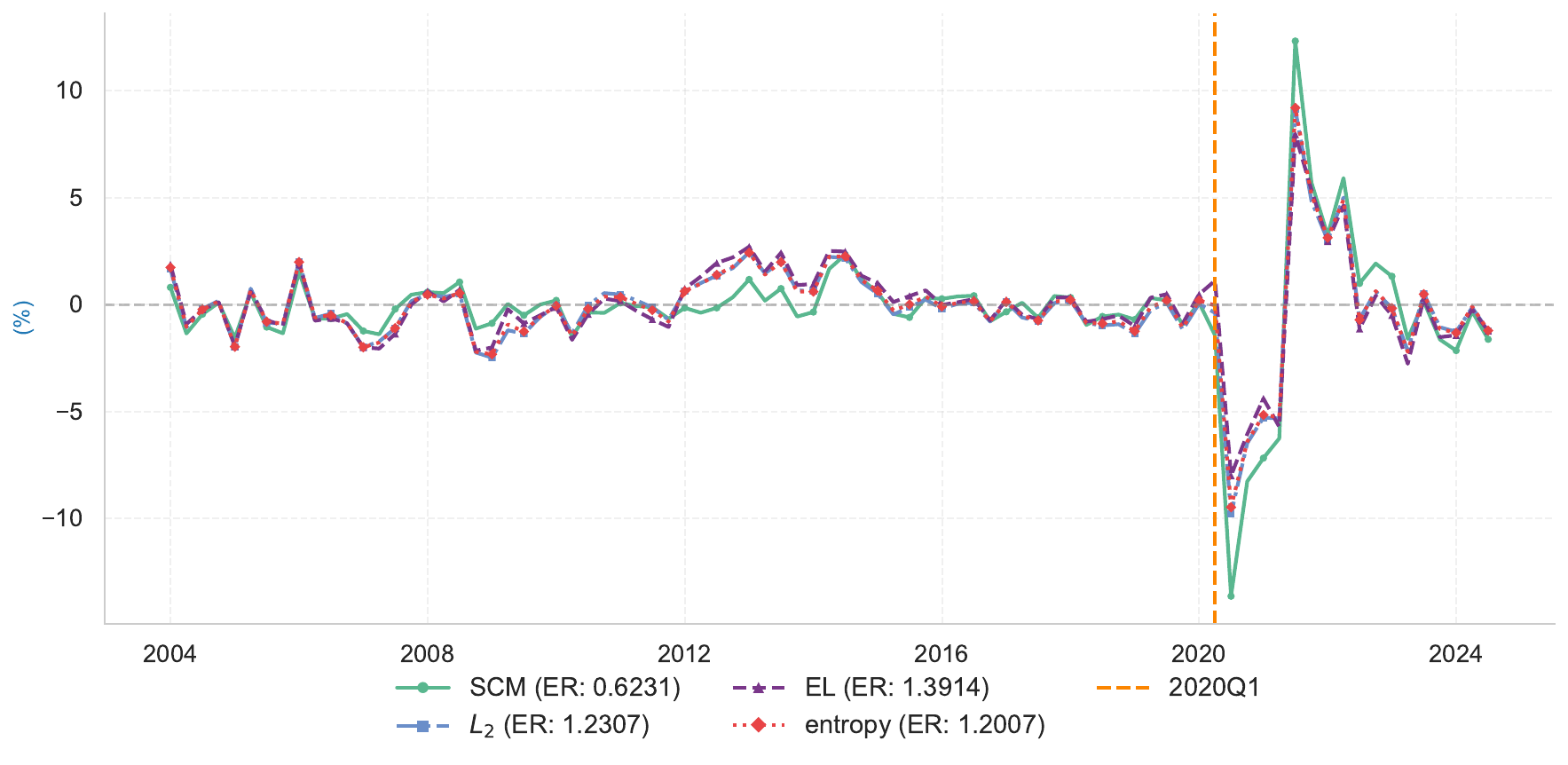}
\par\end{raggedright}
\begin{raggedright}
\footnotesize{Note: ER is the in-sample empirical risk.}
\par\end{raggedright}
\caption{Gap between realized growth rate and the fitted value (Treatment time
2020Q1)}\label{fig:Fitted-GDP-growth-1}
\end{figure}

\begin{figure}[ph]
\centering
\includegraphics[width=0.9\textwidth]{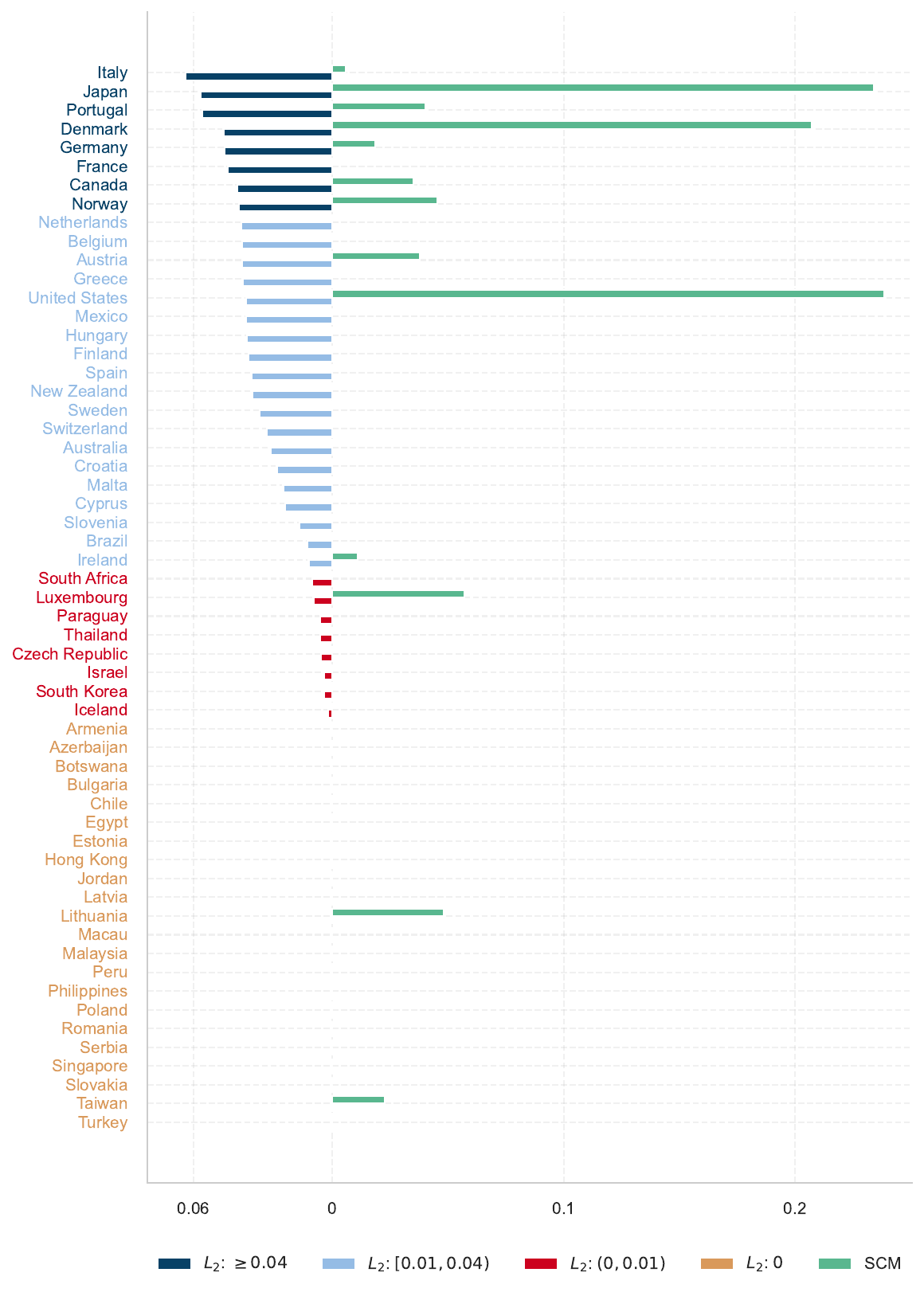}
\caption{$L_{2}$-SCM-relaxation weights versus SCM weights (Treatment time 2020Q1)}\label{fig:weight_L2_SCM_2020}
\end{figure}

\begin{figure}
\centering
\includegraphics[width=1\textwidth]{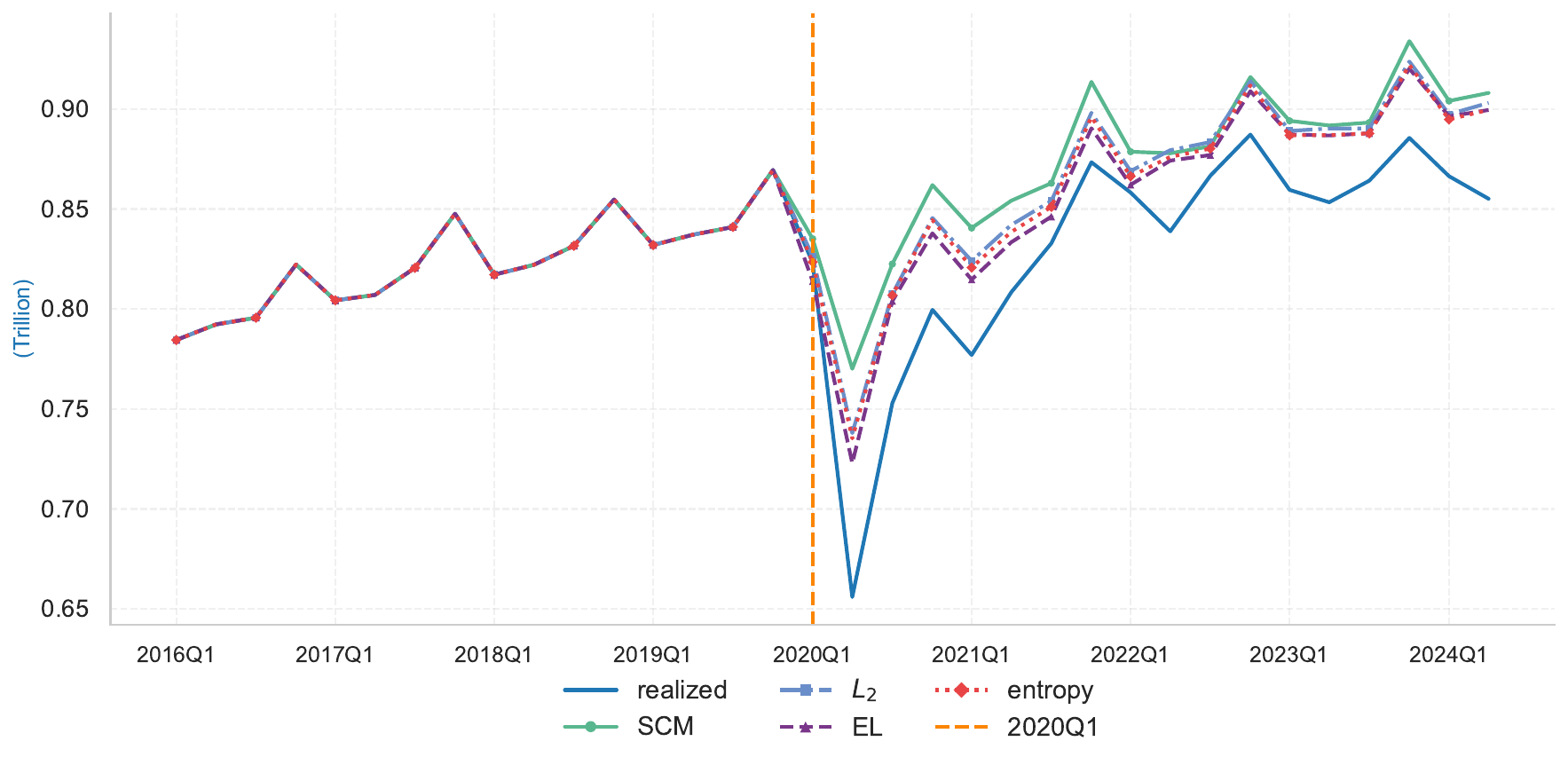}
\caption{Real GDP versus counterfactual GDP (Treatment time 2020Q1)}\label{fig:Difference-between-real}
\end{figure}

While the outcome of the referendum was largely unexpected, economic decoupling has been a work in progress after 2016, with 2020 anticipated as another relevant date. As a robustness check, we carry out re-estimation by setting the treatment starting point at 2020Q1. In \cref{fig:Fitted-GDP-growth-1}, the in-sample empirical risks are similar to those of \cref{fig:Fitted-GDP-growth}, and \cref{fig:weight_L2_SCM_2020} reports the corresponding weights.
The ATTs for growth rates are still insignificant, as demonstrated in the right half of \cref{fig:ATT}. For the loss of total GDP, SCM estimates $5.22\%$ of the UK's total GDP after 2020Q1, while $L_{2}$-, EL- and entropy-SCM-relaxation estimate $4.12\%$, $3.38\%$ and $3.81\%$, respectively. These enlarged effects are due to the removal of the ``mild loss'' period between 2016 and 2020 as shown in \cref{fig:Difference-between-real}. They are smaller than the numbers for the same time period reported at the end of the last paragraph. This is because this exercise uses the 2016--2020 data for in-sample fit, where the downward anticipation is leaked into the pre-treatment period and contaminates the weight estimation.
Therefore, the empirical results starting from the 2016Q3 treatment are cleaner. Overall, regardless of the choices of treatment point, the departure of the UK has inflicted a profound economic loss.

\section{Conclusion}

In this paper, we carry the stable-balancing framework over to the synthetic control setting, providing asymptotic guarantees for counterfactual prediction of a treated unit in the presence of many potential control units.
Given a factor model with latent group structures, $L_{2}$-SCM-relaxation produces weights that converge asymptotically to the oracle weights without sparsity or eigenvalue conditions on the design matrix. The convergence rate is sufficiently fast to ensure oracle prediction accuracy as if the group membership were known.  Among the information-theoretic objectives studied here, the $L_{2}$-norm has two practical advantages in this setting: it does not require the extra technical conditions needed by EL and entropy, and it admits exact zero weights, avoiding the boundary bias that EL and entropy incur when some true group weights are zero.

\bigskip \bigskip

\setlength{\bibsep}{7pt}
\bibliographystyle{chicago}
\bibliography{emplikeli,ref_books,scm}


\newpage\renewcommand{\theequation}{S.\arabic{equation}}
\renewcommand{\thethm}{S.\arabic{thm}}
\renewcommand{\thelem}{S.\arabic{lem}}
\renewcommand{\therem}{S.\arabic{rem}}
\renewcommand{\theassumption}{S.\arabic{assumption}}
\renewcommand{\thetable}{S.\arabic{table}}
\setcounter{table}{0}
\setcounter{equation}{0}
\setcounter{thm}{0}
\setcounter{lem}{0}
\setcounter{assumption}{0}
\setcounter{rem}{0}
\setcounter{page}{1}
\begin{center}
{\LARGE Online Appendices to\\[6pt] ``A Relaxation Approach to Synthetic Control''}\\[1.5em]
{\large Chengwang Liao\quad{}Zhentao Shi\quad{}Yapeng Zheng}\\[8pt]
{\itshape\large The Chinese University of Hong Kong}
\par\end{center}

\appendix
The online appendices have four sections. \Cref{sec:Proof-of-main} presents proofs of the main theoretical results. \Cref{sec:Standardized-SCM-relaxation} shows that the theory extends to the factor model with heterogeneous scales if the data are scale-standardized. \Cref{sec:appendix-simulation} includes additional simulation results. \Cref{sec:without-group} develops the theory for ATT consistency without imposing latent group structure.

\section{Proof of Main Results}\label{sec:Proof-of-main}

This appendix presents the proofs of the main results. \cref{subsec:Proof-of-Lemma} shows that the oracle weight is within-group equal for both $K\leq r$ and $K>r$, and \cref{subsec:Technical-Lemmas} collects essential technical lemmas. Leveraging these results, \cref{subsec:proof-thm1} establishes that the estimated weight of $L_{2}$-SCM-relaxation converges to the oracle weight at a sufficiently fast speed. This convergence speed enables us to derive in \cref{subsec:proof-thm2} the oracle inequalities for the empirical risks $R_{\mathcal{T}_{0}}(\hat{\bm{w}})$ and $R_{\mathcal{T}_{1}}(\hat{\bm{w}})$. \Cref{subsec:proof-thm3} extends the theoretical results to allow for generic objective functions. \cref{subsec:proof-att} proves the consistency of the ATT estimator in \cref{thm:att}. Finally, \cref{subsec:proof-asym-dist} proves the asymptotic distribution of the $t$-statistic in \cref{thm:asym-dist}.

\textbf{Additional notation.} ``w.p.a.1'' means ``with probability
approaching one.'' For a generic matrix $\bm{A}$, let $\bm{A}^{\dagger}$
denote its Moore-Penrose pseudo inverse. For two random variables
$x$ and $y$, we use $x\lesssim_{p}y$ as shorthand for $x\leq By$
for some absolute constant $B$ w.p.a.1., and $x\gtrsim_{p}y$ means
$y\lesssim_{p}x$. For an $m\times n$ matrix $\bm{A}$, its $(i,j)$-th
entry is denoted as $A_{ij}$, and let $\left\Vert \bm{A}\right\Vert _{p}=\sup_{\bm{x}\neq\boldsymbol{0}}\left\Vert \bm{A}\bm{x}\right\Vert _{p}/\left\Vert \bm{x}\right\Vert _{p}$
be its $L_{p}$-norm. In particular, when $p=1$, $\|\bm{A}\|_{1}=\max_{j\in[n]}\sum_{i\in[m]}|A_{i,j}|$
is the maximum absolute column sum. When $p=\infty$, $\left\Vert \bm{A}\right\Vert _{\infty}=\max_{i\in[m]}\sum_{j\in[n]}|A_{ij}|$
is the maximum absolute row sum. The entry-wise maximum norm of $\bm{A}$
is defined as $\|\bm{A}\|_{\max}=\max_{i\in[m],j\in[n]}|A_{ij}|$.\medskip

Throughout proofs, define the covariance matrices for the ``core''
component as $\hat{\bm{\Sigma}}^{\mathrm{co}}:=T_{0}^{-1}\bm{\Lambda}^{\mathrm{co}}\bm{F}'\bm{F}\bm{\Lambda}^{\mathrm{co}}{}'$
and $\hat{\bm{\Upsilon}}^{\mathrm{co}}:=T_{0}^{-1}\bm{\Lambda}^{\mathrm{co}}\bm{F}'\bm{F}\bm{\lambda}_{0}$.
It is clear that $\hat{\bm{\Sigma}}^{*}=\bm{Z}\hat{\bm{\Sigma}}^{\mathrm{co}}\bm{Z}'$
and $\hat{\bm{\Upsilon}}^{*}=\bm{Z}\hat{\bm{\Upsilon}}^{\mathrm{co}}$.

\subsection[Proof of Lemma]{Proof of \cref{lem:oracle_target_relaxation}}\label{subsec:Proof-of-Lemma}
\begin{proof}[Proof of \cref{lem:oracle_target_relaxation}]
First, the solution to \eqref{eq:oracle_relaxation_problem} always exists and is unique. Under \cref{assu:convex_hull}, there exists $\bm{\pi}\in\Delta_{K}$ such that $\bm{\Lambda}^{\mathrm{co}\prime}\bm{\pi}=\bm{\lambda}_{0}$. Then $\bm{w}=\bm{Z}(\bm{Z}'\bm{Z})^{-1}\bm{\pi}$ satisfies $w_{j}=\pi_{k}/J_{k}\geq0$ for $j\in\mathcal{G}_{k}$, $\bm{1}_{J}^{\prime}\bm{w}=1$, and $\hat{\bm{\Sigma}}^{*}\bm{w}=\hat{\bm{\Upsilon}}^{*}$, so it is feasible. The feasible set $S^{*}$ is thus non-empty. Moreover, $S^*$ is a compact convex set and the objective function is strictly convex and continuous. Hence, program \eqref{eq:oracle_relaxation_problem} is a convex optimization problem and has a unique solution.

We now argue that under the group structure \eqref{eq:grouped_structure}, the oracle weight $\bm{w}^*$ must be within-group equal. Let $\bm{w}_{\mathcal{G}}:=\bm{Z}^{\prime}\bm{w}=\bigl(\sum_{j\in\mathcal{G}_{1}}w_{j},\dots,\sum_{j\in\mathcal{G}_{K}}w_{j}\bigr)^{\prime}$ be the $K\times 1$ group sum weight vector. The balancing constraint $\hat{\bm{\Sigma}}^{*}\bm{w}-\hat{\bm{\Upsilon}}^{*}=\bm{0}_{J}$ can be written as $\bm{Z}(\hat{\bm{\Sigma}}^{\mathrm{co}}\bm{w}_{\mathcal{G}}-\hat{\bm{\Upsilon}}^{\mathrm{co}})=\bm{0}_{J}$. Since $\bm{Z}$ has full column rank, the constraint is equivalent to
\begin{equation}
\hat{\bm{\Sigma}}^{\mathrm{co}}\bm{w}_{\mathcal{G}}=\hat{\bm{\Upsilon}}^{\mathrm{co}}.\label{eq:cons_reduced}
\end{equation}
The objective function can be written as $\sum_{k=1}^{K}\sum_{j\in\mathcal{G}_{k}}w_{j}^{2}$. By the Cauchy-Schwarz inequality,
\[
\sum_{k=1}^{K}\sum_{j\in\mathcal{G}_{k}}w_{j}^{2} \geq \sum_{k=1}^{K}\frac{1}{J_k}\Biggl(\sum_{j\in\mathcal{G}_{k}}w_{j}\Biggr)^{2} = \bm{w}_{\mathcal{G}}^{\prime}(\bm{Z}^{\prime}\bm{Z})^{-1}\bm{w}_{\mathcal{G}},
\]
where equality obtains if and only if $w_{j}=w_{\mathcal{G}_{k}}/J_{k}$ for all $j\in\mathcal{G}_{k}$ for each $k=1,\dots,K$. The same argument applies to a general strictly convex objective $\sum_{j=1}^{J}g(w_j)$ by Jensen's inequality, with equality only when the weights are equal within each group. Hence, the oracle problem \eqref{eq:oracle_relaxation_problem} has the same optimal function value as the problem
\begin{equation}
\min_{\bm{w}_{\mathcal{G}}}{}\bm{w}_{\mathcal{G}}^{\prime}(\bm{Z}^{\prime}\bm{Z})^{-1}\bm{w}_{\mathcal{G}}\qquad\text{subject to \eqref{eq:cons_reduced}, }\bm{1}_{K}^{\prime}\bm{w}_{\mathcal{G}}=1\text{ and }w_{\mathcal{G}_{k}}\geq0.\label{eq:oracle_relaxation_problem_reduced}
\end{equation}
At the minimizer $\bm{w}_{\mathcal{G}}^{*}$, it holds that $w_j^*=w_{\mathcal{G}_{k}}^{*}/J_{k}$ for all $j\in\mathcal{G}_{k}$ for each $k=1,\dots,K$. Written compactly, $\bm{w}^{*}=\bm{Z}(\bm{Z}'\bm{Z})^{-1}\bm{w}_{\mathcal{G}}^{*}$.

It remains to find the closed-form solution for $\bm{w}_{\mathcal{G}}^{*}$. If $K\leq r$, $\hat{\bm{\Sigma}}^{\mathrm{co}}=T_{0}^{-1}\bm{\Lambda}^{\mathrm{co}}\bm{F}^{\prime}\bm{F}\bm{\Lambda}^{\mathrm{co}\prime}$ is nonsingular, and therefore \eqref{eq:cons_reduced} yields
\begin{equation}
\bm{w}_{\mathcal{G}}^{*}=(\hat{\bm{\Sigma}}^{\mathrm{co}})^{-1}\hat{\bm{\Upsilon}}^{\mathrm{co}}. \label{eq:w_G_oracle}
\end{equation}
It is also the unique solution to \eqref{eq:cons_reduced}.

If $K>r$, then $\hat{\bm{\Sigma}}^{\mathrm{co}}=T_{0}^{-1}\bm{\Lambda}^{\mathrm{co}}\bm{F}^{\prime}\bm{F}\bm{\Lambda}^{\mathrm{co}\prime}$ is singular, but $\hat{\bm{\Omega}}_{\bm{F}}=T_{0}^{-1}\bm{F}^{\prime}\bm{F}$ is nonsingular in general. By $\hat{\bm{\Sigma}}^{\mathrm{co}}=\bm{\Lambda}^{\mathrm{co}}\hat{\bm{\Omega}}_{\bm{F}}\bm{\Lambda}^{\mathrm{co}\prime}$ and $\hat{\bm{\Upsilon}}^{\mathrm{co}}=\bm{\Lambda}^{\mathrm{co}}\hat{\bm{\Omega}}_{\bm{F}}\bm{\lambda}_{0}$, \eqref{eq:cons_reduced} implies $\bm{\Lambda}^{\mathrm{co}}\hat{\bm{\Omega}}_{\bm{F}}(\bm{\Lambda}^{\mathrm{co}\prime}\bm{w}_{\mathcal{G}}-\bm{\lambda}_{0})=\bm{0}_{K}$. Since $\bm{\Lambda}^{\mathrm{co}}$ has full column rank $r$ and $\hat{\bm{\Omega}}_{\bm{F}}$ is nonsingular, the matrix $\bm{\Lambda}^{\mathrm{co}}\hat{\bm{\Omega}}_{\bm{F}}$ also has full column rank, hence
\begin{equation}
\bm{\Lambda}^{\mathrm{co}\prime}\bm{w}_{\mathcal{G}}=\bm{\lambda}_{0}.\label{eq:cons_linear}
\end{equation}
The problem \eqref{eq:oracle_relaxation_problem_reduced} can be further reduced to
\[
\min_{\bm{w}_{\mathcal{G}}}{}\frac{1}{2}\bm{w}_{\mathcal{G}}^{\prime}(\bm{Z}^{\prime}\bm{Z})^{-1}\bm{w}_{\mathcal{G}}\qquad\text{subject to \eqref{eq:cons_linear}, }\bm{1}_{K}^{\prime}\bm{w}_{\mathcal{G}}=1\text{ and }w_{\mathcal{G}_{k}}\geq0.
\]

Now suppose the solution is in the interior of the simplex, i.e., $w_{\mathcal{G}_{k}}>0$, so that the closed-form solution is available. Let $\bm{\mu}_{1}$ be the $r\times1$ Lagrangian multiplier vector associated with the balancing constraint \eqref{eq:cons_linear}, $\mu_{2}$ the Lagrangian multiplier associated with the sum-to-one constraint. Although the constraints contain $r+1$ equations, the number of linearly independent ones can be either $r$ or $r+1$, each leading to a solution.

\textbf{Case 1: $\bm{1}_{K}$ lies inside the column space of $\bm{\Lambda}^{\mathrm{co}}$}. It must hold that $\bm{\Lambda}^{\mathrm{co}}\bm{d}=\bm{1}_{K}$ for $\bm{d}=\bm{\Lambda}^{\mathrm{co}\dagger}\bm{1}_{K}$ (and $\bm{d}$ is the unique solution) where $\bm{\Lambda}^{\mathrm{co}\dagger}=(\bm{\Lambda}^{\mathrm{co}\prime}\bm{\Lambda}^{\mathrm{co}})^{-1}\bm{\Lambda}^{\mathrm{co}\prime}$ is the pseudo inverse of $\bm{\Lambda}^{\mathrm{co}}$. Then, $\bm{1}_{K}^{\prime}\bm{w}_{\mathcal{G}}=\bm{d}^{\prime}\bm{\Lambda}^{\mathrm{co}\prime}\bm{w}_{\mathcal{G}}=\bm{d}^{\prime}\bm{\lambda}_{0}$, which equals $1$ by \cref{assu:convex_hull} because $\bm{d}^{\prime}\bm{\lambda}_{0}=\bm{d}^{\prime}\bm{\Lambda}^{\mathrm{co}\prime}\bm{\pi}=\bm{1}_{K}^{\prime}\bm{\pi}=1$. Thus the sum-to-one constraint is redundant, and the only effective constraint is $\bm{\Lambda}^{\mathrm{co}\prime}\bm{w}_{\mathcal{G}}=\bm{\lambda}_{0}$. The FOC with respect to $\bm{w}_{\mathcal{G}}$ is $(\bm{Z}^{\prime}\bm{Z})^{-1}\bm{w}_{\mathcal{G}}=\bm{\Lambda}^{\mathrm{co}}\bm{\mu}_{1}.$ Combining these two expressions yields
\begin{equation}
\bm{\mu}_{1}=(\bm{\Lambda}^{\mathrm{co}\prime}\bm{Z}^{\prime}\bm{Z}\bm{\Lambda}^{\mathrm{co}})^{-1}\bm{\lambda}_{0}.\label{eq:mu1-in-col}
\end{equation}
It follows that
\begin{equation}
\bm{w}_{\mathcal{G}}=\bm{Z}^{\prime}\bm{Z}\bm{\Lambda}^{\mathrm{co}}(\bm{\Lambda}^{\mathrm{co}\prime}\bm{Z}^{\prime}\bm{Z}\bm{\Lambda}^{\mathrm{co}})^{-1}\bm{\lambda}_{0}.\label{eq:wg_oracle_r_less_K_in_col}
\end{equation}

\textbf{Case 2: $\bm{1}_{K}$ is out of the column space of $\bm{\Lambda}^{\mathrm{co}}$}. The FOC with respect to $\bm{w}_{\mathcal{G}}$ can be written as
\begin{equation}
(\bm{Z}^{\prime}\bm{Z})^{-1}\bm{w}_{\mathcal{G}}=\bm{\Lambda}^{\mathrm{co}}\bm{\mu}_{1}+\mu_{2}\bm{1}_{K}.\label{eq:foc_mat}
\end{equation}
Using the sum-to-one constraint $\bm{1}_{K}^{\prime}\bm{w}_{\mathcal{G}}=1$, we have $\bm{1}_{J}^{\prime}\bm{Z}\bm{\Lambda}^{\mathrm{co}}\bm{\mu}_{1}+J\mu_{2}=1$, which implies
\begin{equation}
\mu_{2}=(1-\bm{1}_{J}^{\prime}\bm{Z}\bm{\Lambda}^{\mathrm{co}}\bm{\mu}_{1})/J.\label{eq:mu2}
\end{equation}
Plugging this equation into \eqref{eq:foc_mat} and using the constraint $\bm{\Lambda}^{\mathrm{co}\prime}\bm{w}_{\mathcal{G}}=\bm{\lambda}_{0}$ in \eqref{eq:cons_linear}, we have
\[
\bm{\Lambda}^{\mathrm{co}\prime}(\bm{Z}^{\prime}\bm{Z})\left(\bm{M}_{\bm{Z}}\bm{\Lambda}^{\mathrm{co}}\bm{\mu}_{1}+\bm{1}_{K}/J\right)=\bm{\lambda}_{0},
\]
where $\bm{M}_{\bm{Z}}:=\bm{I}_{K}-\frac{1}{J}\bm{1}_{K}\bm{1}_{J}^{\prime}\bm{Z}$. We have $\bm{Z}\bm{M}_{\bm{Z}}=\bm{M}_{J}\bm{Z}$ where $\bm{M}_{J}:=\bm{I}_{J}-\frac{1}{J}\bm{1}_{J}\bm{1}_{J}'$. Under the assumption that $\bm{1}_{K}$ is not in the column space of $\bm{\Lambda}^{\mathrm{co}}$, it follows that $\bm{\Lambda}^{\mathrm{co}\prime}(\bm{Z}^{\prime}\bm{Z})\bm{M}_{\bm{Z}}\bm{\Lambda}^{\mathrm{co}}=\bm{\Lambda}^{\mathrm{co}\prime}\bm{Z}^{\prime}\bm{M}_{J}\bm{Z}\bm{\Lambda}^{\mathrm{co}}$ has full rank $r$, and we obtain
\begin{equation}
\bm{\mu}_{1}=[\bm{\Lambda}^{\mathrm{co}\prime}(\bm{Z}^{\prime}\bm{M}_{J}\bm{Z})\bm{\Lambda}^{\mathrm{co}}]^{-1}\left[\bm{\lambda}_{0}-\bm{\Lambda}^{\mathrm{co}\prime}(\bm{Z}^{\prime}\bm{Z})\bm{1}_{K}/J\right].\label{eq:mu1-not-in-col}
\end{equation}
By \eqref{eq:foc_mat}--\eqref{eq:mu1-not-in-col}, it follows that
\begin{align}
\bm{w}_{\mathcal{G}}^{*} & =\bm{Z}^{\prime}\bm{Z}\left(\bm{M}_{\bm{Z}}\bm{\Lambda}^{\mathrm{co}}\bm{\mu}_{1}+\frac{1}{J}\bm{1}_{K}\right)\nonumber \\
 & =\bm{Z}^{\prime}\bm{Z}\left(\bm{M}_{\bm{Z}}\bm{\Lambda}^{\mathrm{co}}[\bm{\Lambda}^{\mathrm{co}\prime}(\bm{Z}^{\prime}\bm{M}_{J}\bm{Z})\bm{\Lambda}^{\mathrm{co}}]^{-1}\left[\bm{\lambda}_{0}-\frac{1}{J}\bm{\Lambda}^{\mathrm{co}\prime}(\bm{Z}^{\prime}\bm{Z})\bm{1}_{K}\right]+\frac{1}{J}\bm{1}_{K}\right),\label{eq:wg_oracle_r_less_K_not_in_col}
\end{align}
where $\bm{M}_{\bm{Z}}=\bm{I}_{K}-\frac{1}{J}\bm{1}_{K}\bm{1}_{J}^{\prime}\bm{Z}$. Note that this expression only holds when the objective is the $L_{2}$ norm. For other types of objective such as EL- and entropy-relaxation, we do not have a closed-form solution as above.
\end{proof}
The relation between $K$ and $r$ is summarized as follows. The smaller one between them determines the number of effective (or linearly independent) constraints. When $K\leq r$, the variation in the $r$ factors provides enough information for identification of the $K$ groups. When $K>r$, however, under-identification arises, and therefore a regularization is needed to select the solution. It turns out that the strictly convex objective plays two roles: strict convexity not only ensures the within-group equal weights, but it also imposes a criterion for regularizing $\bm{w}_{\mathcal{G}}$ when there are many groups.

\subsection{Technical Lemmas }\label{subsec:Technical-Lemmas}

We establish key technical lemmas to support our analysis. In \cref{lem:bounds_for_mu},
we show two facts. First, when $K\leq r$, the eigenvalues of the ``core''
covariance matrix $\hat{\bm{\Sigma}}^{\mathrm{co}}$ are bounded away
from zero so that $\hat{\bm{\Sigma}}^{\mathrm{co}}$ is non-degenerate.
Second, when $K>r$, $\hat{\bm{\Sigma}}^{\mathrm{co}}$ is rank deficient
and thus the lower bound for minimum eigenvalue no longer holds, while
we prove that the Lagrange multiplier $\bm{\mu}_{1}$ associated with
the balancing constraint satisfies $\|\bm{\mu}_{1}\|_{2}=O_{p}(r/J)$,
which is sufficient for subsequent results in the $K>r$ case. \cref{lem:sampling-error}
provides bounds on $\|\hat{\bm{\Sigma}}^{\mathrm{e}}\|_{1}$ and $\|\hat{\bm{\Upsilon}}^{\mathrm{e}}\|_{\infty}$.
\cref{lem:w_star_feasibility} shows that the oracle weight is a feasible
solution for the $L_{2}$-SCM-relaxation, a critical step for proving
the compatibility inequalities between $\hat{\bm{w}}-\bm{w}^{*}$
and $\hat{\bm{w}}_{\mathcal{G}}-\bm{w}_{\mathcal{G}}^{*}$ in \cref{lem:link_two_norms}.
\begin{lem}
\label{lem:bounds_for_mu}Suppose \cref{assu:DGP} and \cref{assu:eta}\ref{enu:group-size} hold.
\begin{enumerate}[label=(\roman*),parsep=0pt,topsep=0pt,font=\upshape]
\item \label{enu:K_less_r}If $K\leq r$, then $c_{1}\leq\phi_{\min}(\hat{\bm{\Sigma}}^{\mathrm{co}})\leq\phi_{\max}(\hat{\bm{\Sigma}}^{\mathrm{co}})\leq c_{2}$
holds w.p.a.1.~for some absolute constants $c_{1}$ and $c_{2}$.
\item \label{enu:K_gtr_r}If $K>r$, then the $r$-dimensional Lagrange
multiplier $\bm{\mu}_{1}$ given by \eqref{eq:mu1-in-col} or \eqref{eq:mu1-not-in-col}
has order $\|\bm{\mu}_{1}\|_{2}=O_{p}(r/J)$.
\end{enumerate}
\end{lem}
\begin{proof}[Proof of \cref{lem:bounds_for_mu}]
\ref{enu:K_less_r} First, we give the bounds for the spectrum of
the population core $\bm{\Sigma}^{\mathrm{co}}:=\bm{\Lambda}^{\mathrm{co}}\bm{\Omega}_{\bm{F}}\bm{\Lambda}^{\mathrm{co}\prime}$.
The minimum eigenvalue of $\bm{\Sigma}^{\mathrm{co}}$ can be bounded
by
\[
\phi_{\min}(\bm{\Sigma}^{\mathrm{co}})=\phi_{\min}(\bm{\Lambda}^{\mathrm{co}}\bm{\Omega}_{\bm{F}}\bm{\Lambda}^{\mathrm{co}\prime})\geq\phi_{\min}(\bm{\Omega}_{\bm{F}})[\varsigma_{\min}(\bm{\Lambda}^{\mathrm{co}})]^{2}\geq\underline{c}^{2}.
\]
Likewise, we can deduce
\[
\phi_{\max}(\bm{\Sigma}^{\mathrm{co}})\leq\phi_{\max}(\bm{\Omega}_{\bm{F}})[\varsigma_{\max}(\bm{\Lambda}^{\mathrm{co}})]^{2}\leq\bar{c}^{2}.
\]
Next we give the probabilistic bounds for the spectrum of $\hat{\bm{\Sigma}}^{\mathrm{co}}$.
By Weyl's inequality \citep[Corollary~4.3.15]{hornMatrixAnalysis2017},
we have
\begin{align*}
\phi_{\min}(\hat{\bm{\Sigma}}^{\mathrm{co}}) & \geq\phi_{\min}(\bm{\Sigma}^{\mathrm{co}})-\|\bm{\Lambda}^{\mathrm{co}}(\hat{\bm{\Omega}}_{\bm{F}}-\bm{\Omega}_{\bm{F}})\bm{\Lambda}^{\mathrm{co}\prime}\|_{2}\\
 & \geq\phi_{\min}(\bm{\Sigma}^{\mathrm{co}})-\|\bm{\Lambda}^{\mathrm{co}}\|_{2}^{2}\|\hat{\bm{\Omega}}_{\bm{F}}-\bm{\Omega}_{\bm{F}}\|_{2}\\
 & \geq\underline{c}^{2}-\bar{c}\cdot o_{p}(1)\geq\frac{1}{2}\underline{c}^{2}\qquad\text{w.p.a.1},
\end{align*}
where the second line is by the sub-multiplicativity of the spectral
norm and in the last line we use $\|\bm{\Lambda}^{\mathrm{co}}\|_{2}^{2}=\varsigma_{\max}^{2}(\bm{\Lambda}^{\mathrm{co}})\leq\bar{c}$
and $\|\hat{\bm{\Omega}}_{\bm{F}}-\bm{\Omega}_{\bm{F}}\|_{2}=o_{p}(1)$
by \cref{assu:DGP}. Likewise, $\phi_{\max}(\hat{\bm{\Sigma}}^{\mathrm{co}})$
can be bounded as
\[
\phi_{\max}(\hat{\bm{\Sigma}}^{\mathrm{co}})\leq\phi_{\max}(\bm{\Sigma}^{\mathrm{co}})+\|\bm{\Lambda}^{\mathrm{co}}\|_{2}^{2}\|\hat{\bm{\Omega}}_{\bm{F}}-\bm{\Omega}_{\bm{F}}\|_{2}\leq2\bar{c}^{2}\qquad\text{w.p.a.1.}
\]

\ref{enu:K_gtr_r} When $K>r$, we consider two cases. Let $\bm{A}:=\bm{Z}\bm{\Lambda}^{\mathrm{co}}$ and $\bar{\bm{\lambda}}:=J^{-1}\bm{A}'\bm{1}_{J}$. If $\bm{1}_{K}\in\operatorname{col}(\bm{\Lambda}^{\mathrm{co}})$, then \eqref{eq:mu1-in-col} gives $\bm{\mu}_{1}=(\bm{A}'\bm{A})^{-1}\bm{\lambda}_{0}$. By \cref{assu:DGP}\ref{enu:Loadings} and \cref{assu:eta}\ref{enu:group-size},
\[
\phi_{\min}(\bm{A}'\bm{A})\geq\left(\min_{k\in[K]}J_{k}\right)\varsigma_{\min}^{2}(\bm{\Lambda}^{\mathrm{co}})\gtrsim\frac{J}{K}\frac{K}{r}=\frac{J}{r}.
\]
It follows that
\[
\|\bm{\mu}_{1}\|_{2}\leq[\phi_{\min}(\bm{A}'\bm{A})]^{-1}\|\bm{\lambda}_{0}\|_{2}=O(r/J).
\]

If $\bm{1}_{K}\notin\operatorname{col}(\bm{\Lambda}^{\mathrm{co}})$, we have that, by \cref{assu:DGP}\ref{enu:Loadings}, $\phi_{\min}(\bm{A}'\bm{M}_{J}\bm{A})\geq\underline{c}J/r$. Moreover, $\|\bar{\bm{\lambda}}\|_{2}\leq\sum_{k=1}^{K}(J_{k}/J)\|\bm{\lambda}_{k}^{\mathrm{co}}\|_{2}\leq\bar{c}$.
Using \eqref{eq:mu1-not-in-col}, we obtain
\begin{align*}
\|\bm{\mu}_{1}\|_{2} & =\|(\bm{A}'\bm{M}_{J}\bm{A})^{-1}(\bm{\lambda}_{0}-\bar{\bm{\lambda}})\|_{2}\\
& \leq\frac{r}{\underline{c}J}\left(\|\bm{\lambda}_{0}\|_{2}+\|\bar{\bm{\lambda}}\|_{2}\right)=O(r/J).
\end{align*}
The desired bound follows in both cases.
\end{proof}
\begin{lem}
\label{lem:sampling-error}Under \cref{assu:DGP,assu:errors},
we have
\[
\|\hat{\bm{\Sigma}}^{\mathrm{e}}\|_{1}=O_{p}\biggl(J\sqrt{\frac{r\log J}{T_{0}}}\biggr)\quad\text{and}\quad\Vert\hat{\bm{\Upsilon}}^{\mathrm{e}}\Vert_{\infty}=O_{p}\biggl(\sqrt{\frac{r\log J}{T_{0}}}\biggr).
\]
\end{lem}
\begin{proof}[Proof of \cref{lem:sampling-error}]
For matrices $\bm{A}_{m\times\ell}=(\bm{a}_{1},\dots,\bm{a}_{m})^{\prime}$
and $\bm{B}_{\ell\times n}=(\bm{b}_{1},\dots,\bm{b}_{n})$, we have
\begin{align*}
\|\bm{A}\bm{B}\|_{\infty} & =\sup_{i}\sum_{j}|\bm{a}_{i}^{\prime}\bm{b}_{j}|\leq\sup_{i}\|\bm{a}_{i}\|_{1}\sum_{j}\|\bm{b}_{j}\|_{\infty}\leq n\|\bm{A}\|_{\infty}\|\bm{B}\|_{\max},\\
\|\bm{A}\bm{B}\|_{1} & =\sup_{j}\sum_{i}|\bm{a}_{i}^{\prime}\bm{b}_{j}|\leq\sum_{i}\|\bm{a}_{i}\|_{1}\sup_{j}\|\bm{b}_{j}\|_{\infty}\leq m\|\bm{A}\|_{\infty}\|\bm{B}\|_{\max}.
\end{align*}
The first fact, together with \cref{assu:DGP,assu:errors},
leads to
\begin{align*}
\frac{1}{T_{0}}\|\bm{Z}\bm{\Lambda}^{\mathrm{co}}\bm{F}'\bm{U}\|_{\infty} & \leq\frac{J}{T_{0}}\|\bm{Z}\bm{\Lambda}^{\mathrm{co}}\|_{\infty}\|\bm{F}'\bm{U}\|_{\max}\\
 & \leq\frac{\sqrt{r}J}{T_{0}}\sup_{j\in[J]}\|\bm{\lambda}_{j}\|_{2}\cdot\|\bm{F}'\bm{U}\|_{\max}=O_{p}\biggl(J\sqrt{\frac{r\log J}{T_{0}}}\biggr).
\end{align*}
Similarly, the second fact leads to
\[
\frac{1}{T_{0}}\|\bm{Z}\bm{\Lambda}^{\mathrm{co}}\bm{F}'\bm{U}\|_{1}\leq\frac{J}{T_{0}}\|\bm{Z}\bm{\Lambda}^{\mathrm{co}}\|_{\infty}\|\bm{F}'\bm{U}\|_{\max}=O_{p}\biggl(J\sqrt{\frac{r\log J}{T_{0}}}\biggr).
\]
Furthermore, by \cref{assu:errors},
\[
\frac{1}{T_{0}}\|\bm{U}^{\prime}\bm{U}\|_{1}\leq\frac{1}{T_{0}}\|\bm{U}^{\prime}\bm{U}-\mathbb{E}(\bm{U}^{\prime}\bm{U})\|_{1}+\frac{1}{T_{0}}\|\mathbb{E}(\bm{U}^{\prime}\bm{U})\|_{1}=O_{p}\biggl(J\sqrt{\frac{\log J}{T_{0}}}\biggr).
\]
Hence,
\[
\|\hat{\bm{\Sigma}}^{\mathrm{e}}\|_{1}\leq\frac{1}{T_{0}}\|\bm{Z}\bm{\Lambda}^{\mathrm{co}}\bm{F}'\bm{U}\|_{1}+\frac{1}{T_{0}}\|\bm{Z}\bm{\Lambda}^{\mathrm{co}}\bm{F}'\bm{U}\|_{\infty}+\frac{1}{T_{0}}\|\bm{U}^{\prime}\bm{U}\|_{1}=O_{p}\biggl(J\sqrt{\frac{r\log J}{T_{0}}}\biggr).
\]

For $\Vert\hat{\bm{\Upsilon}}^{\mathrm{e}}\Vert_{\infty}$, note that
by \cref{assu:DGP,assu:errors},
\[
\frac{1}{T_{0}}\|\bm{Z}\bm{\Lambda}^{\mathrm{co}}\bm{F}'\bm{u}_{0}\|_{\infty}\leq\frac{1}{T_{0}}\|\bm{Z}\bm{\Lambda}^{\mathrm{co}}\|_{\infty}\|\bm{F}'\bm{u}_{0}\|_{\infty}=O_{p}\biggl(\sqrt{\frac{r\log J}{T_{0}}}\biggr).
\]
Similarly,
\[
\frac{1}{T_{0}}\|\bm{U}^{\prime}\bm{F}\bm{\lambda}_{0}\|_{\infty}\leq\frac{1}{T_{0}}\|\bm{F}'\bm{U}\|_{\max}\|\bm{\lambda}_{0}\|_{1}=O_{p}\biggl(\sqrt{\frac{r\log J}{T_{0}}}\biggr).
\]
Furthermore, by \cref{assu:errors},
\[
\frac{1}{T_{0}}\|\bm{U}^{\prime}\bm{u}_{0}\|_{\infty}=O_{p}\biggl(\sqrt{\frac{\log J}{T_{0}}}\biggr).
\]
Thus,
\[
\Vert\hat{\bm{\Upsilon}}^{\mathrm{e}}\Vert_{\infty}\leq\frac{1}{T_{0}}\|\bm{Z}\bm{\Lambda}^{\mathrm{co}}\bm{F}'\bm{u}_{0}\|_{\infty}+\frac{1}{T_{0}}\|\bm{U}^{\prime}\bm{F}\bm{\lambda}_{0}\|_{\infty}+\frac{1}{T_{0}}\|\bm{U}^{\prime}\bm{u}_{0}\|_{\infty}=O_{p}\biggl(\sqrt{\frac{r\log J}{T_{0}}}\biggr).
\]
This completes the proof.
\end{proof}

Recall that $S_\eta$ has been defined in \eqref{eq:feasible-set} as the feasible set of the SCM-relaxation problem.

\begin{lem}
\label{lem:w_star_feasibility}Under \crefrange{assu:DGP}{assu:eta}, $\bm{w}^{*}\in S_\eta$ w.p.a.1.
\end{lem}
\begin{proof}[Proof of \cref{lem:w_star_feasibility}]
It suffices to show that $\|\hat{\bm{\Sigma}}\bm{w}^{*}-\hat{\bm{\Upsilon}}\|_{\infty}\leq\eta$
w.p.a.1. By the definition of $\bm{w}^{*}$, we have
\begin{align*}
\|\hat{\bm{\Sigma}}\bm{w}^{*}-\hat{\bm{\Upsilon}}\|_{\infty} & =\|\hat{\bm{\Sigma}}\bm{w}^{*}-\hat{\bm{\Upsilon}}-(\hat{\bm{\Sigma}}^{*}\bm{w}^{*}-\hat{\bm{\Upsilon}}^{*})\|_{\infty}\\
 & \leq\|\hat{\bm{\Sigma}}^{\mathrm{e}}\bm{w}^{*}\|_{\infty}+\|\hat{\bm{\Upsilon}}^{\mathrm{e}}\|_{\infty}\leq\|\hat{\bm{\Sigma}}^{\mathrm{e}}\|_{1}\|\bm{w}^{*}\|_{\infty}+\|\hat{\bm{\Upsilon}}^{\mathrm{e}}\|_{\infty}\\
 & \lesssim_{p}K\sqrt{\frac{r\log J}{T_{0}}}=o(\eta),
\end{align*}
where the last line uses the fact that $\|\bm{w}^{*}\|_{\infty}\leq\|\bm{w}_{\mathcal{G}}^{*}\|_{\infty}/\min_{k\in[K]}J_{k}=O_{p}(K/J)$,
\cref{assu:eta}, and \cref{lem:sampling-error}. It then follows that
$\|\hat{\bm{\Sigma}}\bm{w}^{*}-\hat{\bm{\Upsilon}}\|_{\infty}\leq\eta$
w.p.a.1.
\end{proof}
Define $\hat{\bm{w}}_{\mathcal{G}}:=\bm{Z}^{\prime}\hat{\bm{w}}$
as the group-level aggregated weight estimate. The following lemma
establishes compatibility inequalities between $\hat{\bm{w}}-\bm{w}^{*}$
and $\hat{\bm{w}}_{\mathcal{G}}-\bm{w}_{\mathcal{G}}^{*}$ in $L_{2}$-norm.
\begin{lem}
\label{lem:link_two_norms}Under \crefrange{assu:DGP}{assu:eta},
we have
\begin{equation}
\|\hat{\bm{w}}-\bm{w}^{*}\|_{2}^{2}\lesssim_{p}\frac{K}{J}\|\hat{\bm{w}}_{\mathcal{G}}-\bm{w}_{\mathcal{G}}^{*}\|_{2}.\label{eq:relationship_two_errors}
\end{equation}
Furthermore, if $K>r$, it holds that
\begin{equation}
\|\hat{\bm{w}}-\bm{w}^{*}\|_{2}^{2}\lesssim_{p}\frac{r}{J}\|\bm{\Lambda}^{\mathrm{co}\prime}(\hat{\bm{w}}_{\mathcal{G}}-\bm{w}_{\mathcal{G}}^{*})\|_{2}.\label{eq:another-compatible}
\end{equation}
\end{lem}
\begin{proof}[Proof of \cref{lem:link_two_norms}]
First note that for any $x$ and $y$, we have
\begin{equation}
x^{2}-y^{2}=2y(x-y)+(x-y)^{2}.\label{eq:strong_convexity}
\end{equation}
By \cref{lem:w_star_feasibility}, we know that $\bm{w}^{*}\in S_\eta$
is feasible for $L_{2}$-SCM-relaxation w.p.a.1. Because $\hat{\bm{w}}$
minimizes $\sum_{j=1}^{J}w_{j}^{2}$ over the entire feasible set,
we have the basic inequality $\sum_{j=1}^{J}\hat{w}_{j}^{2}\leq\sum_{j=1}^{J}w_{j}^{*2}$
w.p.a.1. It follows that
\begin{equation}
0\geq\sum_{j=1}^{J}\hat{w}_{j}^{2}-\sum_{j=1}^{J}w_{j}^{*2}\geq\sum_{j=1}^{J}2w_{j}^{*}(\hat{w}_{j}-w_{j}^{*})+\|\hat{\bm{w}}-\bm{w}^{*}\|_{2}^{2}\qquad\text{w.p.a.1}.\label{eq:Basic_inequality}
\end{equation}
In view of the fact that $\bm{w}^{*}$ is within-group equal, we obtain
\begin{align}
\|\hat{\bm{w}}-\bm{w}^{*}\|_{2}^{2} & \stackrel{\text{(i)}}{\leq}\left|\sum_{j=1}^{J}2w_{j}^{*}(\hat{w}_{j}-w_{j}^{*})\right|=\left|\sum_{k=1}^{K}\frac{2w_{\mathcal{G}_{k}}^{*}}{J_{k}}\left(\hat{w}_{\mathcal{G}_{k}}-w_{\mathcal{G}_{k}}^{*}\right)\right|\nonumber \\
 & \stackrel{\text{(ii)}}{\leq}\left[\sum_{k=1}^{K}\left(\frac{2w_{\mathcal{G}_{k}}^{*}}{J_{k}}\right)^{2}\right]^{1/2}\left[\sum_{k=1}^{K}\left(\hat{w}_{\mathcal{G}_{k}}-w_{\mathcal{G}_{k}}^{*}\right)^{2}\right]^{1/2}\nonumber \\
 & =2\|(\bm{Z}^{\prime}\bm{Z})^{-1}\bm{w}_{\mathcal{G}}^{*}\|_{2}\cdot\|\hat{\bm{w}}_{\mathcal{G}}-\bm{w}_{\mathcal{G}}^{*}\|_{2}\nonumber \\
 & \stackrel{\text{(iii)}}{\leq}2\Bigl(\min_{k\in[K]}J_{k}\Bigr)^{-1}\|\bm{w}_{\mathcal{G}}^{*}\|_{2}\cdot\|\hat{\bm{w}}_{\mathcal{G}}-\bm{w}_{\mathcal{G}}^{*}\|_{2}\nonumber \\
 & \stackrel{\text{(iv)}}{\lesssim_{p}}\frac{K}{J}\|\hat{\bm{w}}_{\mathcal{G}}-\bm{w}_{\mathcal{G}}^{*}\|_{2},\label{eq:compatible}
\end{align}
where step (i) uses \eqref{eq:Basic_inequality}, step (ii) is due
to the Cauchy-Schwarz inequality, step (iii) comes from the fact that
$\|(\bm{Z}^{\prime}\bm{Z})^{-1}\|_{2}=\phi_{\max}\bigl((\bm{Z}^{\prime}\bm{Z})^{-1}\bigr)=[\phi_{\min}(\bm{Z}^{\prime}\bm{Z})]^{-1}=(\min_{k\in[K]}J_{k})^{-1}$,
and step (iv) uses \cref{assu:eta}\ref{enu:group-size} and $\|\bm{w}_{\mathcal{G}}^{*}\|_{2}\leq\|\bm{w}_{\mathcal{G}}^{*}\|_{1}=1$.

If $K>r$, then by the FOC \eqref{eq:foc_mat}, we can deduce, similar
to \eqref{eq:compatible}, that
\begin{align*}
\|\hat{\bm{w}}-\bm{w}^{*}\|_{2}^{2} & \leq2\bigl|[(\bm{Z}^{\prime}\bm{Z})^{-1}\bm{w}_{\mathcal{G}}^{*}]^{\prime}(\hat{\bm{w}}_{\mathcal{G}}-\bm{w}_{\mathcal{G}}^{*})\bigr|=2|\bm{\mu}_{1}^{\prime}\bm{\Lambda}^{\mathrm{co}\prime}(\hat{\bm{w}}_{\mathcal{G}}-\bm{w}_{\mathcal{G}}^{*})|\\
 & \leq2\|\bm{\mu}_{1}\|_{2}\cdot\|\bm{\Lambda}^{\mathrm{co}\prime}(\hat{\bm{w}}_{\mathcal{G}}-\bm{w}_{\mathcal{G}}^{*})\|_{2}\lesssim_{p}\frac{r}{J}\|\bm{\Lambda}^{\mathrm{co}\prime}(\hat{\bm{w}}_{\mathcal{G}}-\bm{w}_{\mathcal{G}}^{*})\|_{2},
\end{align*}
where the last line follows from \cref{lem:bounds_for_mu}.
\end{proof}

\subsection[Proof of Theorem]{Proof of \cref{thm:w_convergence} }\label{subsec:proof-thm1}

Now we are ready to prove \cref{thm:w_convergence} based on the lemmas
in the preceding section.
\begin{proof}[Proof of \cref{thm:w_convergence}]
 For convenience, denote $\hat{\bm{h}}:=\hat{\bm{w}}-\bm{w}^{*}$
and $\hat{\bm{h}}_{\mathcal{G}}:=\hat{\bm{w}}_{\mathcal{G}}-\bm{w}_{\mathcal{G}}^{*}$.
Clearly, they satisfy $\hat{\bm{h}}_{\mathcal{G}}=\bm{Z}^{\prime}\hat{\bm{h}}$.
First, notice that the relaxation condition $\|\hat{\bm{\Sigma}}\hat{\bm{w}}-\hat{\bm{\Upsilon}}\|_{\infty}\leq\eta$
yields
\begin{align}
\eta & \geq\|\hat{\bm{\Sigma}}\hat{\bm{w}}-\hat{\bm{\Upsilon}}\|_{\infty}\stackrel{\text{(i)}}{=}\|\hat{\bm{\Sigma}}\hat{\bm{w}}-\hat{\bm{\Upsilon}}-(\hat{\bm{\Sigma}}^{*}\bm{w}^{*}-\hat{\bm{\Upsilon}}^{*})\|_{\infty}\nonumber \\
 & =\|\hat{\bm{\Sigma}}^{*}\hat{\bm{h}}+\hat{\bm{\Sigma}}^{\mathrm{e}}\hat{\bm{w}}-\hat{\bm{\Upsilon}}^{\mathrm{e}}\|_{\infty}=\|\bm{Z}\hat{\bm{\Sigma}}^{\mathrm{co}}\hat{\bm{h}}_{\mathcal{G}}+\hat{\bm{\Sigma}}^{\mathrm{e}}\hat{\bm{w}}-\hat{\bm{\Upsilon}}^{\mathrm{e}}\|_{\infty},\label{eq:upper_bound1}
\end{align}
where step~(i) is due to the fact that $\hat{\bm{\Sigma}}^{*}\bm{w}^{*}-\hat{\bm{\Upsilon}}^{*}=\bm{0}_{J}$.
By Hölder's inequality and \eqref{eq:upper_bound1}, we have an upper
bound
\begin{align}
 & \ \left|\hat{\bm{h}}'\left[\bm{Z}\hat{\bm{\Sigma}}^{\mathrm{co}}\hat{\bm{h}}_{\mathcal{G}}+\hat{\bm{\Sigma}}^{\mathrm{e}}\hat{\bm{w}}-\hat{\bm{\Upsilon}}^{\mathrm{e}}\right]\right|\nonumber \\
\leq & \ \|\hat{\bm{h}}\|_{1}\|\bm{Z}\hat{\bm{\Sigma}}^{\mathrm{co}}\hat{\bm{h}}_{\mathcal{G}}+\hat{\bm{\Sigma}}^{\mathrm{e}}\hat{\bm{w}}-\hat{\bm{\Upsilon}}^{\mathrm{e}}\|_{\infty}\leq\eta\|\hat{\bm{h}}\|_{1}.\label{eq:fact-1}
\end{align}

If $K\leq r$, a lower bound can be derived as
\begin{align}
 & \left|\hat{\bm{h}}'\left[\bm{Z}\hat{\bm{\Sigma}}^{\mathrm{co}}\hat{\bm{h}}_{\mathcal{G}}+\hat{\bm{\Sigma}}^{\mathrm{e}}\hat{\bm{w}}-\hat{\bm{\Upsilon}}^{\mathrm{e}}\right]\right|\nonumber \\
\stackrel{\text{(i)}}{=}{} & \left|\hat{\bm{h}}_{\mathcal{G}}'\hat{\bm{\Sigma}}^{\mathrm{co}}\hat{\bm{h}}_{\mathcal{G}}+\hat{\bm{h}}'(\hat{\bm{\Sigma}}^{\mathrm{e}}\hat{\bm{h}}+\hat{\bm{\Sigma}}^{\mathrm{e}}\bm{w}^{*}-\hat{\bm{\Upsilon}}^{\mathrm{e}})\right|\nonumber \\
\stackrel{\text{(ii)}}{\geq}{} & \hat{\bm{h}}_{\mathcal{G}}'\hat{\bm{\Sigma}}^{\mathrm{co}}\hat{\bm{h}}_{\mathcal{G}}+T_{0}^{-1}\hat{\bm{h}}'\bm{U}'\bm{U}\hat{\bm{h}}-2\left|T_{0}^{-1}\hat{\bm{h}}_{\mathcal{G}}'\bm{\Lambda}^{\mathrm{co}}\bm{F}^{\prime}\bm{U}\hat{\bm{h}}\right|-\left|\hat{\bm{h}}'(\hat{\bm{\Sigma}}^{\mathrm{e}}\bm{w}^{*}-\hat{\bm{\Upsilon}}^{\mathrm{e}})\right|\nonumber \\
\stackrel{\text{(iii)}}{\geq}{} & \hat{\bm{h}}_{\mathcal{G}}'\hat{\bm{\Sigma}}^{\mathrm{co}}\hat{\bm{h}}_{\mathcal{G}}-\left(2\|T_{0}^{-1}\bm{U}'\bm{F}\bm{\Lambda}^{\mathrm{co}\prime}\hat{\bm{h}}_{\mathcal{G}}\|_{\infty}+\|\hat{\bm{\Sigma}}^{\mathrm{e}}\bm{w}^{*}-\hat{\bm{\Upsilon}}^{\mathrm{e}}\|_{\infty}\right)\|\hat{\bm{h}}\|_{1}\nonumber \\
\stackrel{\text{(iv)}}{\geq}{} & \hat{\bm{h}}_{\mathcal{G}}'\hat{\bm{\Sigma}}^{\mathrm{co}}\hat{\bm{h}}_{\mathcal{G}}-\left(2\|T_{0}^{-1}\bm{\Lambda}^{\mathrm{co}}\bm{F}'\bm{U}\|_{\max}\|\hat{\bm{h}}_{\mathcal{G}}\|_{1}+\|\hat{\bm{\Sigma}}^{\mathrm{e}}\|_{1}\|\bm{w}^{*}\|_{\infty}+\|\hat{\bm{\Upsilon}}^{\mathrm{e}}\|_{\infty}\right)\|\hat{\bm{h}}\|_{1}\nonumber \\
\stackrel{\text{(v)}}{\gtrsim_{p}}{} & \phi_{\min}(\hat{\bm{\Sigma}}^{\mathrm{co}})\|\hat{\bm{h}}_{\mathcal{G}}\|_{2}^{2}-\biggl(\sqrt{\frac{r\log J}{T_{0}}}+J\sqrt{\frac{r\log J}{T_{0}}}\cdot\frac{K}{J}+\sqrt{\frac{r\log J}{T_{0}}}\biggr)\|\hat{\bm{h}}\|_{1}\nonumber \\
\stackrel{\text{(vi)}}{\gtrsim_{p}}{} & \phi_{\min}(\hat{\bm{\Sigma}}^{\mathrm{co}})\|\hat{\bm{h}}_{\mathcal{G}}\|_{2}^{2}-\eta\|\hat{\bm{h}}\|_{1}.\label{eq:fact-2}
\end{align}
In the above derivation, step~(i) is because $\hat{\bm{h}}_{\mathcal{G}}=\bm{Z}^{\prime}\hat{\bm{h}}$
and $\bm{1}_{J}^{\prime}\hat{\bm{h}}=0$, step~(ii) follows from
the triangle inequality and the definition of $\hat{\bm{\Sigma}}^{\mathrm{e}}$,
step~(iii) invokes Hölder's inequality, step~(iv) uses the facts
that $\|T_{0}^{-1}\bm{U}'\bm{F}\bm{\Lambda}^{\mathrm{co}\prime}\hat{\bm{h}}_{\mathcal{G}}\|_{\infty}\leq\|\hat{\bm{h}}_{\mathcal{G}}\|_{1}\|T_{0}^{-1}\bm{\Lambda}^{\mathrm{co}}\bm{F}'\bm{U}\|_{\max}$
and $\|\hat{\bm{\Sigma}}^{\mathrm{e}}\bm{w}^{*}\|_{\infty}\leq\|\hat{\bm{\Sigma}}^{\mathrm{e}}\|_{1}\|\bm{w}^{*}\|_{\infty}$,
and step~(v) uses the definition of $\phi_{\min}(\cdot)$, the fact
that $\|\hat{\bm{h}}_{\mathcal{G}}\|_{1}\leq\|\hat{\bm{w}}_{\mathcal{G}}\|_{1}+\|\bm{w}_{\mathcal{G}}^{*}\|_{1}=2$,
and \cref{lem:sampling-error}. Finally, step~(vi) is by \cref{assu:eta}.

If $K>r$, similarly we can bound it from below by
\begin{eqnarray}
 &  & \left|\hat{\bm{h}}'\left[\bm{Z}\hat{\bm{\Sigma}}^{\mathrm{co}}\hat{\bm{h}}_{\mathcal{G}}+\hat{\bm{\Sigma}}^{\mathrm{e}}\hat{\bm{w}}-\hat{\bm{\Upsilon}}^{\mathrm{e}}\right]\right|\nonumber \\
 & \gtrsim_{p} & \hat{\bm{h}}_{\mathcal{G}}'\hat{\bm{\Sigma}}^{\mathrm{co}}\hat{\bm{h}}_{\mathcal{G}}-\biggl(K\sqrt{\frac{r\log J}{T_{0}}}\biggr)\|\hat{\bm{h}}\|_{1}\nonumber \\
 & = & (\bm{\Lambda}^{\mathrm{co}\prime}\hat{\bm{h}}_{\mathcal{G}})^{\prime}\hat{\bm{\Omega}}_{\bm{F}}(\bm{\Lambda}^{\mathrm{co}\prime}\hat{\bm{h}}_{\mathcal{G}})-\biggl(K\sqrt{\frac{r\log J}{T_{0}}}\biggr)\|\hat{\bm{h}}\|_{1}\nonumber \\
 & \gtrsim_{p} & \phi_{\min}(\hat{\bm{\Omega}}_{\bm{F}})\|\bm{\Lambda}^{\mathrm{co}\prime}\hat{\bm{h}}_{\mathcal{G}}\|_{2}^{2}-\eta\|\hat{\bm{h}}\|_{1}.\label{eq:fact-2-2}
\end{eqnarray}

\medskip

Combining \eqref{eq:fact-1} and \eqref{eq:fact-2}, if $K\leq r$
we have $\|\hat{\bm{h}}_{\mathcal{G}}\|_{2}^{2}\lesssim_{p}[\phi_{\min}(\hat{\bm{\Sigma}}^{\mathrm{co}})]^{-1}\eta\|\hat{\bm{h}}\|_{1}.$
Then by \eqref{eq:relationship_two_errors} in \cref{lem:link_two_norms},
\begin{equation}
\frac{1}{J^{2}}\|\hat{\bm{h}}\|_{1}^{4}\leq\|\hat{\bm{h}}\|_{2}^{4}\lesssim_{p}\frac{K^{2}}{J^{2}}\|\hat{\bm{h}}_{\mathcal{G}}\|_{2}^{2}\lesssim_{p}\frac{K^{2}\eta}{J^{2}\phi_{\min}(\hat{\bm{\Sigma}}^{\mathrm{co}})}\|\hat{\bm{h}}\|_{1}\leq\frac{K^{2}\eta}{J^{3/2}\phi_{\min}(\hat{\bm{\Sigma}}^{\mathrm{co}})}\|\hat{\bm{h}}\|_{2},\label{eq:l1norminequality}
\end{equation}
where we use the Cauchy-Schwarz inequality $\|\hat{\bm{h}}\|_{1}\leq\sqrt{J}\|\hat{\bm{h}}\|_{2}$.
Hence,
\begin{align*}
\|\hat{\bm{h}}\|_{1} & \lesssim_{p}\left[\frac{K^{2}\eta}{\phi_{\min}(\hat{\bm{\Sigma}}^{\mathrm{co}})}\right]^{1/3}=O_{p}\bigl((K^{2}\eta)^{1/3}\bigr)=o_{p}(1),\text{ and }\\
\|\hat{\bm{h}}\|_{2} & \lesssim_{p}\left[\frac{K^{2}\eta}{J^{3/2}\phi_{\min}(\hat{\bm{\Sigma}}^{\mathrm{co}})}\right]^{1/3}=O_{p}\left(\frac{(K^{2}\eta)^{1/3}}{\sqrt{J}}\right)=o_{p}\left(\frac{1}{\sqrt{J}}\right),
\end{align*}
by \cref{assu:eta} and \cref{lem:bounds_for_mu}.

If $K>r$, we combine \eqref{eq:fact-1} and \eqref{eq:fact-2-2}
to deduce $\|\bm{\Lambda}^{\mathrm{co}\prime}\hat{\bm{h}}_{\mathcal{G}}\|_{2}^{2}\lesssim_{p}[\phi_{\min}(\hat{\bm{\Omega}}_{\bm{F}})]^{-1}\eta\|\hat{\bm{h}}\|_{1}$.
Then invoking \eqref{eq:another-compatible} in \cref{lem:link_two_norms},
we have
\[
\frac{1}{J^{2}}\|\hat{\bm{h}}\|_{1}^{4}\leq\|\hat{\bm{h}}\|_{2}^{4}\lesssim_{p}\frac{r^{2}}{J^{2}}\|\bm{\Lambda}^{\mathrm{co}\prime}\hat{\bm{h}}_{\mathcal{G}}\|_{2}^{2}\leq\frac{r^{2}\eta}{J^{2}\phi_{\min}(\hat{\bm{\Omega}}_{\bm{F}})}\|\hat{\bm{h}}\|_{1}\leq\frac{r^{2}\eta}{J^{3/2}\phi_{\min}(\hat{\bm{\Omega}}_{\bm{F}})}\|\hat{\bm{h}}\|_{2}.
\]
Accordingly, by \cref{assu:eta} and $\phi_{\min}(\hat{\bm{\Omega}}_{\bm{F}})\gtrsim_{p}1$,
we conclude
\begin{align*}
\|\hat{\bm{h}}\|_{1} & \lesssim_{p}\left[\frac{r^{2}\eta}{\phi_{\min}(\hat{\bm{\Omega}}_{\bm{F}})}\right]^{1/3}=O_{p}\bigl((r^{2}\eta)^{1/3}\bigr)=o_{p}(1),\\
\|\hat{\bm{h}}\|_{2} & \lesssim_{p}\left[\frac{r^{2}\eta}{J^{3/2}\phi_{\min}(\hat{\bm{\Omega}}_{\bm{F}})}\right]^{1/3}=O_{p}\left(\frac{(r^{2}\eta)^{1/3}}{\sqrt{J}}\right)=o_{p}\left(\frac{1}{\sqrt{J}}\right).\qedhere
\end{align*}
\end{proof}

\subsection[Proof of Theorem]{Proof of \cref{thm:oracle_inequalities} }\label{subsec:proof-thm2}
\begin{proof}[Proof of \cref{thm:oracle_inequalities}]
Part (i). First, $R_{\mathcal{T}_{0}}(\hat{\bm{w}})-R_{\mathcal{T}_{0}}(\bm{w}^{*})$
can be written as
\begin{align*}
 & T_{0}^{-1}\sum_{t\in\mathcal{T}_{0}}(\hat{\bm{w}}'\bm{y}_{t}-y_{0t})^{2}-T_{0}^{-1}\sum_{t\in\mathcal{T}_{0}}(\bm{w}^{*}{}'\bm{y}_{t}-y_{0t})^{2}\\
={} & T_{0}^{-1}\|\bm{Y}\hat{\bm{w}}-\bm{y}_{0}\|_{2}^{2}-T_{0}^{-1}\|\bm{Y}\bm{w}^{*}-\bm{y}_{0}\|_{2}^{2}\\
={} & \hat{\bm{h}}^{\prime}\hat{\bm{\Sigma}}\hat{\bm{h}}+2\bm{w}^{*\prime}\hat{\bm{\Sigma}}\hat{\bm{h}}-2\hat{\bm{\Upsilon}}^{\prime}\hat{\bm{h}}\\
={} & \hat{\bm{h}}'(\hat{\bm{\Sigma}}^{*}+\hat{\bm{\Sigma}}^{\mathrm{e}})\hat{\bm{h}}+2\hat{\bm{h}}'[(\hat{\bm{\Sigma}}^{*}+\hat{\bm{\Sigma}}^{\mathrm{e}})\bm{w}^{*}-\hat{\bm{\Upsilon}}^{*}-\hat{\bm{\Upsilon}}^{\mathrm{e}}]\\
\stackrel{\text{(i)}}{=}{} & \hat{\bm{h}}'(\hat{\bm{\Sigma}}^{*}+\hat{\bm{\Sigma}}^{\mathrm{e}})\hat{\bm{h}}+2\hat{\bm{h}}'(\hat{\bm{\Sigma}}^{\mathrm{e}}\bm{w}^{*}-\hat{\bm{\Upsilon}}^{\mathrm{e}})\\
\leq{} & \hat{\bm{h}}_{\mathcal{G}}'\hat{\bm{\Sigma}}^{\mathrm{co}}\hat{\bm{h}}_{\mathcal{G}}+\|\hat{\bm{\Sigma}}^{\mathrm{e}}\|_{2}\|\hat{\bm{h}}\|_{2}^{2}+2\|\hat{\bm{\Sigma}}^{\mathrm{e}}\|_{1}\|\hat{\bm{h}}\|_{1}\left\Vert \bm{w}^{*}\right\Vert _{\infty}+2\|\hat{\bm{\Upsilon}}^{\mathrm{e}}\|_{\infty}\|\hat{\bm{h}}\|_{1}\\
={} & (I)+(\mathit{II})+(\mathit{III})+(\mathit{IV}),
\end{align*}
where step~(i) uses the fact $\hat{\bm{\Sigma}}^{*}\bm{w}^{*}-\hat{\bm{\Upsilon}}^{*}=\bm{0}_{J}$.

For $(I)$, by \eqref{eq:fact-1} and the second line of \eqref{eq:fact-2-2},
whose derivation does not use $K>r$, we have
\[
(I)=\hat{\bm{h}}_{\mathcal{G}}'\hat{\bm{\Sigma}}^{\mathrm{co}}\hat{\bm{h}}_{\mathcal{G}}\lesssim_{p}\biggl(\eta+K\sqrt{\frac{r\log J}{T_{0}}}\biggr)\|\hat{\bm{h}}\|_{1}\lesssim\eta\|\hat{\bm{h}}\|_{1}=o_{p}(1),
\]
where the last two steps use \cref{assu:eta} and \cref{thm:w_convergence}.
As $\phi_{\min}(\hat{\bm{\Omega}}_{\bm{F}})\gtrsim_{p}1$, it also
follows that $\|\bm{\Lambda}^{\mathrm{co}\prime}\hat{\bm{h}}_{\mathcal{G}}\|_{2}^{2}\leq[\phi_{\min}(\hat{\bm{\Omega}}_{\bm{F}})]^{-1}(I)=o_{p}(1)$.
For $(\mathit{II})$, using the same argument
as in \cref{lem:sampling-error} we can show $\|\hat{\bm{\Sigma}}^{\mathrm{e}}\|_{2}=O_{p}\bigl(J\sqrt{(r\log J)/T_{0}}\bigr)$.
It follows that
\[
(\mathit{II})=\|\hat{\bm{\Sigma}}^{\mathrm{e}}\|_{2}\|\hat{\bm{h}}\|_{2}^{2}=O_{p}\biggl(J\sqrt{\frac{r\log J}{T_{0}}}\biggr)\cdot o_{p}(J^{-1})=o_{p}\biggl(\sqrt{\frac{r\log J}{T_{0}}}\biggr)=o_{p}(1).
\]
For $(\mathit{III})$, by \cref{lem:sampling-error},
\[
(\mathit{III})=2\|\hat{\bm{\Sigma}}^{\mathrm{e}}\|_{1}\|\hat{\bm{h}}\|_{1}\left\Vert \bm{w}^{*}\right\Vert _{\infty}=O_{p}\biggl(J\sqrt{\frac{r\log J}{T_{0}}}\biggr)\cdot o_{p}(1)\cdot O_{p}\left(\frac{K}{J}\right)=o_{p}(1).
\]
For $(\mathit{\mathit{IV}})$, again by \cref{lem:sampling-error},
\[
(\mathit{\mathit{IV}})=2\|\hat{\bm{\Upsilon}}^{\mathrm{e}}\|_{\infty}\|\hat{\bm{h}}\|_{1}=O_{p}\biggl(\sqrt{\frac{r\log J}{T_{0}}}\biggr)\cdot o_{p}(1)=o_{p}(1).
\]
Hence, $(I)+(\mathit{II})+(\mathit{III})+(\mathit{IV})=o_{p}(1)$
and Part~(i) follows.

Part (ii). Let $\tilde{\bm{\Sigma}}:=T_{1}^{-1}\sum_{t\in\mathcal{T}_{1}}\bm{y}_{t}\bm{y}_{t}'$
and $\tilde{\bm{\Upsilon}}:=T_{1}^{-1}\sum_{t\in\mathcal{T}_{1}}\bm{y}_{t}y_{0t}^{N}$
be the post-treatment period sample covariance matrices, and other
matrices like $\tilde{\bm{\Omega}}_{\bm{F}}$, $\tilde{\bm{\Sigma}}^{*}$, $\tilde{\bm{\Sigma}}^{\mathrm{e}}$,
$\tilde{\bm{\Sigma}}^{\mathrm{co}}$, $\tilde{\bm{\Upsilon}}^{*}$,
and $\tilde{\bm{\Upsilon}}^{\mathrm{e}}$ are defined conformably.
Since $y_{jt}^{N}$ for $t\in\mathcal{T}_{1}$ follows the same
DGP as for $t\in\mathcal{T}_{0}$, then $\tilde{\bm{\Omega}}_{\bm{F}}$, $\tilde{\bm{\Sigma}}^{\mathrm{e}}$,
$\tilde{\bm{\Sigma}}^{\mathrm{co}}$, and $\tilde{\bm{\Upsilon}}^{\mathrm{e}}$
have the same properties as $\hat{\bm{\Omega}}_{\bm{F}}$, $\hat{\bm{\Sigma}}^{\mathrm{e}}$, $\hat{\bm{\Sigma}}^{\mathrm{co}}$,
and $\hat{\bm{\Upsilon}}^{\mathrm{e}}$, respectively. Then by a direct
calculation, we have
\begin{align*}
 & \ T_{1}^{-1}\sum_{t\in\mathcal{T}_{1}}(\hat{\bm{w}}'\bm{y}_{t}^{N}-y_{0t}^{N})^{2}-T_{1}^{-1}\sum_{t\in\mathcal{T}_{1}}(\bm{w}^{*}{}'\bm{y}_{t}^{N}-y_{0t}^{N})^{2}\\
= & \ \hat{\bm{h}}^{\prime}\tilde{\bm{\Sigma}}\hat{\bm{h}}+2\bm{w}^{*\prime}\tilde{\bm{\Sigma}}\hat{\bm{h}}-2\tilde{\bm{\Upsilon}}^{\prime}\hat{\bm{h}}\\
= & \ \hat{\bm{h}}'(\hat{\bm{\Sigma}}^{*}+\hat{\bm{\Sigma}}^{\mathrm{e}})\hat{\bm{h}}+2\hat{\bm{h}}'(\hat{\bm{\Sigma}}^{\mathrm{e}}\bm{w}^{*}-\hat{\bm{\Upsilon}}^{\mathrm{e}})\\
 & \ +\hat{\bm{h}}'(\tilde{\bm{\Sigma}}-\hat{\bm{\Sigma}})\hat{\bm{h}}+2\hat{\bm{h}}'[(\tilde{\bm{\Sigma}}-\hat{\bm{\Sigma}})\bm{w}^{*}-(\tilde{\bm{\Upsilon}}-\hat{\bm{\Upsilon}})].
\end{align*}
As shown in Part~(i), $\hat{\bm{h}}'(\hat{\bm{\Sigma}}^{*}+\hat{\bm{\Sigma}}^{\mathrm{e}})\hat{\bm{h}}+2\hat{\bm{h}}'(\hat{\bm{\Sigma}}^{\mathrm{e}}\bm{w}^{*}-\hat{\bm{\Upsilon}}^{\mathrm{e}})=o_{p}(1)$. Hence, we only need to prove that $\hat{\bm{h}}'(\tilde{\bm{\Sigma}}-\hat{\bm{\Sigma}})\hat{\bm{h}}+2\hat{\bm{h}}'[(\tilde{\bm{\Sigma}}-\hat{\bm{\Sigma}})\bm{w}^{*}-(\tilde{\bm{\Upsilon}}-\hat{\bm{\Upsilon}})]=o_{p}(1)$.

Since $\tilde{\bm{\Sigma}}-\hat{\bm{\Sigma}}=\tilde{\bm{\Sigma}}^{*}-\hat{\bm{\Sigma}}^{*}+\tilde{\bm{\Sigma}}^{\mathrm{e}}-\hat{\bm{\Sigma}}^{\mathrm{e}}$
and $r\log(J)/T_{1}=O(1)$, we have
\begin{align*}
\hat{\bm{h}}'(\tilde{\bm{\Sigma}}-\hat{\bm{\Sigma}})\hat{\bm{h}} & =\hat{\bm{h}}_{\mathcal{G}}'(\tilde{\bm{\Sigma}}^{\mathrm{co}}-\hat{\bm{\Sigma}}^{\mathrm{co}})\hat{\bm{h}}_{\mathcal{G}}+\hat{\bm{h}}'(\tilde{\bm{\Sigma}}^{\mathrm{e}}-\hat{\bm{\Sigma}}^{\mathrm{e}})\hat{\bm{h}}\\
 & \leq\|\tilde{\bm{\Omega}}_{\bm{F}}-\hat{\bm{\Omega}}_{\bm{F}}\|_{2}\|\bm{\Lambda}^{\mathrm{co}\prime}\hat{\bm{h}}_{\mathcal{G}}\|_{2}^{2}+\left(\|\tilde{\bm{\Sigma}}^{\mathrm{e}}\|_{2}+\|\hat{\bm{\Sigma}}^{\mathrm{e}}\|_{2}\right)\|\hat{\bm{h}}\|_{2}^{2}\\
 & =o_{p}(1)\cdot o_{p}(1)+\left[O_{p}\biggl(J\sqrt{\frac{r\log J}{T_{1}}}\biggr)+O_{p}\biggl(J\sqrt{\frac{r\log J}{T_{0}}}\biggr)\right]\cdot o_{p}(J^{-1})\\
 & =o_{p}(1).
\end{align*}
Similarly, since $\bm{\Lambda}'\bm{w}^{*}=\bm{\Lambda}^{\mathrm{co}\prime}\bm{w}_{\mathcal{G}}^{*}=\bm{\lambda}_{0}$ and $\mathbb{E}(\tilde{\bm{\Sigma}}^{\mathrm{e}})=\mathbb{E}(\hat{\bm{\Sigma}}^{\mathrm{e}})$ under the same DGP, we have
\begin{align*}
\hat{\bm{h}}'(\tilde{\bm{\Sigma}}-\hat{\bm{\Sigma}})\bm{w}^{*} & =(\bm{\Lambda}^{\mathrm{co}\prime}\hat{\bm{h}}_{\mathcal{G}})'(\tilde{\bm{\Omega}}_{\bm{F}}-\hat{\bm{\Omega}}_{\bm{F}})\bm{\lambda}_{0}\\
 & \quad+\hat{\bm{h}}'\bigl[\tilde{\bm{\Sigma}}^{\mathrm{e}}-\mathbb{E}(\tilde{\bm{\Sigma}}^{\mathrm{e}})\bigr]\bm{w}^{*}-\hat{\bm{h}}'\bigl[\hat{\bm{\Sigma}}^{\mathrm{e}}-\mathbb{E}(\hat{\bm{\Sigma}}^{\mathrm{e}})\bigr]\bm{w}^{*}\\
 & \leq\|\tilde{\bm{\Omega}}_{\bm{F}}-\hat{\bm{\Omega}}_{\bm{F}}\|_{2}\|\bm{\lambda}_{0}\|_{2}\|\bm{\Lambda}^{\mathrm{co}\prime}\hat{\bm{h}}_{\mathcal{G}}\|_{2}\\
 & \quad+\Bigl(\|\tilde{\bm{\Sigma}}^{\mathrm{e}}-\mathbb{E}(\tilde{\bm{\Sigma}}^{\mathrm{e}})\|_{\max}+\|\hat{\bm{\Sigma}}^{\mathrm{e}}-\mathbb{E}(\hat{\bm{\Sigma}}^{\mathrm{e}})\|_{\max}\Bigr)\|\hat{\bm{h}}\|_{1}\|\bm{w}^{*}\|_{1}\\
 & =o_{p}(1)\cdot o_{p}(1)+O_{p}\biggl(\sqrt{\frac{r\log J}{T_{1}}}+\sqrt{\frac{r\log J}{T_{0}}}\biggr)\cdot o_{p}(1)=o_{p}(1),
\end{align*}
where the entry-wise bounds follow from \cref{assu:errors}\ref{enu:Fu} and \ref{enu:uu_dev} as in the proof of \cref{lem:sampling-error}.
To bound the last term $\hat{\bm{h}}'(\tilde{\bm{\Upsilon}}-\hat{\bm{\Upsilon}})$,
notice that $\tilde{\bm{\Upsilon}}-\hat{\bm{\Upsilon}}=\bm{Z}\bm{\Lambda}^{\mathrm{co}}(\tilde{\bm{\Omega}}_{\bm{F}}-\hat{\bm{\Omega}}_{\bm{F}})\bm{\lambda}_{0}+(\tilde{\bm{\Upsilon}}^{\mathrm{e}}-\hat{\bm{\Upsilon}}^{\mathrm{e}})$.
Hence, as $r\log(J)/T_{1}=O(1)$ and \cref{assu:eta} implies $r\log(J)/T_{0}=o(1)$,
we have
\begin{align*}
\hat{\bm{h}}'(\tilde{\bm{\Upsilon}}-\hat{\bm{\Upsilon}}) & =(\bm{\Lambda}^{\mathrm{co}\prime}\hat{\bm{h}}_{\mathcal{G}})'(\tilde{\bm{\Omega}}_{\bm{F}}-\hat{\bm{\Omega}}_{\bm{F}})\bm{\lambda}_{0}+\hat{\bm{h}}'\tilde{\bm{\Upsilon}}^{\mathrm{e}}-\hat{\bm{h}}'\hat{\bm{\Upsilon}}^{\mathrm{e}}\\
 & \leq\|\tilde{\bm{\Omega}}_{\bm{F}}-\hat{\bm{\Omega}}_{\bm{F}}\|_{2}\|\bm{\lambda}_{0}\|_{2}\|\bm{\Lambda}^{\mathrm{co}\prime}\hat{\bm{h}}_{\mathcal{G}}\|_{2}+\left(\|\tilde{\bm{\Upsilon}}^{\mathrm{e}}\|_{\infty}+\|\hat{\bm{\Upsilon}}^{\mathrm{e}}\|_{\infty}\right)\|\hat{\bm{h}}\|_{1}\\
 & \leq o_{p}(1)\cdot o_{p}(1)+\left[O_{p}\biggl(\sqrt{\frac{r\log J}{T_{1}}}\biggr)+O_{p}\biggl(\sqrt{\frac{r\log J}{T_{0}}}\biggr)\right]\cdot o_{p}(1)=o_{p}(1).
\end{align*}
This completes the proof.
\end{proof}

\subsection[Proof of Theorem]{Proof of \cref{thm:wg_convergence} }\label{subsec:proof-thm3}

Similar to \cref{lem:link_two_norms}, we first use the following lemma
that establishes a compatibility condition between $\hat{\bm{w}}_{(g)}-\bm{w}_{(g)}^{*}$
and $\hat{\bm{w}}_{(g),\mathcal{G}}-\bm{w}_{(g),\mathcal{G}}^{*}$.
\begin{lem}
\label{lem:link_two_norms-wg}Suppose \cref{assu:DGP,assu:convex_hull,assu:errors,assu:add_assu}
and \cref{assu:eta}\ref{enu:group-size} hold. We have
\begin{equation}
\|\hat{\bm{w}}_{(g)}-\bm{w}_{(g)}^{*}\|_{2}^{2}\lesssim_{p}\frac{\sqrt{K}\beta_{g}}{\alpha_{g}}\|\hat{\bm{w}}_{(g),\mathcal{G}}-\bm{w}_{(g),\mathcal{G}}^{*}\|_{2}.\label{eq:relationship_two_errors-wg}
\end{equation}
Furthermore, if $K>r$, it holds that
\begin{equation}
\|\hat{\bm{w}}_{(g)}-\bm{w}_{(g)}^{*}\|_{2}^{2}\lesssim_{p}\frac{\sqrt{r}\beta_{g}}{\alpha_{g}}\|\bm{\Lambda}^{\mathrm{co}\prime}(\hat{\bm{w}}_{(g),\mathcal{G}}-\bm{w}_{(g),\mathcal{G}}^{*})\|_{2}.\label{eq:another-compatible-wg}
\end{equation}
\end{lem}
\begin{proof}[Proof of \cref{lem:link_two_norms-wg}]
We employ the same argument as used in the proof of \cref{lem:link_two_norms}.
In view of the fact that $\hat{\bm{w}}_{(g)}$
is the minimizer and that $g$ is $\alpha_{g}$-strongly convex, we
have the basic inequality
\begin{align*}
0 & \geq\sum_{j=1}^{J}g(\hat{w}_{(g),j})-\sum_{j=1}^{J}g(w_{(g),j}^{*})\\
 & \geq\sum_{j=1}^{J}\frac{\mathrm{d}g(w_{(g),j}^{*})}{\mathrm{d}w_{(g),j}^{*}}(\hat{w}_{(g),j}-w_{(g),j}^{*})+\frac{\alpha_{g}}{2}\|\hat{\bm{w}}_{(g)}-\bm{w}_{(g)}^{*}\|_{2}^{2}\qquad\text{w.p.a.1.},
\end{align*}
which implies
\begin{align*}
\frac{\alpha_{g}}{2}\|\hat{\bm{w}}_{(g)}-\bm{w}_{(g)}^{*}\|_{2}^{2} & \leq\left|\sum_{j=1}^{J}\frac{\mathrm{d}g(w_{(g),j}^{*})}{\mathrm{d}w_{(g),j}^{*}}(\hat{w}_{(g),j}-w_{(g),j}^{*})\right|=\left|\sum_{k=1}^{K}\frac{\mathrm{d}g(w_{(g),\mathcal{G}_{k}}^{*}/J_{k})}{\mathrm{d}(w_{(g),\mathcal{G}_{k}}^{*}/J_{k})}\left(\hat{w}_{(g),\mathcal{G}_{k}}-w_{(g),\mathcal{G}_{k}}^{*}\right)\right|\\
 & \leq\left[\sum_{k=1}^{K}\left(\frac{\mathrm{d}g(w_{(g),\mathcal{G}_{k}}^{*}/J_{k})}{\mathrm{d}(w_{(g),\mathcal{G}_{k}}^{*}/J_{k})}\right)^{2}\right]^{1/2}\left[\sum_{k=1}^{K}\left(\hat{w}_{(g),\mathcal{G}_{k}}-w_{(g),\mathcal{G}_{k}}^{*}\right)^{2}\right]^{1/2}\\
 & \leq\sqrt{K}\beta_{g}\cdot\|\hat{\bm{w}}_{(g),\mathcal{G}}-\bm{w}_{(g),\mathcal{G}}^{*}\|_{2}\qquad\text{w.p.a.1},
\end{align*}
where the second line uses the Cauchy-Schwarz inequality and the last
line is by \cref{assu:add_assu} that $g$ is $\beta_{g}$-Lipschitz.

If $K>r$ and $\bm{1}_{K}$ is in the column space of $\bm{\Lambda}^{\mathrm{co}}$,
the FOC with respect to $w_{(g),\mathcal{G}_{k}}^{*}$ is
\[
\frac{\mathrm{d}g(w_{(g),\mathcal{G}_{k}}^{*}/J_{k})}{\mathrm{d}(w_{(g),\mathcal{G}_{k}}^{*}/J_{k})}-\bm{\mu}_{1}^{\prime}\bm{\lambda}_{k}^{\mathrm{co}}=0,
\]
which implies
\[
\bm{\mu}_{1}=\bm{\Lambda}^{\mathrm{co}\dagger}\left[\frac{\mathrm{d}g(w_{(g),\mathcal{G}_{1}}^{*}/J_{1})}{\mathrm{d}(w_{(g),\mathcal{G}_{1}}^{*}/J_{1})},\dots,\frac{\mathrm{d}g(w_{(g),\mathcal{G}_{K}}^{*}/J_{K})}{\mathrm{d}(w_{(g),\mathcal{G}_{K}}^{*}/J_{K})}\right]^{\prime}.
\]
It follows that
\[
\|\bm{\mu}_{1}\|_{2}\leq\|\bm{\Lambda}^{\mathrm{co}\dagger}\|_{2}\cdot\sqrt{K}\beta_{g}\lesssim\sqrt{r}\beta_{g}.
\]
The same result can be obtained for the case when $\bm{1}_{K}$ is
not in the column space of $\bm{\Lambda}^{\mathrm{co}}$. Therefore
\begin{align*}
\frac{\alpha_{g}}{2}\|\hat{\bm{w}}_{(g)}-\bm{w}_{(g)}^{*}\|_{2}^{2} & \leq\|\bm{\mu}_{1}\|_{2}\cdot\|\bm{\Lambda}^{\mathrm{co}\prime}(\hat{\bm{w}}_{(g),\mathcal{G}}-\bm{w}_{(g),\mathcal{G}}^{*})\|_{2}\\
 & \lesssim_{p}\sqrt{r}\beta_{g}\|\bm{\Lambda}^{\mathrm{co}\prime}(\hat{\bm{w}}_{(g),\mathcal{G}}-\bm{w}_{(g),\mathcal{G}}^{*})\|_{2}.
\end{align*}
This completes the proof.
\end{proof}
Next we prove \cref{thm:wg_convergence}.
\begin{proof}[Proof of \cref{thm:wg_convergence}]
 Because the argument for \eqref{eq:upper_bound1}--\eqref{eq:fact-2-2}
only involves the constraints but is independent of the objective
function, they remain true here.

If $K\leq r$, we use \eqref{eq:fact-1}, \eqref{eq:fact-2}, and
\eqref{eq:relationship_two_errors-wg} to deduce
\[
\frac{1}{J^{2}}\|\hat{\bm{h}}\|_{1}^{4}\leq\|\hat{\bm{h}}\|_{2}^{4}\lesssim_{p}\frac{K\beta_{g}^{2}}{\alpha_{g}^{2}}\|\hat{\bm{h}}_{\mathcal{G}}\|_{2}^{2}\lesssim_{p}\frac{K\beta_{g}^{2}\eta}{\alpha_{g}^{2}\phi_{\min}(\hat{\bm{\Sigma}}^{\mathrm{co}})}\|\hat{\bm{h}}\|_{1}\leq\frac{K\sqrt{J}\beta_{g}^{2}\eta}{\alpha_{g}^{2}\phi_{\min}(\hat{\bm{\Sigma}}^{\mathrm{co}})}\|\hat{\bm{h}}\|_{2}.
\]
We conclude that
\begin{align*}
\|\hat{\bm{h}}\|_{1} & \lesssim_{p}\left[\frac{KJ^{2}\beta_{g}^{2}\eta}{\alpha_{g}^{2}\phi_{\min}(\hat{\bm{\Sigma}}^{\mathrm{co}})}\right]^{1/3}=O_{p}\left(\left[\frac{KJ^{2}\beta_{g}^{2}\eta}{\alpha_{g}^{2}}\right]^{1/3}\right)=o_{p}(1),\text{ and }\\
\|\hat{\bm{h}}\|_{2} & \lesssim_{p}\left[\frac{K\sqrt{J}\beta_{g}^{2}\eta}{\alpha_{g}^{2}\phi_{\min}(\hat{\bm{\Sigma}}^{\mathrm{co}})}\right]^{1/3}=O_{p}\left(\left[\frac{KJ^{2}\beta_{g}^{2}\eta}{\alpha_{g}^{2}}\right]^{1/3}\frac{1}{\sqrt{J}}\right)=o_{p}\left(\frac{1}{\sqrt{J}}\right).
\end{align*}

If $K>r$, we use \eqref{eq:fact-1}, \eqref{eq:fact-2-2}, and \eqref{eq:another-compatible-wg}
to get
\[
\frac{1}{J^{2}}\|\hat{\bm{h}}\|_{1}^{4}\leq\|\hat{\bm{h}}\|_{2}^{4}\lesssim_{p}\frac{r\beta_{g}^{2}}{\alpha_{g}^{2}}{\|\bm{\Lambda}^{\mathrm{co}\prime}\hat{\bm{h}}_{\mathcal{G}}\|_{2}^{2}}\lesssim_{p}\frac{r\beta_{g}^{2}\eta}{\alpha_{g}^{2}\phi_{\min}(\hat{\bm{\Omega}}_{\bm{F}})}\|\hat{\bm{h}}\|_{1}\leq\frac{r\sqrt{J}\beta_{g}^{2}\eta}{\alpha_{g}^{2}\phi_{\min}(\hat{\bm{\Omega}}_{\bm{F}})}\|\hat{\bm{h}}\|_{2}.
\]
The desired results then follow.
\end{proof}
\begin{rem}
When $g(x)=x\log x$ and $K\leq r$, the $L_{1}$-convergence of
$\hat{\bm{w}}_{(g)}$ does not rely on \cref{lem:link_two_norms-wg}.
Actually, we can establish the result by an application of the well-known
Pinsker's inequality:
\begin{align}
\left\Vert \hat{\bm{w}}_{(g)}-\bm{w}_{(g)}^{*}\right\Vert_{1}^{2} & \stackrel{\text{(i)}}{\leq}2\sum_{j=1}^{J}\hat{w}_{(g),j}\log\hat{w}_{(g),j}-2\sum_{j=1}^{J}\hat{w}_{(g),j}\log w_{(g),j}^{*}\nonumber \\
 & \stackrel{\text{(ii)}}{\leq}2\sum_{j=1}^{J}w_{(g),j}^{*}\log w_{(g),j}^{*}-2\sum_{j=1}^{J}\hat{w}_{(g),j}\log w_{(g),j}^{*}\nonumber \\
 & \stackrel{\text{(iii)}}{=}2\sum_{k=1}^{K}(w_{(g),\mathcal{G}_{k}}^{*}-\hat{w}_{(g),\mathcal{G}_{k}})\log\frac{w_{(g),\mathcal{G}_{k}}^{*}}{J_{k}}\nonumber \\
 & \leq2\left[\sum_{k=1}^{K}\biggl(\log\frac{w_{(g),\mathcal{G}_{k}}^{*}}{J_{k}}\biggr)^{2}\right]^{1/2}\|\hat{\bm{w}}_{(g),\mathcal{G}}-\bm{w}_{(g),\mathcal{G}}^{*}\|_{2}\qquad\text{w.p.a.1.},\label{eq:fact_entropy1}
\end{align}
where step~(i) uses Pinsker's inequality. Step~(ii) is by the fact
that $\sum_{j=1}^{J}\hat{w}_{(g),j}\log\hat{w}_{(g),j}\leq\sum_{j=1}^{J}w_{(g),j}^{*}\log w_{(g),j}^{*}$
w.p.a.1.~as $\hat{\bm{w}}_{(g)}$ is the minimizer and $\bm{w}_{(g)}^{*}$
satisfies the constraint w.p.a.1. Step~(iii) holds because $w_{(g),j}^{*}=w_{(g),\mathcal{G}_{k}}^{*}/J_{k}$
if $j\in\mathcal{G}_{k}$, and the last line is by the Cauchy-Schwarz
inequality. By \eqref{eq:fact_entropy1} and the same argument for
\eqref{eq:fact-1} and \eqref{eq:fact-2}, we have for $K\leq r$,
\[
\|\hat{\bm{h}}\|_{1}^{4}\lesssim_{p}\left[\sum_{k=1}^{K}\biggl(\log\frac{w_{(g),\mathcal{G}_{k}}^{*}}{J_{k}}\biggr)^{2}\right]\|\hat{\bm{h}}_{\mathcal{G}}\|_{2}^{2}\lesssim_{p}\frac{\eta\sum_{k=1}^{K}[\log(w_{(g),\mathcal{G}_{k}}^{*}/J_{k})]^{2}}{\phi_{\min}(\hat{\bm{\Sigma}}^{\mathrm{co}})}\|\hat{\bm{h}}\|_{1}.
\]
Hence, if $\eta=o(1/[K(\log J)^{2}])$, then
\begin{align*}
\|\hat{\bm{h}}\|_{1} & \lesssim_{p}\left[\frac{\eta\sum_{k=1}^{K}[\log(w_{(g),\mathcal{G}_{k}}^{*}/J_{k})]^{2}}{\phi_{\min}(\hat{\bm{\Sigma}}^{\mathrm{co}})}\right]^{1/3}\\
 & =O_{p}\left(\{\eta K[\log(K/J)]^{2}\}^{1/3}\right)=O_{p}\left([\eta K(\log J)^{2}]^{1/3}\right)=o_{p}(1).
\end{align*}

We outline this alternative proof here for completeness. However,
it is difficult to establish similar results by Pinsker's inequality
for other criterion functions like $g(x)=-\log x$. Moreover, the
case $K>r$ and the stronger $L_{2}$-convergence rate rely crucially
on \cref{lem:link_two_norms-wg}.
\end{rem}

\subsection[Proof of Theorem]{Proof of \cref{thm:att} }\label{subsec:proof-att}

\begin{proof}[Proof of \cref{thm:att}]
Let $\hat{\bm{h}}:=\hat{\bm{w}}-\bm{w}^{*}$ and $\hat{\bm{h}}_{\mathcal{G}}:=\bm{Z}'\hat{\bm{h}}$. Substituting $y_{jt}^{N}=\bm{\lambda}_{j}'\bm{f}_{t}+u_{jt}$ into $\hat{\tau}=T_{1}^{-1}\sum_{t\in\mathcal{T}_{1}}(y_{0t}^{I}-\sum_{j\in[J]}\hat{w}_{j}y_{jt}^N)$ and then by \cref{assu:convex_hull} gives
\begin{align*}
\hat{\tau}-\tau_{0}
&= \frac{1}{T_{1}}\sum_{t\in\mathcal{T}_{1}}\Biggl(y_{0t}^{N}-\sum_{j\in[J]}\hat{w}_{j}y_{jt}^{N}\Biggr) \\
&= \underbrace{\frac{1}{T_1}\sum_{t\in\mathcal{T}_{1}}\bm{f}_{t}'\bm{\Lambda}^{\prime}(\bm{w}^{*}-\hat{\bm{w}})}_{B_1}+\underbrace{\frac{1}{T_1}\sum_{t\in\mathcal{T}_{1}}u_{0t}}_{B_{2}}-\underbrace{\frac{1}{T_1}\sum_{t\in\mathcal{T}_{1}}\sum_{j\in[J]}\hat{w}_{j}u_{jt}}_{B_{3}}.
\end{align*}

We first bound the factor-loading imbalance in $B_{1}$ directly from the relaxed moment condition. Let $\bm{A}:=(\bm{Z}'\bm{Z})^{-1}\bm{Z}'$. Recall from \eqref{eq:upper_bound1} that
\[
\left\|\hat{\bm{\Sigma}}^{*}\hat{\bm{h}}+\hat{\bm{\Sigma}}^{\mathrm{e}}\hat{\bm{w}}-\hat{\bm{\Upsilon}}^{\mathrm{e}}\right\|_{\infty}\leq\eta .
\]
Since $\|\bm{A}\|_{\infty}=1$, multiplying the preceding display by $\bm{A}$ yields
\begin{align*}
\left\|\hat{\bm{\Sigma}}^{\mathrm{co}}\hat{\bm{h}}_{\mathcal{G}}\right\|_{\infty}
&\leq \eta+\left\|\bm{A}\left(\hat{\bm{\Sigma}}^{\mathrm{e}}\hat{\bm{w}}-\hat{\bm{\Upsilon}}^{\mathrm{e}}\right)\right\|_{\infty},
\end{align*}
where we use $\bm{A}\hat{\bm{\Sigma}}^{*}\hat{\bm{h}}=\hat{\bm{\Sigma}}^{\mathrm{co}}\hat{\bm{h}}_{\mathcal{G}}$. For the second term, $\|\hat{\bm{w}}\|_{1}=1$, $\min_{k\in[K]}J_k\gtrsim J/K$, and \cref{lem:sampling-error} imply
\begin{align*}
\left\|\bm{A}\hat{\bm{\Sigma}}^{\mathrm{e}}\hat{\bm{w}}\right\|_{\infty}
&\leq \frac{1}{\min_{k\in[K]}J_k}\left\|\hat{\bm{\Sigma}}^{\mathrm{e}}\hat{\bm{w}}\right\|_{1}
\leq \frac{1}{\min_{k\in[K]}J_k}\left\|\hat{\bm{\Sigma}}^{\mathrm{e}}\right\|_{1}
=O_p\!\left(K\sqrt r\,\phi_{J,T_{0}}\right),\\
\left\|\bm{A}\hat{\bm{\Upsilon}}^{\mathrm{e}}\right\|_{\infty}
&\leq \left\|\hat{\bm{\Upsilon}}^{\mathrm{e}}\right\|_{\infty}
=O_p\!\left(\sqrt r\,\phi_{J,T_{0}}\right).
\end{align*}
By \cref{assu:eta}, $K\sqrt r\,\phi_{J,T_{0}}=o(\eta)$. Thus,
\[
\left\|\hat{\bm{\Sigma}}^{\mathrm{co}}\hat{\bm{h}}_{\mathcal{G}}\right\|_{\infty}=O_p(\eta).
\]

We convert this moment bound into a loading-imbalance bound. If $K\leq r$, then by \cref{lem:bounds_for_mu}, $\phi_{\min}(\hat{\bm{\Sigma}}^{\mathrm{co}})$ is bounded away from zero w.p.a.1., and \cref{assu:DGP} gives $\|\bm{\Lambda}^{\mathrm{co}}\|_{2}=O_p(1)$. Therefore,
\[
\|\bm{\Lambda}'\hat{\bm{h}}\|_2
=\|\bm{\Lambda}^{\mathrm{co}\prime}\hat{\bm{h}}_{\mathcal{G}}\|_{2}
\leq \|\bm{\Lambda}^{\mathrm{co}}\|_{2}\|(\hat{\bm{\Sigma}}^{\mathrm{co}})^{-1}\|_{2}\|\hat{\bm{\Sigma}}^{\mathrm{co}}\hat{\bm{h}}_{\mathcal{G}}\|_{2}
=O_p(\sqrt K\,\eta).
\]
If $K>r$, write $\bm{x}:=\bm{\Lambda}^{\mathrm{co}\prime}\hat{\bm{h}}_{\mathcal{G}}$. Then $\hat{\bm{\Sigma}}^{\mathrm{co}}\hat{\bm{h}}_{\mathcal{G}}=\bm{\Lambda}^{\mathrm{co}}\hat{\bm{\Omega}}_{\bm F}\bm{x}$. By \cref{assu:DGP}, $\varsigma_{\min}(\bm{\Lambda}^{\mathrm{co}})\asymp \sqrt{K/r}$ and $\phi_{\min}(\hat{\bm{\Omega}}_{\bm F})$ is bounded away from zero w.p.a.1., so that
\[
\|\hat{\bm{\Sigma}}^{\mathrm{co}}\hat{\bm{h}}_{\mathcal{G}}\|_{2}\gtrsim_p \sqrt{K/r}\,\|\bm{x}\|_2.
\]
Since $\|\hat{\bm{\Sigma}}^{\mathrm{co}}\hat{\bm{h}}_{\mathcal{G}}\|_{2}\leq\sqrt K\|\hat{\bm{\Sigma}}^{\mathrm{co}}\hat{\bm{h}}_{\mathcal{G}}\|_{\infty}=O_p(\sqrt K\,\eta)$, we obtain
\[
\|\bm{\Lambda}'\hat{\bm{h}}\|_2=\|\bm{x}\|_2=O_p(\sqrt r\,\eta).
\]
Combining the two cases gives
\[
\|\bm{\Lambda}'(\hat{\bm{w}}-\bm{w}^{*})\|_2=O_p\!\left(\sqrt{K\wedge r}\,\eta\right).
\]
Under \cref{assu:additional}, $\|T_{1}^{-1}\sum_{t\in\mathcal{T}_{1}}\bm{f}_{t}\|_2=O_p(\sqrt r)$. Hence,
\[
|B_1|=O_p\bigl(\sqrt{r(K\wedge r)}\,\eta\bigr).
\]

For $B_{2}$, the sequence $u_{0t}$ is zero-mean, strictly stationary, and $\beta$-mixing with finite variance under \cref{assu:errors,assu:additional}. By the CLT, $B_{2}=O_{p}(T_{1}^{-1/2})$.

For $B_{3}$, by \cref{thm:w_convergence},
\[
\|\hat{\bm{w}}\|_{2}\leq\|\bm{w}^{*}\|_{2}+\|\hat{\bm{w}}-\bm{w}^{*}\|_{2}=O_p(\sqrt{K/J}),
\]
where $\|\bm{w}^{*}\|_{2}=O_p(\sqrt{K/J})$ follows from the within-group equal form of $\bm{w}^{*}$ and \cref{assu:eta}\ref{enu:group-size}. Therefore,
\begin{align*}
|B_{3}| &\leq \|\hat{\bm{w}}\|_2\cdot\left[\sum_{j\in[J]}\Biggl(\frac{1}{T_{1}}\sum_{t\in\mathcal{T}_{1}}u_{jt}\Biggr)^{2}\right]^{1/2} \\
&\leq O_p(\sqrt{K/J})\cdot O_p(\sqrt{J/T_1}) = O_p(\sqrt{K/T_1}),
\end{align*}
where we use the fact that
\[
\frac1J\sum_{j\in[J]}\Biggl(\frac{1}{T_{1}}\sum_{t\in\mathcal{T}_{1}}u_{jt}\Biggr)^{2} = O_p\biggl(\frac{1}{T_1}\biggr)
\]
since by the mixing and tail conditions of \cref{assu:additional}
\[
\frac1J\sum_{j\in[J]}\mathbb{E}\left[\Biggl(\frac{1}{T_{1}}\sum_{t\in\mathcal{T}_{1}}u_{jt}\Biggr)^{2}\right] \leq \frac{1}{T_1}\Biggl[\frac1J\sum_{j\in[J]}\sum_{\ell=-\infty}^{\infty}|\mathbb{E}(u_{jt}u_{j,t-\ell})|\Biggr] = O\biggl(\frac{1}{T_1}\biggr).
\]
Since $K\ge1$, the term $O_p(T_1^{-1/2})$ is absorbed by $O_p(\sqrt{K/T_1})$. In consequence,
\[
\hat{\tau} - \tau_{0} = O_p\bigl(\sqrt{r(K\wedge r)}\,\eta\bigr)+O_p\bigl(\sqrt{K/T_{1}}\bigr).
\]
This completes the proof.
\end{proof}

\subsection[Proof of Theorem]{Proof of \cref{thm:asym-dist}}\label{subsec:proof-asym-dist}
\begin{proof}[Proof of \cref{thm:asym-dist}]
We verify the conditions of \citet*[CWZ26, henceforth]{chernozhukov_t-test_2024}.
Their Assumption 1 (covariance stationarity) is satisfied, and Assumption
2 holds in view of our \cref{thm:wg_convergence}. It remains to verify
the conditions in CWZ26's Assumption~3 one by one.

Let $\bm{x}_{t}:=(y_{1t},y_{2t},\dots,y_{Jt})^{\prime}=\bm{\Lambda}\bm{f}_{t}+\bm{u}_{t}$, where $\bm{u}_{t}:=(u_{1t},\dots,u_{Jt})^{\prime}$.
We express $y_{0t}^{N}$ as
\[
y_{0t}^{N}=\bm{x}_{t}^{\prime}\bm{w}^{*}+v_{t},
\]
where $\bm{w}^{*}$ is the oracle weight defined as the solution to
\eqref{eq:oracle_relaxation_problem} (or \eqref{eq:g-oracle} for a generic objective function) and
\[
v_{t}:=y_{0t}^{N}-\bm{x}_{t}^{\prime}\bm{w}^{*}=(\bm{\lambda}_{0}-\bm{\Lambda}^{\prime}\bm{w}^{*})'\bm{f}_{t}+(u_{0t}-\bm{u}_{t}^{\prime}\bm{w}^{*})
\]
 is the prediction error.

Let $\mathscr{T}\subseteq[T]$. For a random variable $A$, define
$\tilde{A}:=A-\mathbb{E}(A)$ as its centered version. We have
\begin{align}
 & \frac{1}{|\mathscr{T}|}\mathbb{E}\left[\left(\sum_{t\in\mathscr{T}}\tilde{\bm{x}}_{t}\right)\left(\sum_{t\in\mathscr{T}}\tilde{\bm{x}}_{t}\right)^{\prime}\right]\nonumber \\
=\  & \bm{\Lambda}\mathbb{E}\left[\frac{1}{|\mathscr{T}|}\left(\sum_{t\in\mathscr{T}}\tilde{\bm{f}}_{t}\right)\left(\sum_{t\in\mathscr{T}}\tilde{\bm{f}}_{t}\right)^{\prime}\right]\bm{\Lambda}^{\prime}+\mathbb{E}\left[\frac{1}{|\mathscr{T}|}\left(\sum_{t\in\mathscr{T}}\bm{u}_{t}\right)\left(\sum_{t\in\mathscr{T}}\bm{u}_{t}\right)^{\prime}\right].\label{eq:covariance}
\end{align}
By \cref{assu:additional}, the restriction on $\beta$-mixing coefficient
and uniform boundedness of $\delta$-th moment guarantee that the long-run
variances of $\bm{f}_{t}$ and $\bm{u}_{t}$ exist, since by Davydov's
inequality we have
\begin{align*}
\sum_{t=1}^{\infty}|\mathbb{E}(\tilde{f}_{\ell0}\tilde{f}_{kt})| & \leq C_{\delta}\bigl[\mathbb{E}(|\tilde{f}_{\ell0}|^{\delta})\mathbb{E}(|\tilde{f}_{k0}|^{\delta})\bigr]^{1/\delta}\sum_{t=1}^{\infty}\beta_{\mathrm{mix}}(t)^{1-2/\delta}<\infty\\
\sum_{t=1}^{\infty}|\mathbb{E}(u_{i0}u_{jt})| & \leq C_{\delta}\bigl[\mathbb{E}(|u_{i0}|^{\delta})\mathbb{E}(|u_{jt}|^{\delta})\bigr]^{1/\delta}\sum_{t=1}^{\infty}\beta_{\mathrm{mix}}(t)^{1-2/\delta}<\infty,
\end{align*}
for some absolute constant $C_{\delta}$. Here we take $\delta=4+\varepsilon$
and use the condition $\beta_{\mathrm{mix}}(t)\leq Ct^{-\kappa}$ for
some $\kappa\geq2$. Hence the entries of the two long-run covariance matrices in \eqref{eq:covariance} are bounded, and by the Gershgorin circle theorem their largest eigenvalues are $O(r)$ and $O(J)$, respectively. As the factor part is multiplied by $\bm{\Lambda}$, the largest eigenvalue of \eqref{eq:covariance} grows with $J$, so CWZ26's Assumption~3.1 does not hold as stated. CWZ26 use it only to bound $\sum_{t\in A_{k}}\tilde{\bm{x}}_{t}'\bm{\Delta}_{(k)}$ by $O_{p}(\sqrt{|A_{k}|}\,\|\bm{\Delta}_{(k)}\|_{2})$, where $\bm{\Delta}_{(k)}:=\hat{\bm{w}}_{(k)}-\bm{w}^{*}$ and $A_{k}$ is a set of periods independent of $\bm{\Delta}_{(k)}$ after their coupling. With the two eigenvalue bounds, the same argument gives
\[
\sum_{t\in A_{k}}\tilde{\bm{x}}_{t}'\bm{\Delta}_{(k)}=O_{p}\Bigl(\sqrt{|A_{k}|}\bigl(\sqrt{r}\,\|\bm{\Lambda}'\bm{\Delta}_{(k)}\|_{2}+\sqrt{J}\,\|\bm{\Delta}_{(k)}\|_{2}\bigr)\Bigr)=o_{p}\bigl(\sqrt{|A_{k}|}\bigr),
\]
since $\|\bm{\Lambda}'\bm{\Delta}_{(k)}\|_{2}=O_{p}(\sqrt{K\wedge r}\,\eta)$ as in the proof of \cref{thm:att}, $\sqrt{J}\|\bm{\Delta}_{(k)}\|_{2}=o_{p}(1)$ by \cref{thm:w_convergence} (or \cref{thm:wg_convergence} for a generic $g$), and $\sqrt{r(K\wedge r)}\,\eta\to0$. The rest of their proof then goes through. CWZ26's Assumption~3.2
is implied by \cref{assu:additional}\ref{enu:tail}. In
particular, if we take the sequence $\rho_{T}$ to be $(JT)^{1/(4+\varepsilon)}\log(JT)$,
then by the union bound, Markov's inequality, and $\mathbb{E}(|\tilde{x}_{jt}|^{4+\varepsilon})=O(r^{2+\varepsilon/2})$ since $\|\bm{\lambda}_{j}\|_{1}\leq\sqrt{r}\,\bar{c}$, it holds that
\[
\mathbb{P}\biggl(\max_{j,t}|\tilde{x}_{jt}|\geq\rho_{T}\biggr)\leq\sum_{j,t}\frac{\mathbb{E}(|\tilde{x}_{jt}|^{4+\varepsilon})}{\rho_{T}^{4+\varepsilon}}=O\biggl(\frac{r^{2+\varepsilon/2}}{[\log(JT)]^{4+\varepsilon}}\biggr)\to0
\]
as $r=o((\log T)^{2})$. For CWZ26's Assumption~3.3, we can take $\gamma_{T}=\log(T)$ for
example. Then clearly we have $\rho_{T}\gamma_{T}=o(\sqrt{T_{0}\wedge T_{1}})$
provided $J=O(T)$ and $T_{0}\asymp T_{1}$. Thus their Assumption
3.3 holds. CWZ26's Assumption 3.4 is clearly implied by \cref{assu:additional}.

Therefore, CWZ26's Assumptions~1, 2, and 3.2--3.4 are satisfied under our assumptions, and the bound above replaces their Assumption~3.1.
We can then invoke their Theorem~2 to conclude our statement.
\end{proof}

\section{Feature Engineering: Standardized SCM-relaxation}\label{sec:Standardized-SCM-relaxation}

In practice, the magnitude of outcomes often varies across control
units, necessitating data standardization as a pre-processing step.
Typically, our theoretical results can be extended to allow for heterogeneous
scales:
\begin{equation}
y_{jt}=\theta_{j}(\bm{\lambda}_{j}'\bm{f}_{t}+u_{jt}),\quad j\in\{0\}\cup[J],\label{eq:heterogeneous_scales}
\end{equation}
if we carry out data standardization. Let $\hat{\sigma}_{j}^{2}:=(T_{0}-1)^{-1}\sum_{t\in\mathcal{T}_{0}}(y_{jt}-\bar{y}_{j})^{2}$
be the sample variance of pre-treatment outcomes for unit $j\in\{0\}\cup[J]$,
where $\bar{y}_{j}:=T_{0}^{-1}\sum_{t\in\mathcal{T}_{0}}y_{jt}$ is
the sample mean. Let $\hat{\bm{\sigma}}=(\hat{\sigma}_{1},\dots,\hat{\sigma}_{J})'$
and $\mathrm{diag}(\hat{\bm{\sigma}})$ be the diagonal matrix with
$\hat{\bm{\sigma}}$ on the diagonal. The covariance matrices for
standardized data are $\hat{\bm{\Sigma}}^{\mathrm{s}}=\mathrm{diag}(\hat{\bm{\sigma}})^{-1}\hat{\bm{\Sigma}}\mathrm{diag}(\hat{\bm{\sigma}})^{-1}$
and $\hat{\bm{\Upsilon}}^{\mathrm{s}}=\hat{\sigma}_{0}^{-1}\mathrm{diag}(\hat{\bm{\sigma}})^{-1}\hat{\bm{\Upsilon}}$,
where the superscript ``s'' stands for ``standardized.'' The standardized
SCM-relaxation estimates the weights by solving:
\begin{equation}
\min_{\bm{w}\in\Delta_{J}}{}\sum_{j\in[J]}g(w_{j})\quad\text{s.t. }\left\Vert \hat{\bm{\Sigma}}^{\mathrm{s}}\bm{w}-\hat{\bm{\Upsilon}}^{\mathrm{s}}\right\Vert _{\infty}\leq\eta,\label{eq:standardized-SCM-relaxation}
\end{equation}
and the solution is denoted as $\hat{\bm{w}}_{(g)}^{\mathrm{s}}$.
The counterfactual outcomes of unit $0$ in the post-treatment periods
are estimated via
\[
\hat{y}_{0t}^{N}=\left[\hat{\sigma}_{0}\mathrm{diag}(\hat{\bm{\sigma}})^{-1}\hat{\bm{w}}_{(g)}^{\mathrm{s}}\right]'\bm{y}_{t}.
\]

We build theoretical connections between SCM-relaxation and the standardized
version. Define the oracle weight for standardized SCM-relaxation,
denoted as $\bm{w}_{(g)}^{*,\mathrm{s}}$, similarly to \eqref{eq:g-oracle}.
Hence, $\bm{w}_{(g)}^{*,\mathrm{s}}$ satisfies
\[
\mathrm{diag}(\hat{\bm{\sigma}}^{*})^{-1}\hat{\bm{\Sigma}}^{*}\mathrm{diag}(\hat{\bm{\sigma}}^{*})^{-1}\bm{w}_{(g)}^{*,\mathrm{s}}-\hat{\sigma}_{0}^{*-1}\mathrm{diag}(\hat{\bm{\sigma}}^{*})^{-1}\hat{\bm{\Upsilon}}^{*}=\bm{0}_{J},
\]
where $\hat{\sigma}_{j}^{*}$ is the standard deviation of the noiseless
part $y_{jt}^{*}=\theta_{j}\bm{\lambda}_{j}'\bm{f}_{t}$ for unit $j$.
It is verified that $\hat{\sigma}_{0}^{*}\mathrm{diag}(\hat{\bm{\sigma}}^{*})^{-1}\bm{w}_{(g)}^{*,\mathrm{s}}$
satisfies $\hat{\bm{\Sigma}}^{*}\bm{w}-\hat{\bm{\Upsilon}}^{*}=\bm{0}_{J}$,
and $\bm{w}_{(g)}^{*,\mathrm{s}}$ has a group structure. Under the heterogeneous scales model \eqref{eq:heterogeneous_scales},
suppose $\hat{\sigma}_{0}^{*}/\hat{\sigma}_{j}^{*}=\hat{\sigma}_{0}/\hat{\sigma}_{j}$
for $j\in[J]$, which holds approximately when the idiosyncratic errors
are small relative to the common components. Using the same argument
in the proof of \cref{thm:wg_convergence}, we can show that
\[
\left\Vert \hat{\bm{w}}_{(g)}^{\mathrm{s}}-\bm{w}_{(g)}^{*,\mathrm{s}}\right\Vert _{2}=o_{p}(J^{-1/2}),\text{ and }\left\Vert \hat{\bm{w}}_{(g)}^{\mathrm{s}}-\bm{w}_{(g)}^{*,\mathrm{s}}\right\Vert _{1}=o_{p}(1).
\]
Then the prediction error
\[
\underbrace{\left[\hat{\sigma}_{0}\mathrm{diag}(\hat{\bm{\sigma}})^{-1}\hat{\bm{w}}_{(g)}^{\mathrm{s}}\right]'\bm{y}_{t}}_{\hat{y}_{0t}^{N}}-\underbrace{\left[\hat{\sigma}_{0}^{*}\mathrm{diag}(\hat{\bm{\sigma}}^{*})^{-1}\bm{w}_{(g)}^{*,\mathrm{s}}\right]'\bm{y}_{t}}_{\text{oracle prediction}}=\left[\hat{\sigma}_{0}^{*}\mathrm{diag}(\hat{\bm{\sigma}}^{*})^{-1}(\hat{\bm{w}}_{(g)}^{\mathrm{s}}-\bm{w}_{(g)}^{*,\mathrm{s}})\right]'\bm{y}_{t},
\]
just depends on $\hat{\bm{w}}_{(g)}^{\mathrm{s}}-\bm{w}_{(g)}^{*,\mathrm{s}}$. This
implies that the standardized SCM-relaxation can achieve near-oracle prediction
risk under the same conditions imposed in \cref{thm:wg_convergence}.

\section{Additional Simulation Results}\label{sec:appendix-simulation}

This appendix reports additional simulation results.

\phantomsection\label{para:eta-sensitivity}
To assess whether CV selects a near-optimal $\eta$, we evaluate $L_{2}$-SCM-relaxation at every point of the CV grid in the correctly specified exact group design with $K=4$, and compute RMSE\textsuperscript{o} for each combination of $J$ and $T_0$. \Cref{tab:eta-sensitivity} shows that RMSE\textsuperscript{o} is U-shaped in $\eta$, which reflects the bias-variance trade-off, and that the $\eta$ that CV selects most often is the optimal $\eta$, the grid point minimizing RMSE\textsuperscript{o}, or adjacent to it. The optimal $\eta$ falls as $T_{0}$ grows and rises slightly with $J$, which agrees with the order $\phi_{J,T_{0}}=\sqrt{(\log J)/T_{0}}$ of the sampling error that $\eta$ has to dominate.

In the correctly specified exact group design, \cref{tab:att_bias} shows that the absolute bias of the ATT estimators is small for all methods, below $0.032$ in all settings. \Cref{tab:ci_ratio} presents the ratios of the average lengths of nominal 90\% CIs to those of SCM CIs. Compared with SCM CIs, $L_{2}$- and entropy-SCM-relaxation CIs are shorter in all settings, and those of $L_2$-SCM-relaxation are 2\% to 5\% shorter. Ridge and Group Lasso also give shorter CIs, while the CIs of EL, Lasso, AugSC, and GSC are close to those of SCM. PenSC, SDID, and fsPDA give longer CIs than SCM in all settings.

Under misspecification, we report RMSE\textsuperscript{$\star$} for the approximate group and general factor designs in \cref{tab:rmse_truth_approx_misspec,tab:rmse_truth_general_misspec}. The patterns are similar to those in \cref{tab:rmse_truth_misspec} under exact group structures, and $L_{2}$- and entropy-SCM-relaxation perform better than the other methods in almost all settings. SDID performs best in the general factor design with $(J,T_0)=(120,30)$, where its time weights help the prediction under misspecification.

\renewcommand{\arraystretch}{0.85}
\begin{table}[!htbp]
\centering
\caption{Tuning parameter sensitivity of $L_{2}$-SCM-relaxation (exact group, $K=4$)}
\label{tab:eta-sensitivity}
\small
\sbox{\tabbox}{\hdreset\hdbump{$\eta$}\hdbump{RMSE\textsuperscript{o}}\hdbump{\textbf{8.888}}\setlength{\tabcolsep}{4pt}%
\setlength{\tabcolsep}{3pt}
\begin{tabular}{@{}cS[table-format=1.3]S[table-format=1.3]S[table-format=1.3]S[table-format=1.3]S[table-format=1.3]S[table-format=1.3]@{}}
\toprule
 & \multicolumn{2}{c}{$T_0 = 30$} & \multicolumn{2}{c}{$T_0 = 60$} & \multicolumn{2}{c}{$T_0 = 90$}  \\
\cmidrule(lr){2-3}\cmidrule(lr){4-5}\cmidrule(lr){6-7}
 & \hd{$\eta$} & \hd{RMSE\textsuperscript{o}} & \hd{$\eta$} & \hd{RMSE\textsuperscript{o}} & \hd{$\eta$} & \hd{RMSE\textsuperscript{o}} \\
\midrule
$J = 40$ & 6.300 & 1.584 & 5.872 & 1.578 & 5.978 & 1.608 \\
 & 4.393 & 1.211 & 4.022 & 1.172 & 3.992 & 1.170 \\
 & 3.085 & .953 & 2.770 & .896 & 2.681 & .878 \\
 & 2.181 & .768 & 1.918 & .700 & 1.811 & .676 \\
 & 1.552 & .630 & 1.335 & .553 & 1.231 & .524 \\
 & 1.113 & .520 & .934 & .439 & .842 & .405 \\
 & .803 & .433 & .658 & .347 & .580 & .312 \\
 & .583 & .372 & .466 & .281 & .401 & .242 \\
 & \textbf{\textit{.426}} & .342 & \textbf{\textit{.332}} & .246 & \textbf{\textit{.280}} & .202 \\
 & .314 & .365 & .237 & .268 & .196 & .223 \\
\midrule
$J = 80$ & 6.304 & 1.563 & 6.005 & 1.614 & 5.697 & 1.571 \\
 & 4.410 & 1.190 & 4.104 & 1.190 & 3.840 & 1.143 \\
 & 3.103 & .930 & 2.820 & .900 & 2.604 & .851 \\
 & 2.196 & .744 & 1.950 & .700 & 1.776 & .650 \\
 & 1.564 & .602 & 1.356 & .546 & 1.218 & .498 \\
 & 1.120 & .484 & .948 & .424 & .840 & .379 \\
 & .806 & .394 & .667 & .329 & .583 & .287 \\
 & .584 & .339 & .472 & .262 & .407 & .223 \\
 & \textbf{.426} & .325 & \textbf{\textit{.337}} & .230 & \textbf{\textit{.286}} & .194 \\
 & \textit{.312} & .360 & .241 & .261 & .202 & .225 \\
\midrule
$J = 120$ & 6.447 & 1.573 & 6.099 & 1.611 & 5.975 & 1.582 \\
 & 4.513 & 1.197 & 4.175 & 1.193 & 4.035 & 1.156 \\
 & 3.176 & .935 & 2.873 & .905 & 2.738 & .864 \\
 & 2.248 & .748 & 1.988 & .698 & 1.868 & .660 \\
 & 1.601 & .605 & 1.383 & .537 & 1.280 & .507 \\
 & 1.146 & .491 & .968 & .414 & .882 & .385 \\
 & .825 & .404 & .681 & .321 & .611 & .292 \\
 & .597 & .345 & .482 & .256 & .425 & .225 \\
 & \textbf{\textit{.435}} & .326 & \textbf{\textit{.343}} & .230 & \textbf{\textit{.298}} & .192 \\
 & .319 & .359 & .246 & .262 & .209 & .225 \\
\bottomrule
\end{tabular}
}\usebox{\tabbox}
\vspace{.5\baselineskip}
\begin{minipage}[t]{\wd\tabbox}\footnotesize
\textit{Notes}: 1.~The column $\eta$ reports the average of each grid point over the replications. 2.~Bold indicates the $\eta$ minimizing RMSE\textsuperscript{o}. 3.~Italic indicates the $\eta$ that cross-validation selects most often. 4.~Bold italic indicates that the two coincide.
\end{minipage}
\end{table}

\begin{table}[!htbp]
\centering
\small
\caption{Absolute bias ($\times 100$) of ATT (exact group)}
\label{tab:att_bias}
\smallskip
\sbox{\tabbox}{\resizebox{\textwidth}{!}{\hdreset\hdbump{$L_2$}\hdbump{EL}\hdbump{Entropy}\hdbump{SCM}\hdbump{Lasso}\hdbump{Ridge}\hdbump{GLasso}\hdbump{fsPDA}\hdbump{AugSC}\hdbump{PenSC}\hdbump{SDID}\hdbump{GSC}\hdbump{\textbf{8.888}}\setlength{\tabcolsep}{4pt}%
\begin{tabular}{rrS[table-format=1.3]S[table-format=1.3]S[table-format=1.3]S[table-format=1.3]S[table-format=1.3]S[table-format=1.3]S[table-format=1.3]S[table-format=1.3]S[table-format=1.3]S[table-format=1.3]S[table-format=1.3]S[table-format=1.3]}
\toprule
\multicolumn{2}{c}{} & \multicolumn{3}{c}{Relaxation} & \multicolumn{9}{c}{} \\
\cmidrule(lr){3-5}
$J$ & $T_0$ & \hd{$L_2$} & \hd{EL} & \hd{Entropy} & \hd{SCM} & \hd{Lasso} & \hd{Ridge} & \hd{GLasso} & \hd{fsPDA} & \hd{AugSC} & \hd{PenSC} & \hd{SDID} & \hd{GSC} \\
\midrule
\multicolumn{14}{c}{\footnotesize\textbf{Panel A. $\bm{K<r}$}} \\
40 & 30 & .427 & .342 & .389 & .449 & .181 & .354 & .368 & 1.665 & .449 & .971 & 1.076 & .096 \\
 & 60 & .029 & .134 & .029 & .076 & .155 & .116 & .102 & 1.076 & .053 & .154 & .908 & .010 \\
 & 90 & .212 & .396 & .273 & .291 & .236 & .224 & .264 & .998 & .266 & .422 & .173 & .249 \\
\addlinespace[.45em]
80 & 30 & .137 & .247 & .196 & .038 & .069 & .149 & .042 & .312 & .041 & .607 & .591 & .206 \\
 & 60 & .488 & .567 & .492 & .570 & .591 & .538 & .567 & 1.008 & .511 & .572 & .351 & .429 \\
 & 90 & .363 & .459 & .344 & .341 & .363 & .483 & .495 & .776 & .280 & .500 & 1.479 & .204 \\
\addlinespace[.45em]
120 & 30 & .083 & .293 & .206 & .021 & .220 & .273 & .169 & .265 & .007 & .130 & .282 & .317 \\
 & 60 & .263 & .196 & .219 & .457 & .453 & .376 & .399 & 1.100 & .462 & .865 & 1.278 & .254 \\
 & 90 & .559 & .474 & .512 & .681 & .743 & .642 & .650 & 1.154 & .694 & .359 & .666 & .674 \\
\midrule
\multicolumn{14}{c}{\footnotesize\textbf{Panel B. $\bm{K=r}$}} \\
40 & 30 & 1.793 & 1.913 & 1.858 & 1.993 & 2.088 & 2.008 & 1.988 & 1.139 & 2.100 & 1.416 & 1.439 & 1.913 \\
 & 60 & .228 & .422 & .218 & .442 & .433 & .305 & .257 & .798 & .427 & .746 & .517 & .055 \\
 & 90 & .214 & .376 & .280 & .555 & .404 & .371 & .368 & 1.604 & .547 & .717 & 2.023 & .356 \\
\addlinespace[.45em]
80 & 30 & .170 & .050 & .112 & .246 & .232 & .095 & .099 & .195 & .239 & .071 & .463 & .170 \\
 & 60 & .511 & .649 & .567 & .671 & .505 & .513 & .443 & .644 & .695 & .909 & .310 & .331 \\
 & 90 & .473 & .524 & .448 & .359 & .516 & .447 & .518 & .455 & .350 & .417 & .018 & .111 \\
\addlinespace[.45em]
120 & 30 & .293 & .205 & .227 & .255 & .041 & .092 & .040 & .937 & .309 & .217 & .587 & .123 \\
 & 60 & .140 & .045 & .149 & .182 & .207 & .122 & .170 & .739 & .187 & .489 & .003 & .114 \\
 & 90 & .205 & .219 & .228 & .162 & .097 & .106 & .151 & 1.189 & .178 & .143 & 1.613 & .319 \\
\midrule
\multicolumn{14}{c}{\footnotesize\textbf{Panel C. $\bm{K>r}$}} \\
40 & 30 & .331 & .411 & .343 & .365 & .043 & .043 & .067 & .096 & .267 & .330 & 1.437 & .414 \\
 & 60 & .525 & .473 & .536 & .691 & .764 & .540 & .633 & 1.224 & .673 & .222 & .281 & .491 \\
 & 90 & .193 & .248 & .181 & .081 & .135 & .182 & .181 & .506 & .078 & .086 & .311 & .268 \\
\addlinespace[.45em]
80 & 30 & .610 & .893 & .796 & 1.004 & .837 & .792 & .929 & .667 & .956 & 1.806 & .494 & .752 \\
 & 60 & .123 & .340 & .186 & .112 & .260 & .212 & .214 & 1.011 & .125 & .117 & .253 & .060 \\
 & 90 & .011 & .062 & .037 & .059 & .033 & .034 & .011 & .036 & .050 & .143 & .253 & .049 \\
\addlinespace[.45em]
120 & 30 & .142 & .098 & .144 & .260 & .195 & .149 & .139 & 3.192 & .281 & .289 & .597 & .535 \\
 & 60 & .052 & .097 & .092 & .142 & .087 & .037 & .028 & .258 & .115 & .451 & .185 & .079 \\
 & 90 & .287 & .303 & .246 & .329 & .275 & .192 & .291 & .199 & .331 & .102 & .323 & .387 \\
\bottomrule
\end{tabular}
}}\usebox{\tabbox}
\par\vspace{.35\baselineskip}
\begin{minipage}{\wd\tabbox}\footnotesize
\textit{Notes}: The true ATT is zero by construction. The absolute bias is the absolute value of the average of $\hat{\tau}$ over the 1000 replications. All entries are multiplied by $100$.
\end{minipage}
\end{table}

\begin{table}[!htbp]
\centering
\small
\caption{Ratio of the average length of the CIs to that of SCM (exact group)}
\label{tab:ci_ratio}
\smallskip
\sbox{\tabbox}{\resizebox{\textwidth}{!}{\hdreset\hdbump{$L_2$}\hdbump{EL}\hdbump{Entropy}\hdbump{Lasso}\hdbump{Ridge}\hdbump{GLasso}\hdbump{fsPDA}\hdbump{AugSC}\hdbump{PenSC}\hdbump{SDID}\hdbump{GSC}\hdbump{\textbf{8.888}}\setlength{\tabcolsep}{4pt}%
\begin{tabular}{rrS[table-format=1.3]S[table-format=1.3]S[table-format=1.3]S[table-format=1.3]S[table-format=1.3]S[table-format=1.3]S[table-format=1.3]S[table-format=1.3]S[table-format=1.3]S[table-format=1.3]S[table-format=1.3]}
\toprule
\multicolumn{2}{c}{} & \multicolumn{3}{c}{Relaxation} & \multicolumn{8}{c}{} \\
\cmidrule(lr){3-5}
$J$ & $T_0$ & \hd{$L_2$} & \hd{EL} & \hd{Entropy} & \hd{Lasso} & \hd{Ridge} & \hd{GLasso} & \hd{fsPDA} & \hd{AugSC} & \hd{PenSC} & \hd{SDID} & \hd{GSC} \\
\midrule
\multicolumn{13}{c}{\footnotesize\textbf{Panel A. $\bm{K<r}$}} \\
40 & 30 & .954 & .992 & .962 & 1.000 & .983 & .983 & 1.356 & 1.015 & 1.172 & 1.270 & 1.020 \\
 & 60 & .961 & 1.000 & .969 & .998 & .978 & .982 & 1.174 & 1.000 & 1.110 & 1.331 & .995 \\
 & 90 & .975 & 1.001 & .975 & .998 & .986 & .988 & 1.129 & 1.000 & 1.096 & 1.377 & .989 \\
\addlinespace[.45em]
80 & 30 & .953 & .993 & .961 & 1.001 & .973 & .974 & 1.470 & 1.001 & 1.181 & 1.099 & .994 \\
 & 60 & .963 & 1.005 & .968 & .992 & .980 & .980 & 1.198 & 1.003 & 1.132 & 1.154 & .999 \\
 & 90 & .959 & .991 & .964 & .988 & .973 & .970 & 1.151 & 1.003 & 1.087 & 1.151 & .978 \\
\addlinespace[.45em]
120 & 30 & .949 & .993 & .956 & 1.003 & .973 & .975 & 1.496 & 1.000 & 1.112 & 1.052 & .990 \\
 & 60 & .961 & .993 & .962 & .993 & .978 & .975 & 1.208 & 1.000 & 1.128 & 1.085 & .987 \\
 & 90 & .963 & .992 & .967 & .998 & .979 & .979 & 1.161 & 1.000 & 1.091 & 1.085 & .981 \\
\midrule
\multicolumn{13}{c}{\footnotesize\textbf{Panel B. $\bm{K=r}$}} \\
40 & 30 & .977 & .987 & .973 & .995 & .984 & .982 & 1.321 & 1.013 & 1.192 & 1.205 & 1.036 \\
 & 60 & .966 & .984 & .967 & .992 & .983 & .983 & 1.154 & 1.005 & 1.138 & 1.266 & 1.002 \\
 & 90 & .970 & .990 & .972 & .992 & .978 & .981 & 1.116 & 1.000 & 1.118 & 1.318 & .986 \\
\addlinespace[.45em]
80 & 30 & .961 & .980 & .961 & .993 & .970 & .969 & 1.444 & 1.003 & 1.199 & 1.078 & 1.026 \\
 & 60 & .959 & .978 & .961 & .991 & .977 & .975 & 1.191 & 1.004 & 1.128 & 1.122 & .988 \\
 & 90 & .961 & .984 & .966 & .996 & .981 & .979 & 1.143 & 1.002 & 1.110 & 1.148 & .977 \\
\addlinespace[.45em]
120 & 30 & .969 & .994 & .978 & 1.000 & .976 & .973 & 1.499 & 1.001 & 1.168 & 1.023 & 1.017 \\
 & 60 & .969 & .992 & .971 & .994 & .980 & .983 & 1.220 & 1.000 & 1.132 & 1.077 & .994 \\
 & 90 & .977 & .989 & .976 & 1.003 & .988 & .988 & 1.169 & 1.000 & 1.105 & 1.090 & .997 \\
\midrule
\multicolumn{13}{c}{\footnotesize\textbf{Panel C. $\bm{K>r}$}} \\
40 & 30 & .966 & .971 & .966 & .989 & .973 & .970 & 1.310 & 1.014 & 1.209 & 1.105 & 1.003 \\
 & 60 & .975 & .981 & .976 & .987 & .979 & .980 & 1.113 & 1.002 & 1.146 & 1.231 & .978 \\
 & 90 & .978 & .989 & .980 & .998 & .983 & .987 & 1.111 & 1.000 & 1.123 & 1.237 & .984 \\
\addlinespace[.45em]
80 & 30 & .964 & .968 & .964 & .990 & .975 & .973 & 1.436 & 1.003 & 1.192 & 1.037 & 1.009 \\
 & 60 & .974 & .981 & .976 & .991 & .977 & .978 & 1.175 & 1.004 & 1.151 & 1.083 & .986 \\
 & 90 & .967 & .979 & .970 & .987 & .976 & .972 & 1.138 & 1.002 & 1.107 & 1.098 & .978 \\
\addlinespace[.45em]
120 & 30 & .964 & .978 & .971 & .996 & .974 & .975 & 1.485 & 1.002 & 1.200 & 1.001 & 1.014 \\
 & 60 & .967 & .981 & .967 & .990 & .976 & .978 & 1.204 & 1.001 & 1.114 & 1.031 & .988 \\
 & 90 & .961 & .969 & .963 & .986 & .972 & .973 & 1.135 & 1.000 & 1.106 & 1.048 & .968 \\
\bottomrule
\end{tabular}
}}\usebox{\tabbox}
\par\vspace{.35\baselineskip}
\begin{minipage}{\wd\tabbox}\footnotesize
\textit{Notes}: 1.~The entry is the average length of the nominal 90\% CI divided by that of SCM. Values below one indicate CIs shorter than those of SCM. 2.~See the notes in \cref{tab:cp}.
\end{minipage}
\end{table}

\begin{table}[!htbp]
\centering
\caption{RMSE\textsuperscript{$\star$} (approximate group, misspecified)}
\label{tab:rmse_truth_approx_misspec}
\smallskip
\sbox{\tabbox}{\resizebox{\textwidth}{!}{\hdreset\hdbump{$L_2$}\hdbump{EL}\hdbump{Entropy}\hdbump{SCM}\hdbump{Lasso}\hdbump{Ridge}\hdbump{GLasso}\hdbump{fsPDA}\hdbump{AugSC}\hdbump{PenSC}\hdbump{SDID}\hdbump{GSC}\hdbump{\textbf{8.888}}\setlength{\tabcolsep}{4pt}%
\begin{tabular}{rrS[table-format=1.3]S[table-format=1.3]S[table-format=1.3]S[table-format=1.3]S[table-format=1.3]S[table-format=1.3]S[table-format=1.3]S[table-format=1.3]S[table-format=1.3]S[table-format=1.3]S[table-format=1.3]S[table-format=1.3]}
\toprule
\multicolumn{2}{c}{} & \multicolumn{3}{c}{Relaxation} & \multicolumn{9}{c}{} \\
\cmidrule(lr){3-5}
$J$ & $T_0$ & \hd{$L_2$} & \hd{EL} & \hd{Entropy} & \hd{SCM} & \hd{Lasso} & \hd{Ridge} & \hd{GLasso} & \hd{fsPDA} & \hd{AugSC} & \hd{PenSC} & \hd{SDID} & \hd{GSC} \\
\midrule
\multicolumn{14}{c}{\footnotesize\textbf{Panel A. $\bm{K<r}$}} \\
40 & 30 & \textbf{1.171} & 1.195 & 1.175 & 1.225 & 1.206 & 1.189 & 1.190 & 1.424 & 1.245 & 1.344 & 1.341 & 1.199 \\
 & 60 & \textbf{1.144} & 1.163 & 1.147 & 1.179 & 1.166 & 1.155 & 1.155 & 1.262 & 1.180 & 1.245 & 1.297 & 1.153 \\
 & 90 & \textbf{1.141} & 1.155 & 1.143 & 1.167 & 1.156 & 1.146 & 1.146 & 1.222 & 1.167 & 1.223 & 1.290 & 1.147 \\
\addlinespace[.45em]
80 & 30 & \textbf{1.163} & 1.184 & 1.167 & 1.208 & 1.191 & 1.172 & 1.173 & 1.496 & 1.212 & 1.318 & 1.232 & 1.188 \\
 & 60 & \textbf{1.136} & 1.150 & 1.137 & 1.169 & 1.154 & 1.142 & 1.141 & 1.268 & 1.170 & 1.239 & 1.198 & 1.146 \\
 & 90 & \textbf{1.127} & 1.142 & 1.129 & 1.156 & 1.143 & 1.132 & 1.132 & 1.221 & 1.156 & 1.209 & 1.186 & 1.135 \\
\addlinespace[.45em]
120 & 30 & \textbf{1.158} & 1.180 & 1.162 & 1.208 & 1.191 & 1.171 & 1.171 & 1.569 & 1.209 & 1.314 & 1.202 & 1.181 \\
 & 60 & \textbf{1.136} & 1.154 & 1.137 & 1.170 & 1.157 & 1.143 & 1.143 & 1.280 & 1.171 & 1.237 & 1.167 & 1.145 \\
 & 90 & \textbf{1.124} & 1.137 & 1.125 & 1.149 & 1.139 & 1.126 & 1.126 & 1.218 & 1.149 & 1.197 & 1.148 & 1.130 \\
\midrule
\multicolumn{14}{c}{\footnotesize\textbf{Panel B. $\bm{K=r}$}} \\
40 & 30 & \textbf{1.187} & 1.201 & 1.190 & 1.231 & 1.211 & 1.197 & 1.198 & 1.423 & 1.250 & 1.400 & 1.327 & 1.215 \\
 & 60 & \textbf{1.157} & 1.168 & 1.159 & 1.187 & 1.172 & 1.163 & 1.163 & 1.261 & 1.188 & 1.268 & 1.291 & 1.161 \\
 & 90 & \textbf{1.144} & 1.153 & 1.145 & 1.171 & 1.158 & 1.149 & 1.151 & 1.215 & 1.172 & 1.222 & 1.270 & 1.147 \\
\addlinespace[.45em]
80 & 30 & \textbf{1.176} & 1.189 & 1.178 & 1.222 & 1.201 & 1.187 & 1.185 & 1.497 & 1.225 & 1.334 & 1.233 & 1.201 \\
 & 60 & \textbf{1.141} & 1.151 & 1.142 & 1.175 & 1.159 & 1.149 & 1.149 & 1.261 & 1.178 & 1.252 & 1.191 & 1.149 \\
 & 90 & \textbf{1.137} & 1.148 & 1.138 & 1.161 & 1.148 & 1.139 & 1.139 & 1.219 & 1.161 & 1.217 & 1.189 & 1.143 \\
\addlinespace[.45em]
120 & 30 & \textbf{1.167} & 1.180 & 1.168 & 1.210 & 1.194 & 1.177 & 1.177 & 1.533 & 1.212 & 1.325 & 1.198 & 1.190 \\
 & 60 & \textbf{1.136} & 1.148 & 1.138 & 1.169 & 1.155 & 1.143 & 1.144 & 1.273 & 1.170 & 1.242 & 1.158 & 1.144 \\
 & 90 & \textbf{1.129} & 1.137 & 1.130 & 1.152 & 1.139 & 1.131 & 1.131 & 1.216 & 1.152 & 1.195 & 1.147 & 1.133 \\
\midrule
\multicolumn{14}{c}{\footnotesize\textbf{Panel C. $\bm{K>r}$}} \\
40 & 30 & \textbf{1.201} & 1.210 & 1.202 & 1.245 & 1.222 & 1.209 & 1.210 & 1.433 & 1.257 & 1.375 & 1.309 & 1.223 \\
 & 60 & \textbf{1.156} & 1.164 & 1.158 & 1.187 & 1.171 & 1.162 & 1.164 & 1.250 & 1.188 & 1.264 & 1.260 & 1.164 \\
 & 90 & \textbf{1.146} & 1.153 & 1.148 & 1.168 & 1.157 & 1.150 & 1.150 & 1.209 & 1.169 & 1.223 & 1.255 & 1.152 \\
\addlinespace[.45em]
80 & 30 & 1.192 & 1.198 & \textbf{1.191} & 1.229 & 1.208 & 1.197 & 1.198 & 1.497 & 1.233 & 1.360 & 1.227 & 1.207 \\
 & 60 & \textbf{1.150} & 1.156 & 1.151 & 1.177 & 1.163 & 1.154 & 1.155 & 1.261 & 1.179 & 1.256 & 1.182 & 1.159 \\
 & 90 & \textbf{1.133} & 1.137 & 1.133 & 1.156 & 1.143 & 1.136 & 1.136 & 1.208 & 1.156 & 1.210 & 1.172 & 1.135 \\
\addlinespace[.45em]
120 & 30 & \textbf{1.185} & 1.193 & 1.186 & 1.228 & 1.206 & 1.190 & 1.192 & 1.535 & 1.230 & 1.340 & 1.197 & 1.207 \\
 & 60 & \textbf{1.142} & 1.148 & 1.142 & 1.169 & 1.157 & 1.145 & 1.146 & 1.265 & 1.171 & 1.236 & 1.152 & 1.147 \\
 & 90 & \textbf{1.137} & 1.142 & 1.138 & 1.158 & 1.146 & 1.138 & 1.140 & 1.212 & 1.160 & 1.206 & 1.150 & 1.140 \\
\bottomrule
\end{tabular}
}}\usebox{\tabbox}
\par\vspace{.35\baselineskip}
\begin{minipage}{\wd\tabbox}\footnotesize
\textit{Notes}: See the notes in \cref{tab:rmse_truth_misspec}.
\end{minipage}
\end{table}

\begin{table}[!htbp]
\centering
\small
\caption{RMSE\textsuperscript{$\star$} (general factor, misspecified)}
\label{tab:rmse_truth_general_misspec}
\smallskip
\sbox{\tabbox}{\resizebox{\textwidth}{!}{\hdreset\hdbump{$L_2$}\hdbump{EL}\hdbump{Entropy}\hdbump{SCM}\hdbump{Lasso}\hdbump{Ridge}\hdbump{fsPDA}\hdbump{AugSC}\hdbump{PenSC}\hdbump{SDID}\hdbump{GSC}\hdbump{\textbf{8.888}}\setlength{\tabcolsep}{4pt}%
\begin{tabular}{rrS[table-format=1.3]S[table-format=1.3]S[table-format=1.3]S[table-format=1.3]S[table-format=1.3]S[table-format=1.3]S[table-format=1.3]S[table-format=1.3]S[table-format=1.3]S[table-format=1.3]S[table-format=1.3]}
\toprule
\multicolumn{2}{c}{} & \multicolumn{3}{c}{Relaxation} & \multicolumn{8}{c}{} \\
\cmidrule(lr){3-5}
$J$ & $T_0$ & \hd{$L_2$} & \hd{EL} & \hd{Entropy} & \hd{SCM} & \hd{Lasso} & \hd{Ridge} & \hd{fsPDA} & \hd{AugSC} & \hd{PenSC} & \hd{SDID} & \hd{GSC} \\
\midrule
40 & 30 & \textbf{1.221} & 1.232 & 1.224 & 1.257 & 1.239 & 1.228 & 1.336 & 1.278 & 1.393 & 1.337 & 1.235 \\
 & 60 & \textbf{1.160} & 1.169 & 1.162 & 1.191 & 1.177 & 1.170 & 1.198 & 1.192 & 1.269 & 1.295 & 1.163 \\
 & 90 & \textbf{1.150} & 1.159 & 1.152 & 1.173 & 1.163 & 1.156 & 1.179 & 1.174 & 1.225 & 1.289 & 1.151 \\
\addlinespace[.45em]
80 & 30 & \textbf{1.214} & 1.222 & 1.216 & 1.249 & 1.230 & 1.219 & 1.369 & 1.253 & 1.369 & 1.253 & 1.235 \\
 & 60 & \textbf{1.152} & 1.159 & 1.154 & 1.181 & 1.167 & 1.159 & 1.207 & 1.184 & 1.255 & 1.195 & 1.160 \\
 & 90 & \textbf{1.131} & 1.136 & 1.132 & 1.157 & 1.142 & 1.137 & 1.169 & 1.158 & 1.204 & 1.170 & 1.139 \\
\addlinespace[.45em]
120 & 30 & 1.218 & 1.224 & 1.218 & 1.254 & 1.233 & 1.222 & 1.386 & 1.257 & 1.366 & \textbf{1.212} & 1.236 \\
 & 60 & \textbf{1.148} & 1.155 & 1.150 & 1.178 & 1.163 & 1.155 & 1.213 & 1.180 & 1.247 & 1.156 & 1.160 \\
 & 90 & \textbf{1.137} & 1.143 & 1.139 & 1.159 & 1.147 & 1.141 & 1.175 & 1.160 & 1.209 & 1.142 & 1.145 \\
\bottomrule
\end{tabular}
}}\usebox{\tabbox}
\par\vspace{.35\baselineskip}
\begin{minipage}{\wd\tabbox}\footnotesize
\textit{Notes}: See the notes in \cref{tab:rmse_truth_misspec}.
\end{minipage}
\end{table}

\FloatBarrier

\section{ATT Consistency without Group Structure}\label{sec:without-group}

This section provides a theoretical framework for establishing ATT consistency without imposing group structure. The proofs of the theorems are collected in \cref{subsec:without-group-proofs}.

\subsection{Theoretical Framework}\label{subsec:without-group-framework}

The group structure is a key assumption in our theory.
While both the group structure and the factor structure are dimension reduction devices, they play different roles. The common factors mean that the potential outcome of the treated unit shares the same signal as the control units, enabling synthetic control to approximate the treated unit's potential outcome. They do not by themselves establish convergence of the individual weights to the oracle at a desirable rate.
Although the oracle is uniquely defined whenever the feasible set is nonempty, estimating this high-dimensional vector requires further analysis. The group structure makes the oracle weights equal within groups, thus reducing the complexity of the donor pool and enabling estimation of a dense oracle.

The critical intermediate result is the compatibility condition:
\[
\|\hat{\bm{w}}-\bm{w}^{*}\|_{2}^{2} \lesssim_p \frac{K}{J}\|\hat{\bm{w}}_{\mathcal{G}}-\bm{w}_{\mathcal{G}}^{*}\|_{2}.
\]
The group structure assumption is a sufficient condition for the compatibility condition to hold, as shown in \cref{lem:link_two_norms}. If we abandon the group structure assumption, the compatibility condition breaks down in general and it is very difficult, if not impossible, to derive meaningful convergence rates for the estimated weight.

In this section, without the group structure, we seek a weaker result for a weight vector that balances the factor loadings:
\[
\|\bm{\Lambda}'\hat{\bm{w}}-\bm{\lambda}_0\|_2 \to_p 0.
\]
This implies that the treated unit's factor loading can be asymptotically recovered by the control units' factor loadings, which further leads to ATT consistency.

In this section, we treat each control unit as its own group, i.e., $\bm{Z}=\bm{I}_{J}$, $K=J$, and $\bm{\Lambda}^{\mathrm{co}}=\bm{\Lambda}$, so that \cref{assu:DGP} and the notation of the main text apply. To facilitate formal discussion, we make the following assumption.

\begin{assumption}\label{assu:diffuse_weight}
There exists $\bm{w}^o\in\Delta_{J}$ such that $\bm{\Lambda}'\bm{w}^o = \bm{\lambda}_0$ and $\|\bm{w}^o\|_\infty \leq C/J$ for some absolute constant $C$.
\end{assumption}

This assumption requires the existence of weights that exactly reproduce the treated unit's factor loadings and are sufficiently spread across control units. The bound on individual weights ensures that no single control unit receives too much weight. The weight vector $\bm{w}^o$ need not be unique, and its existence does not require group structure.
See \Cref{subsec:without-group-proofs} for primitive conditions for the dispersion of the weights.

\begin{rem}
This assumption can be relaxed to require only approximate recovery of the treated unit's factor loading, i.e., $\|\bm{\Lambda}'\bm{w}^o-\bm{\lambda}_0\|_2 = O(\delta_J)$ where $\delta_J\to0$ as $J\to\infty$ is a sequence controlling the approximation error. The consequence is that the convergence rate of the ATT estimator would involve $\delta_J$. This type of condition is also imposed in, for example, \citet[Assumption~3.2]{fermanPropertiesSyntheticControl2021}.
\end{rem}

Under \cref{assu:diffuse_weight}, the following theorem bounds the discrepancy between the treated unit's factor loadings and those recovered by the estimated weights.

\begin{thm}\label{thm:diffuse_weight_convergence}
If \cref{assu:DGP,assu:errors,assu:diffuse_weight} hold and $\sqrt{r}\phi_{J,T_{0}} = o(\eta)$, then
\[
\|\bm{\Lambda}'\hat{\bm{w}}-\bm{\lambda}_0\|_2 = O_p\Bigl(\sqrt{\eta+\sqrt{r}\phi_{J,T_{0}}}\Bigr).
\]
\end{thm}

This result suggests that the $L_{2}$-SCM-relaxation estimator achieves asymptotic factor-loading balance whenever $\eta + \sqrt{r}\phi_{J,T_{0}} = o(1)$. Without additional structure such as latent groups, the estimand that can be stably recovered is the factor-loading combination $\bm{\Lambda}'\bm{w}$ rather than the high-dimensional weight vector itself. This also echoes the similar result for the SCM estimator in \citet*{fermanPropertiesSyntheticControl2021}, which shows $\|\bm{\Lambda}'\hat{\bm{w}}^{\mathrm{SC}}-\bm{\lambda}_0\|_2 \to_p 0$.

To translate this result into ATT consistency, we impose additional conditions on the average factors during the post-treatment period and the temporal dependence of the idiosyncratic errors.

\begin{assumption}\label{assu:factors_long_run_variance}\mbox{}
\begin{enumerate}[label=(\alph*),parsep=0pt,topsep=0pt,font=\upshape]
\item \label{enu:factors_mixing}$\|T_{1}^{-1}\sum_{t\in\mathcal{T}_{1}}\bm{f}_{t}\|_2 = O_p(\sqrt{r})$.
\item \label{enu:long_run_variance}$\sum_{\ell=-\infty}^{\infty}|\mathbb{E}(u_{0t}u_{0,t-\ell})|+J^{-1}\sum_{j\in[J]}\sum_{\ell=-\infty}^{\infty}|\mathbb{E}(u_{jt}u_{j,t-\ell})| = O(1)$.
\end{enumerate}
\end{assumption}

Combining \cref{thm:diffuse_weight_convergence} with these conditions yields the following bound on the ATT estimation error.

\begin{thm}\label{thm:att_convergence}
If \cref{assu:DGP,assu:errors,assu:diffuse_weight,assu:factors_long_run_variance} hold, and $\sqrt{r}\phi_{J,T_{0}} = o(\eta)$, then
\[
\hat{\tau} - \tau_{0} = O_p\Bigl(\sqrt{r\eta+r^{3/2}\phi_{J,T_{0}}}\Bigr) + O_p(T_{1}^{-1/2}).
\]
\end{thm}

This theorem shows that the relaxation-based ATT estimator is consistent for the ATT even without group structure. Specifically, under the theorem's conditions, consistency follows when $r\eta+r^{3/2}\phi_{J,T_{0}}\to0$ and $T_1\to\infty$.

Now that ATT consistency can be obtained without group structure, why do we still impose the group structure in the main text? There are three reasons. First, the group structure remains critical for the stronger goal pursued---to recover a high-dimensional dense weight vector with an interpretable oracle. This vector describes how the treated unit is synthesized from a broad donor pool. Without additional structure, many weight vectors can reproduce the treated unit's factor loadings. The group structure reduces the dimension of the problem while preserving a dense oracle weight vector through equal weighting within groups, making that vector interpretable and estimable.

Second, the group structure also provides a primitive bound on individual weights: $\|\bm{w}^*\|_\infty=O_p(K/J)$. When $K$ is bounded, this gives the $O_p(J^{-1})$ bound used in this appendix. Without group structure, a bound on individual weights must be imposed directly.

Finally, the group structure sharpens the convergence rate for the ATT estimator. Under the group structure, \cref{thm:att} implies
\[
\hat{\tau}-\tau_{0}=O_{p}\bigl(\sqrt{r(K\wedge r)}\,\eta\bigr)+O_{p}\bigl(\sqrt{K/T_{1}}\bigr).
\]
By contrast, \cref{thm:att_convergence} obtains the rate
\[
\hat{\tau} - \tau_{0} = O_p\Bigl(\sqrt{r\eta+r^{3/2}\phi_{J,T_{0}}}\Bigr) + O_p(T_{1}^{-1/2})
\]
without group structure, where $\phi_{J,T_{0}}=\sqrt{\log(J)/T_{0}}$ and $\sqrt r\phi_{J,T_{0}}=o(\eta)$ is required. When $r$ and $K$ are fixed, the group-structure result is of order $O_p(\eta)+O_p(T_1^{-1/2})$, whereas the no-group result is of order $O_p(\sqrt{\eta})+O_p(T_1^{-1/2})$. Thus, the group structure yields a faster leading rate for the ATT estimator, in addition to enabling convergence to a meaningful dense oracle weight.

\subsection{Proofs}\label{subsec:without-group-proofs}

\begin{proof}[Proof of \cref{thm:att_convergence}]
We can write
\begin{align*}
\hat{\tau} - \tau_{0} &= \frac{1}{T_{1}}\sum_{t\in\mathcal{T}_{1}}\Biggl(y_{0t}^{N}-\sum_{j\in[J]}\hat{w}_{j}y_{jt}^{N}\Biggr) \\
&= \frac{1}{T_{1}}\sum_{t\in\mathcal{T}_{1}}\bm{f}_{t}'(\bm{\lambda}_{0}-\bm{\Lambda}'\hat{\bm{w}}) + \frac{1}{T_{1}}\sum_{t\in\mathcal{T}_{1}}u_{0t} - \frac{1}{T_{1}}\sum_{t\in\mathcal{T}_{1}}\sum_{j\in[J]}\hat{w}_{j}u_{jt}.
\end{align*}
By \cref{thm:diffuse_weight_convergence}, the first term can be bounded as
\[
\Biggl|\frac{1}{T_{1}}\sum_{t\in\mathcal{T}_{1}}\bm{f}_{t}'(\bm{\Lambda}'\hat{\bm{w}}-\bm{\lambda}_{0})\Biggr| \leq \Biggl\|\frac{1}{T_{1}}\sum_{t\in\mathcal{T}_{1}}\bm{f}_{t}\Biggr\|_2 \|\bm{\Lambda}'\hat{\bm{w}}-\bm{\lambda}_{0}\|_2 = O_p\Bigl(\sqrt{r\eta+r^{3/2}\phi_{J,T_{0}}}\Bigr),
\]
where $\|T_{1}^{-1}\sum_{t\in\mathcal{T}_{1}}\bm{f}_{t}\|_2 = O_p(\sqrt{r})$ by \cref{assu:factors_long_run_variance}. By Chebyshev's inequality and \cref{assu:factors_long_run_variance}\ref{enu:long_run_variance}, the second term is
\[
\Biggl|\frac{1}{T_{1}}\sum_{t\in\mathcal{T}_{1}}u_{0t}\Biggr| = O_p(T_{1}^{-1/2}).
\]
The third term is bounded by
\begin{align*}
\Biggl|\frac{1}{T_{1}}\sum_{t\in\mathcal{T}_{1}}\sum_{j\in[J]}\hat{w}_{j}u_{jt}\Biggr| &\leq \|\hat{\bm{w}}\|_2\cdot\left[\sum_{j\in[J]}\Biggl(\frac{1}{T_{1}}\sum_{t\in\mathcal{T}_{1}}u_{jt}\Biggr)^{2}\right]^{1/2} \\
&\leq O_p(J^{-1/2})\cdot O_p(\sqrt{J/T_1}) = O_p(T_1^{-1/2}),
\end{align*}
where we use the fact that
\[
\frac1J\sum_{j\in[J]}\Biggl(\frac{1}{T_{1}}\sum_{t\in\mathcal{T}_{1}}u_{jt}\Biggr)^{2} = O_p\biggl(\frac{1}{T_1}\biggr)
\]
since by \cref{assu:factors_long_run_variance}\ref{enu:long_run_variance}
\[
\frac1J\sum_{j\in[J]}\mathbb{E}\left[\Biggl(\frac{1}{T_{1}}\sum_{t\in\mathcal{T}_{1}}u_{jt}\Biggr)^{2}\right] \leq \frac{1}{T_1}\Biggl[\frac1J\sum_{j\in[J]}\sum_{\ell=-\infty}^{\infty}|\mathbb{E}(u_{jt}u_{j,t-\ell})|\Biggr] = O\biggl(\frac{1}{T_1}\biggr).
\]
Hence, we have
\[
\hat{\tau} - \tau_{0} = O_p\Bigl(\sqrt{r\eta+r^{3/2}\phi_{J,T_{0}}}\Bigr) + O_p(T_{1}^{-1/2}).
\]
This completes the proof.
\end{proof}

\begin{proof}[Proof of \cref{thm:diffuse_weight_convergence}]
Since $\bm{\Lambda}'\bm{w}^o = \bm{\lambda}_0$, we have
\[
\hat{\bm{\Sigma}}^{*}\bm{w}^o - \hat{\bm{\Upsilon}}^{*} = \bm{\Lambda}\hat{\bm{\Omega}}_{\bm{F}}(\bm{\Lambda}'\bm{w}^o - \bm{\lambda}_0) = \bm{0}.
\]
Let $\hat{\bm{h}} := \hat{\bm{w}}-\bm{w}^o$. Because $\hat{\bm{w}},\bm{w}^o\in\Delta_{J}$, $\|\hat{\bm{h}}\|_1\leq2$. By the relaxation constraint,
\begin{align*}
\eta &\geq \|\hat{\bm{\Sigma}}\hat{\bm{w}}-\hat{\bm{\Upsilon}}\|_\infty \stackrel{\text{(i)}}{=}\|\hat{\bm{\Sigma}}\hat{\bm{w}}-\hat{\bm{\Upsilon}}-(\hat{\bm{\Sigma}}^{*}\bm{w}^o-\hat{\bm{\Upsilon}}^{*})\|_{\infty}\nonumber \\
& =\|\hat{\bm{\Sigma}}^{*}\hat{\bm{h}}+\hat{\bm{\Sigma}}^{\mathrm{e}}\hat{\bm{w}}-\hat{\bm{\Upsilon}}^{\mathrm{e}}\|_{\infty},
\end{align*}
where equality (i) is implied by the fact that $\hat{\bm{\Sigma}}^{*}\bm{w}^o-\hat{\bm{\Upsilon}}^{*}=\bm{0}$. See also \eqref{eq:upper_bound1}. Multiplying both sides by $\hat{\bm{h}}'$, we get
\[
\bigl|\hat{\bm{h}}'\hat{\bm{\Sigma}}^{*}\hat{\bm{h}} + \hat{\bm{h}}'(\hat{\bm{\Sigma}}^{\mathrm{e}}\hat{\bm{w}}-\hat{\bm{\Upsilon}}^{\mathrm{e}})\bigr| \leq \eta\|\hat{\bm{h}}\|_1 \leq 2\eta,
\]
which further implies
\[
\hat{\bm{h}}'\hat{\bm{\Sigma}}^{*}\hat{\bm{h}} \leq 2\eta + \bigl|\hat{\bm{h}}'(\hat{\bm{\Sigma}}^{\mathrm{e}}\hat{\bm{w}}-\hat{\bm{\Upsilon}}^{\mathrm{e}})\bigr|.
\]
Using the argument in the proof of \cref{lem:w_star_feasibility} with $\|\bm{w}^o\|_\infty\leq C/J$ and $\sqrt{r}\phi_{J,T_{0}}=o(\eta)$, we can show that $\bm{w}^o\in S_\eta$ w.p.a.1. Thus, by the optimality of $\hat{\bm{w}}$, it holds that w.p.a.1,
\[
\|\hat{\bm{w}}\|_2\leq\|\bm{w}^o\|_2\leq\sqrt{J}\|\bm{w}^o\|_\infty=O_p(1/\sqrt{J}).
\]
It also follows that $\|\hat{\bm{h}}\|_2 = O_p(1/\sqrt{J})$. These two results together with \cref{lem:sampling-error} yield
\begin{align*}
\bigl|\hat{\bm{h}}'(\hat{\bm{\Sigma}}^{\mathrm{e}}\hat{\bm{w}}-\hat{\bm{\Upsilon}}^{\mathrm{e}})\bigr| &\leq \|\hat{\bm{h}}\|_2\|\hat{\bm{\Sigma}}^{\mathrm{e}}\|_2\|\hat{\bm{w}}\|_2 + \|\hat{\bm{h}}\|_1\|\hat{\bm{\Upsilon}}^{\mathrm{e}}\|_{\infty} \\
&\leq \|\hat{\bm{h}}\|_2\|\hat{\bm{\Sigma}}^{\mathrm{e}}\|_1\|\hat{\bm{w}}\|_2 + \|\hat{\bm{h}}\|_1\|\hat{\bm{\Upsilon}}^{\mathrm{e}}\|_{\infty} \\
&\lesssim_p \frac{1}{\sqrt{J}}\cdot J\sqrt{\frac{r\log J}{T_{0}}}\cdot\frac{1}{\sqrt{J}}+\sqrt{\frac{r\log J}{T_{0}}} = \sqrt{r}\phi_{J,T_{0}}.
\end{align*}
Hence, we have
\[
\hat{\bm{h}}'\hat{\bm{\Sigma}}^{*}\hat{\bm{h}} \lesssim_p 2\eta + \sqrt{r}\phi_{J,T_{0}}.
\]
Since $\phi_{\min}(\hat{\bm{\Omega}}_{\bm{F}}) \geq \underline{c}/2>0$ w.p.a.1,
\[
\hat{\bm{h}}'\hat{\bm{\Sigma}}^{*}\hat{\bm{h}} = (\bm{\Lambda}'\hat{\bm{h}})'\hat{\bm{\Omega}}_{\bm{F}}(\bm{\Lambda}'\hat{\bm{h}}) \geq (\underline{c}/2)\|\bm{\Lambda}'\hat{\bm{h}}\|_{2}^{2}.
\]
It follows that $\|\bm{\Lambda}'\hat{\bm{w}}-\bm{\lambda}_0\|_2 = O_p\bigl(\sqrt{\eta+\sqrt{r}\phi_{J,T_{0}}}\bigr)$.
\end{proof}

\noindent
\textbf{Primitive conditions}
for weights sufficiently spread across control units.
Assume that the loadings $\bm{\lambda}_1,\bm{\lambda}_2,\dots,\bm{\lambda}_J$ are i.i.d.\ and bounded. Suppose that there is a bounded reweighting function $q(\bm{\lambda})$ such that $0\leq q(\bm{\lambda})\leq C<\infty$, $\mathbb{E}[q(\bm{\lambda})] = 1$ and $\mathbb{E}[q(\bm{\lambda})\bm{\lambda}] = \bm{\lambda}_0$, and the weighted covariance matrix $\mathbb{E}[q(\bm{\lambda})(\bm{\lambda}-\bm{\lambda}_0)(\bm{\lambda}-\bm{\lambda}_0)']$ has minimum eigenvalue bounded away from zero.

Intuitively, the reweighting function $q(\bm{\lambda})$ can be understood as a density ratio or tilting that creates a target distribution whose expectation is $\bm{\lambda}_0$. Under these conditions, as $J\to\infty$, with high probability there exists $\bm{w}^o$ such that $\bm{\Lambda}'\bm{w}^o = \bm{\lambda}_0$ and $\|\bm{w}^o\|_\infty = O_p(J^{-1})$. On the other hand, the group structure assumption imposed in the main text provides an easily interpretable condition for the existence of a diluted weight vector because it implies $\|\bm{w}^*\|_\infty = O_p(K/J)$.

\begin{proof}
Write $q_j=q(\bm{\lambda}_j)$ and $\bm{x}_j=\bm{\lambda}_j-\bm{\lambda}_0$. The moment restrictions imply $\mathbb{E}(q_j)=1$ and $\mathbb{E}(q_j\bm{x}_j)=\bm{0}$. Let
$s_J=\frac{1}{J}\sum_{j=1}^{J}q_j,$
$\bm{b}_J=\frac{1}{J}\sum_{j=1}^{J}q_j\bm{x}_j,$ and
$\bm{A}_J=\frac{1}{J}\sum_{j=1}^{J}q_j\bm{x}_j\bm{x}_j'.$
Since $\bm{\lambda}_1,\bm{\lambda}_2,\dots,\bm{\lambda}_J$ are i.i.d.\ and bounded, we have $s_J=1+O_p(J^{-1/2})$ and $\|\bm{b}_J\|_2=O_p(J^{-1/2})$. Let $\bm{A}=\mathbb{E}(q_j\bm{x}_j\bm{x}_j')$. Then,
\begin{align*}
\mathbb{E}\|\bm{A}_{J}-\bm{A}\|_F^2 &= \frac{1}{J^2}\sum_{j=1}^{J}\mathbb{E}\|q_j\bm{x}_j\bm{x}_j'-\bm{A}\|_F^2 = \frac{1}{J}\left(\mathbb{E}\|q_j\bm{x}_j\bm{x}_j'\|_F^2-\mathbb{E}\|\bm{A}\|_F^2\right) \\
&\leq \frac{1}{J}\mathbb{E}\|q_j\bm{x}_j\bm{x}_j'\|_F^2 = O(J^{-1}),
\end{align*}
where $\|\cdot\|_F$ denotes the Frobenius norm of a matrix and in the last line we use the boundedness of $q_j$ and $\bm{\lambda}_j$ (and thus the boundedness of $\bm{x}_j$). By Markov's inequality and $\|\cdot\|_2\leq\|\cdot\|_F$, it follows that $\|\bm{A}_J-\bm{A}\|_2=O_p(J^{-1/2})$. Thus, since $\phi_{\min}(\bm{A})$ is bounded away from zero, $\bm{A}_J$ is invertible w.p.a.1 and $\|\bm{A}_J^{-1}\|_2=O_p(1)$. Therefore,
\[
\max_{j\in[J]}|\bm{x}_j'\bm{A}_J^{-1}\bm{b}_J|
\leq\max_{j\in[J]}\|\bm{x}_j\|_2\,\|\bm{A}_J^{-1}\|_2\,\|\bm{b}_J\|_2
=O_p(J^{-1/2})=o_p(1).
\]

Define $d_J=s_J-\bm{b}_J'\bm{A}_J^{-1}\bm{b}_J=1+o_p(1)$, and set
\[
w_j^o=\frac{q_j\bigl(1-\bm{x}_j'\bm{A}_J^{-1}\bm{b}_J\bigr)}{Jd_J},\qquad \text{for } j\in[J].
\]
Define an event
$$\mathcal{E}_J = \{\bm{A}_J \mbox{ is invertible},\, d_J\geq1/2, \mbox{ and } \max_{j\in[J]}|\bm{x}_j'\bm{A}_J^{-1}\bm{b}_J|\leq1/2 \}.$$
It holds that $\mathbb{P}(\mathcal{E}_J)\to 1$. On $\mathcal{E}_J$, every weight is nonnegative, and by definition,
\[
\sum_{j=1}^{J}w_j^o
=\frac{s_J-\bm{b}_J'\bm{A}_J^{-1}\bm{b}_J}{d_J}=1,
\qquad
\sum_{j=1}^{J}w_j^o\bm{x}_j
=\frac{\bm{b}_J-\bm{A}_J\bm{A}_J^{-1}\bm{b}_J}{d_J}=\bm{0}.
\]
Thus $\bm{w}^o\in\Delta_J$ and $\bm{\Lambda}'\bm{w}^o=\bm{\lambda}_0$. Moreover, $q_j\leq C$ implies
$
\|\bm{w}^o\|_\infty\leq\frac{C(1+1/2)}{J/2}=\frac{3C}{J}.
$
\end{proof}

\end{document}